\documentclass[preprint, journal]{vgtc}       

\ifpdf
  \pdfoutput=1\relax                   
  \usepackage{graphicx}                
  \DeclareGraphicsExtensions{.pdf,.png,.jpg,.jpeg} 
\else
  \usepackage{graphicx}                
  \DeclareGraphicsExtensions{.eps}     
\fi%

\graphicspath{{figures/}{pictures/}{images/}{./}} 

\usepackage{microtype}                 
\PassOptionsToPackage{warn}{textcomp}  
\usepackage{textcomp}                  
\usepackage{mathptmx}                  
\usepackage{times}                     
\usepackage{cite}                      
\usepackage{tabu}                      
\usepackage{booktabs}                  
\usepackage{xspace}

\usepackage{bibunits}
\defaultbibliographystyle{abbrv-doi}
\defaultbibliography{references}

\newcommand{\etal}{\textit{et al.}\xspace}

\ieeedoi{10.1109/TVCG.2021.3114773}

\onlineid{0}

\vgtccategory{Research}
\vgtcpapertype{Algorithm/technique}

\title{Simultaneous Matrix Orderings for Graph Collections}

\author{Nathan van Beusekom, Wouter Meulemans, and Bettina Speckmann}
\authorfooter{
 TU Eindhoven, the Netherlands. E-mail: [n.a.c.v.beusekom,w.meulemans,b.speckmann]@tue.nl.
}

\shortauthortitle{Van Beusekom \MakeLowercase{\textit{et al.}}: Simultaneous Matrix Orderings for Graph Collections}

\abstract{%
Undirected graphs are frequently used to model phenomena that deal with interacting objects, such as social networks, brain activity and communication networks. The topology of an undirected graph $G$ can be captured by an adjacency matrix; this matrix in turn can be visualized directly to give insight into the graph structure. Which visual patterns appear in such a matrix visualization crucially depends on the \emph{ordering} of its rows and columns. Formally defining the quality of an ordering and then automatically computing a high-quality ordering are both challenging problems; however, effective heuristics exist and are used in practice.

Often, graphs do not exist in isolation but as part of a collection of graphs on the same set of vertices, for example, brain scans over time or of different people.
To visualize such graph collections, we need a \emph{single} ordering that works well for all matrices \emph{simultaneously}.
The current state-of-the-art solves this problem by taking a (weighted) union over all graphs and applying existing heuristics. However, this union leads to a loss of information, specifically in those parts of the graphs which are different. 
We propose a \emph{collection-aware} approach to avoid this loss of information and apply it to two popular heuristic methods: leaf order and barycenter.

The de-facto standard computational quality metrics for matrix ordering capture only block-diagonal patterns (cliques). Instead, we propose to use \emph{Moran's $I$}, a spatial auto-correlation metric, which captures the full range of established patterns. Moran's $I$ refines previously proposed stress measures. Furthermore, the popular leaf order method heuristically optimizes a similar measure which further supports the use of Moran's $I$ in this context. An ordering that maximizes Moran's $I$ can be computed via solutions to the Traveling Salesperson Problem (TSP); orderings that approximate the optimal ordering can be computed more efficiently, using any of the approximation algorithms for metric TSP.

We evaluated our methods for simultaneous orderings on real-world datasets using Moran's $I$ as the quality metric. Our results show that our collection-aware approach matches or improves performance
compared to the union approach, depending on the similarity of the graphs in the collection. 
Specifically, our Moran's $I$-based collection-aware leaf order implementation consistently outperforms other implementations.
Our collection-aware implementations carry no significant additional computational costs.
}


\CCScatlist{ 
 \CCScat{K.6.1}{Management of Computing and Information Systems}%
{Project and People Management}{Life Cycle};
 \CCScat{K.7.m}{The Computing Profession}{Miscellaneous}{Ethics}
}

\teaser{
  \centering
  \includegraphics{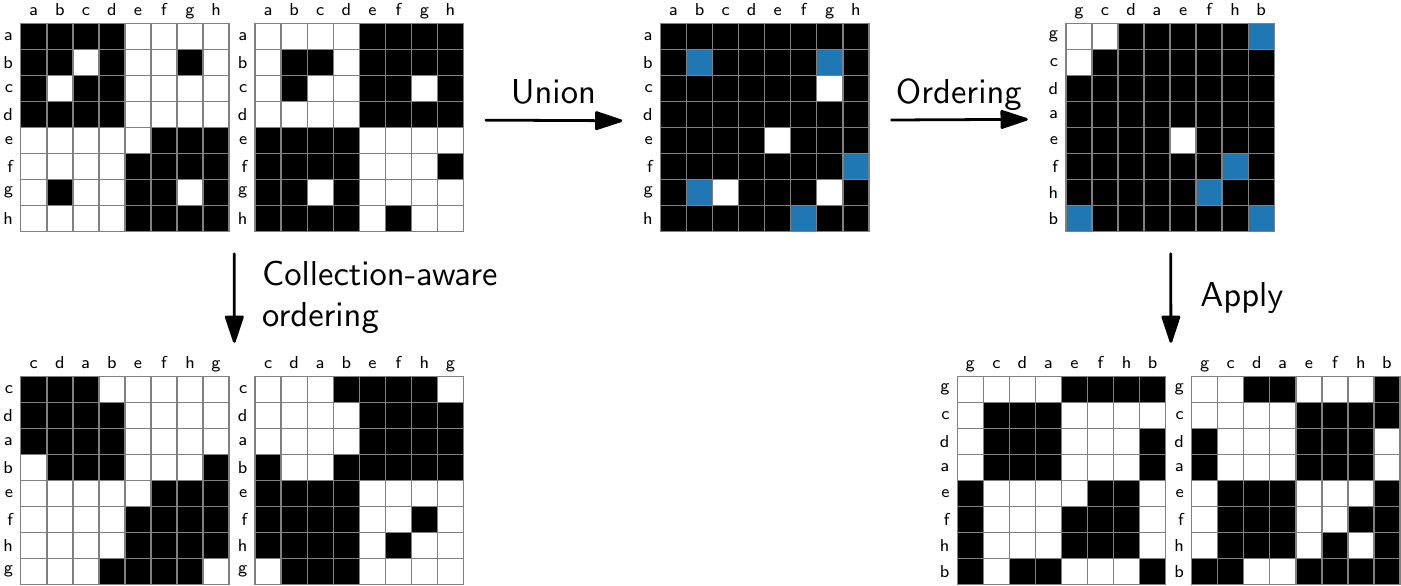}
  \caption{A collection of two matrices (top left). The state-of-the-art first computes a (weighted) union (top middle and right, blue squares have weight 2), then orders the union, and finally applies this ordering to all matrices in the collection. The union leads to a loss of information, specifically, on those parts of the matrices which are different (bottom right). We propose a collection-aware approach to compute orderings which avoids this loss of information (bottom left). Our approach can be applied to existing ordering methods; examples in this figure use the popular leaf order heuristic.}
  \label{fig:teaser}
}

\vgtcinsertpkg

\usepackage{amsthm,amsmath}
\newtheorem{lemma}[]{Lemma}

\usepackage{stfloats}
\renewcommand{\subparagraph}[1]{\smallskip\noindent\textbf{#1}}

\usepackage{pdfpages}

\usepackage{flushend}

\begin{document}
\begin{bibunit}

\firstsection{Introduction}

\maketitle

Graphs are a common tool to model interacting entities, be it humans in social networks, synapses in human brains, or servers in computer networks. Visualizations are a natural tool to explore the structure of graphs and to analyze the underlying interactions. The majority of current graph analysis tools make use of node-link diagrams to visualize graphs. However, matrix visualizations of graphs -- which directly draw the graphs' adjacency matrix -- have been shown to be effective for low-level analysis tasks~\cite{GFC2004} (`are the vertices $x$ and $y$ connected') and for comparisons of (large) graphs~\cite{ABHIF2013}. A matrix visualization can highlight local structures in graphs -- such as clusters, bi-cliques, or stars -- but this relies on a suitable \emph{ordering} of the rows and columns to make these structures manifest as visual patterns. Formally defining the quality of an ordering and then automatically computing a high-quality ordering are challenging problems. A multitude of different techniques have been proposed over the years; see the survey by Behrisch~\etal~\cite{behrisch2016matrix} for an extensive discussion.

Often, graphs do not appear in isolation, but they are rather part of a collection of graphs which share a common set of vertices. Examples include dynamic graphs (such as brain activity of one person over time) and graphs with distinct edge sets (such as brain activity patterns of different persons). In either case, visualizing these graph collections with matrices requires us to compute a \emph{single} ordering that works well for all graphs \emph{simultaneously}. Using one single ordering for all graphs in a collection facilitates comparisons between the matrices; at the same time, the ordering should allow local structures in each matrix to visually manifest themselves.

There are only a few visual analysis tools for graph collections that use matrix visualizations. The \emph{Cubix} system by Bach~\etal~\cite{bach2014cubix} allows the user to explore the evolution of a dynamic network using a matrix visualization for each time step. These matrix visualizations are then combined into a so-called matrix cube, using a single, possibly user-defined, ordering for all matrices. The user can reorder any of the matrices in the cube using a variety of provided ordering algorithms; the new ordering is computed based on the values in this single selected matrix only and then applied to all others. The \emph{MultiPiles} technique by Bach~\etal~\cite{bach2015multipiles} is another approach that facilitates the analysis of time-series of graphs. Here collections of adjacency matrices are ``piled'' together, using temporal clustering or based on user interactions; the piles are then arranged in a small multiples style layout to facilitate comparisons. Among the large set of user interactions is the option to impose a so-called local order on a pile (collection) of matrices: this ordering is computed for all matrices in the pile using the weighted union of all matrices as input for the leaf order heuristic. 
Using a union matrix inevitably loses information about the individual matrices, and hence, the ordering algorithms have less information at their disposal.
Nevertheless, this approach works well if all matrices in the pile are fairly similar. However, as illustrated in the example in Figure~\ref{fig:teaser}, if the graphs contain complementary edges, then the union does not contain enough information to arrive at a good simultaneous ordering.

In this paper we propose a \emph{collection-aware} approach to computing simultaneous matrix orderings for collections of graphs. As most previous work, we focus our attention on undirected graphs, that is, symmetric 0-1 valued adjacency matrices (indicated with 0:white and 1:black squares in all figures). Our basic premise is that a method will produce an ordering which works well for all matrices simultaneously, if the decisions that the method takes are based on the individual matrices -- as far as that is possible. Since any method needs to produce a single ordering in the end, eventually information from all matrices in the collection needs to be aggregated. However, the further down the algorithmic pipeline this aggregation is happening, the more of the local structures in the input matrices are preserved. We demonstrate the soundness of our premise and the feasibility of our approach by describing collection-aware variants of two popular matrix ordering methods: the leaf order~\cite{bar2001fast} and barycenter~\cite{makinen2005barycenter,gansner1993technique,eades1994edge} heuristics. 

The de-facto standard computational quality metrics for matrix visualizations are based on distances induced by the input graph. Specifically, the \emph{bandwidth}, \emph{profile}, and \emph{linear arrangement} measures count how far each edge is removed from the diagonal (using different schemes to aggregate this information into a single value).
By design these measures are optimized by orderings that move matrix cells of edges as close to the diagonal as possible. 
However, many meaningful patterns that occur in matrix visualizations, representing graph features such as bi-cliques or stars (see the survey by Behrisch~\etal~\cite{behrisch2016matrix}), are in fact not close to the diagonal and hence are not captured by these measures. 
The same survey is calling for future research into quantitative measures which evaluate the quality of all patterns.
To overcome the current lack of such quantitative measures, Behrisch~\etal~\cite{behrisch2016magnostics} propose to describe and compare matrices via feature vectors (\emph{Magnostics}) that describe specific patterns. These vectors are well-suited to support database queries and user interactions, but they are less suited for computational benchmarks and algorithmic optimization.

We observe that salient patterns in matrix visualizations are formed by clusters of black and white cells -- a simple form of spatial auto-correlation. We hence propose to use Moran's $I$~\cite{moran1950notes}, a prominent measure for spatial auto-correlation, as a quality metric for matrix visualizations. Moran's $I$ counts vertical and horizontal adjacencies in three classes: black-black, white-white, and black-white (see Figure~\ref{fig:Moran}). Black-black and white-white adjacencies contribute positively to the count, weighted by relative frequency of white or black, while black-white adjacencies contribute negatively. The final score is normalized to lie between $-1$ (checkerboard pattern) and $+1$ (a white or black matrix). 
Moran's $I$ is closely related to several quality metrics which are used internally in matrix re-ordering algorithms (see Section 7.2 in the survey by Behrisch~\etal~\cite{behrisch2018quality} on quality metrics). However, until now, none of these metrics have been used as a quality metric for the final matrix visualization. 

\begin{figure}
    \centering
    \includegraphics{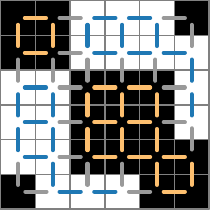}
    \caption{Moran's $I$ counts vertical and horizontal adjacencies in three classes: black-black, white-white, and black-white.}
    \label{fig:Moran}
\end{figure}

\subparagraph{Contributions.} Summarizing the above, the contributions of our paper are twofold. Our primary contribution is the new collection-aware approach hinted at above, with a description of how to apply this approach to the leaf order and barycenter heuristics. 

Our second contribution is the introduction of Moran's $I$ as a quality metric for matrix visualizations. Specifically, we explain how Moran's~$I$ relates to existing measures and algorithms, we show that Moran's~$I$ captures the full range of established patterns and anti-patterns in matrix visualizations, and we explain how to efficiently approximate an optimal matrix ordering for Moran's $I$. 

Throughout the main body of the paper we focus on collections of undirected, unweighted graphs with symmetric adjacency matrices. In Section~\ref{app:moransi} in the supplementary material we consider the more general case of directed graphs. There we also argue that for undirected graphs it is optimal (under Moran's $I$) to use a single ordering for rows and columns. Hence, we restrict our discussion in the main paper to a single ordering. However, our techniques and measures readily generalize to weighted or directed graphs.

We validate the efficacy of our methodology using a short computational experiment using implementations of the collection-aware and union approaches, to showcase that the loss of information indeed occurs in real-world data, and that the collection-aware approach overcomes this problem.

The paper is organized as follows. In Section~\ref{sec:prelims} we briefly summarize notations and definitions. We provide details on measuring quality and distances for matrix orderings and provide our exposition on Moran's~$I$ in Section~\ref{sec:metrics}. We then describe and review the union approach, and our newly proposed collection-aware approach and its implementation in Section~\ref{sec:simulalgorithms}. We describe in Section~\ref{sec:experiments}  the setup and results of our computational experiments, and close in Section~\ref{sec:conclusion} with a review of our findings and possible future work.

\subparagraph{Related work.} 
We are not aware of previous work that computes simultaneous orderings for matrices, beyond the union approach of MultiPiles \cite{bach2015multipiles}; see above, as well as Section~\ref{sec:simulalgorithms} for a discussion of this method. A discussion of related work in terms of ordering quality is deferred to Section~\ref{sec:metrics}. 

In the graph-drawing literature, we find related work on simultaneously drawing graphs in the node-link diagram style \cite{blasius2013}. The goal here typically is to avoid crossings, i.e., draw graphs planarly, on a common vertex set. In other words, can locations for each vertex be found, such that for each graph, all edges of that graph can be drawn using these locations (e.g., with straight lines or few bends per edge) while not introducing intersections. Note that crossings between edges of different graphs are thus allowed. Results here focus on theoretical aspects, establishing computational complexity, i.e., showing that the problem is hard in general \cite{estrella2008} but that other variants admit polynomial-time solutions \cite{cabello2011geometric,fowler2011characterizations,haeupler2013}. Recently, simultaneous embeddings were generalized to graphs drawn in 3D, with the goal that different two-dimensional projections preserve user specified distances~\cite{HHKN2021}. Beck~\etal~\cite{beck2017dynamic} present an extensive overview on visualization techniques for dynamic graphs, which also touches upon matrix visualizations.

We focus on simple black-and-white representations of undirected, unweighted graphs. However, there are many possibilities to augment matrix visualizations. For example, cells can be used to display auxiliary data of the edges, including temporal data \cite{Elmqvist2008ZAME,yi2010timematrix}. Moreover, augmentations with lines can help overcome some of the drawbacks of matrix visualizations, such as the identification of paths \cite{shen2007paths,henryriche2007matlink}.

\section{Preliminaries}
\label{sec:prelims}

\subparagraph{Graphs.}
A graph $G = (V,E)$ consists of a set $V$ of $n$ vertices and a set $E \subseteq V^2$ of $m$ edges.
We assume graphs to be undirected, that is, $(u,v) \in E \Leftrightarrow (v,u) \in E$. Typically, undirected graphs do not contain self-loops $(u,u) \in E$, but we assume that these may occur.

With $N(G,v)$ we denote the neighborhood of vertex $v$ in graph $G$. Specifically, we interpret it as an $n$-dimensional 0-1 vector: an entry is 1 if and only if $(v,u_i) \in E$ where $u_i$ is the $i$th vertex in some arbitrary fixed order of the vertices in $V$.

\subparagraph{Orderings.}
An \emph{ordering} of a graph $G$ is a permutation of its vertices $V$, which we represent as a bijective function on the indices, $\rho \colon \{1,\ldots,n\} \rightarrow  V$.
So, $\rho(1)$ is the first vertex in the ordering, and $\rho^{-1}(v)$ is the rank (position) of $v \in V$ in the ordering.\footnote{Note that the survey by Behrisch~\etal~\cite{behrisch2016matrix} defines the ordering function the other way around, from vertex to index; being a bijection, this is but a notational difference. We found that in our exposition, we rely mostly on the resulting row order so our definition avoids excessive use of the inverse.}
We use $\rho(i,j)$ as a shorthand for the pair $(\rho(i),\rho(j))$, that is, a pair of vertices which may or may not constitute an edge in $E$. 

Given an ordering $\rho$ of $G$, we can create a table with $n$ rows and $n$ columns, where each row and column is associated with a vertex through the ordering. We color each cell $[i,j]$ in row $i$ and column $j$ of this table black if the edge $\rho(i,j)$ is in $E$, and white otherwise.  


\subparagraph{Simultaneous orderings.}
Assume that we are given a set $\mathcal{G} = \{ G_1, \ldots, G_k \}$ of $k$ undirected graphs .
Each graph $G_i = (V,E_i)$ is defined on the same set $V$ of vertices but has its own set of edges $E_i \subseteq V^2$. 
Our goal is to find a \emph{simultaneous} ordering $\rho$ for the set $\mathcal{G}$, that is, an ordering of the vertices $V$ that results in good matrix visualizations for all graphs in the set $\mathcal{G}$.

\begin{figure}[b]
    \centering
    \includegraphics{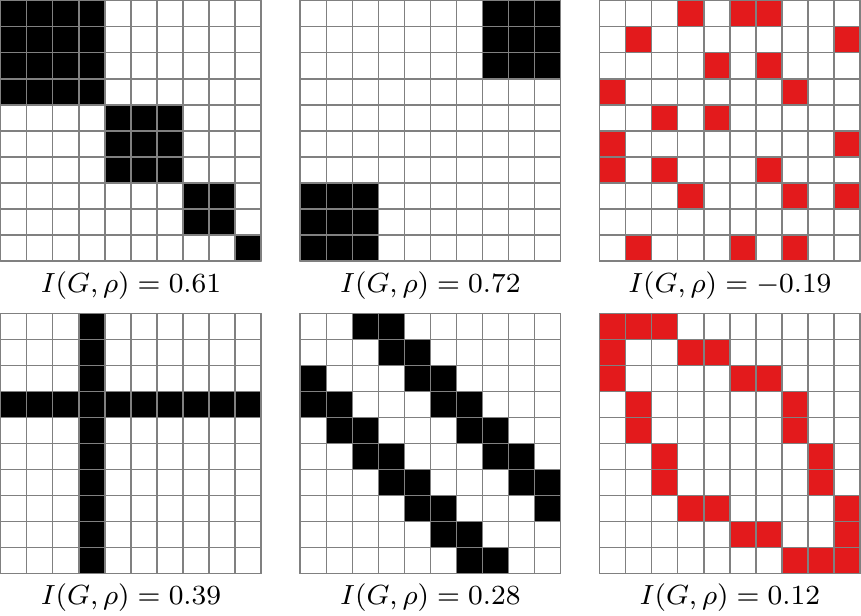}
    \caption{Patterns (black) and anti-patterns (red)~\cite{behrisch2016matrix} with Moran's $I$ score (higher = more correlated). Top row: block pattern, off-diagonal block pattern, noise anti-pattern; bottom row: line/star pattern, bands pattern, bandwidth anti-pattern.}
    \label{fig:patternsandantipatterns}
\end{figure}

\section{Measuring ordering quality}
\label{sec:metrics}

Computational quality measures play an important role when designing and evaluating matrix ordering algorithms. Here measures are used essentially in three different ways: to evaluate the final quality of a matrix visualization, as part of the objective function to be optimized by an algorithm, or simply as distance function between two rows of the matrix. In principle each measure can be used in all three settings, however, we observe that so far, certain measures have been used only in one or two of them. We can roughly separate existing measures into two categories: measures that are based on distances in the ordering between vertices and measures that are based on adjacencies of row (or column) vectors in the matrix.

\subparagraph{Ordering distance.} 
Measures in this category focus on combinatorial, connectivity aspects of the input graph. That is, they measure how well the ordering represents the (edge) connectivity in the graph. The measures attempt to capture the idea that nodes which are adjacent or at least close in the graph should be adjacent or close in the order and vice versa.
The de-facto standard computational quality metrics are all based on the concept of rank difference (distance in the ordering): $\lambda_\rho(u,v) =  |\rho^{-1}(u) - \rho^{-1}(v)|$ should be small for all $(u,v) \in E$.
The function $\lambda$ effectively expresses the deviation of an edge from the diagonal of the matrix. This principle is used to define three common measures~\cite{behrisch2016matrix}:
\begin{description}
\item[bandwidth] which is the maximum deviation, $\max_{(u,v) \in E}{\lambda(u,v)}$
\item[profile] which measures per vertex the maximum deviation to an adjacent vertex earlier in the order, \[\sum_{i=1}^n ( i - \min_{j < i \wedge \rho(j,i)\in E} j ) = \sum_{i=1}^n \max_{j < i \wedge \rho(j,i)\in E} \lambda(\rho(j),\rho(i))\]
\item[linear arrangement] which accounts for all edges, \[\sum_{u \in V} \sum_{(u,v)\in E} \lambda(u,v) = \sum_{(u,v) \in E} \lambda(u,v)\]
\end{description}
All three measures are focused on keeping edges close to the diagonal and are hence optimized by orderings which form block patterns along the diagonal (see Figure~\ref{fig:patternsandantipatterns} top left).
They are used both in optimization functions and, specifically linear arrangement, as a measure for the final quality of a matrix visualization. 

However, several meaningful patterns which correspond to salient structures in the input graph are not related to distance from the diagonal. Furthermore, the bandwidth anti-pattern (which does not match a logical structure in the input graph) is in fact a typical result of optimizing for ordering distance. See Figure~\ref{fig:patternsandantipatterns} for an illustration of the most common patterns according to the survey by Behrisch~\etal\cite{behrisch2016matrix}.
Table~\ref{tab:patternsandantipatterns} lists each of the three ordering distance measures for each pattern and anti-pattern (lower = better). Bandwidth and linear arrangement, for example, 
assess the off-diagonal block pattern as worse than the bandwidth anti-pattern, and profile cannot distinguish between these. In contrast, Moran's $I$ consistently ranks the patterns higher than the anti-patterns (higher = better).

\begin{table}[t]
\caption{Metrics for the patterns and anti-patterns in Figure~\ref{fig:patternsandantipatterns}. For the first three metrics lower is better, while for the Moran's $I$ higher is better.}
\label{tab:patternsandantipatterns}
\small
\begin{tabu} to \linewidth {X[1,l,m]|X[1,r,m]|X[1,r,m]|X[1,r,m]|X[1,r,m]}
        \toprule
        {\bfseries Pattern} &
        {\bfseries Bandwidth} &
        {\bfseries Profile} &
        {\bfseries Linear arrangement} &
        {\bfseries Moran's $I$} \\ 
        \midrule
        Block & 3 & 10 & 30 & 0.61 \\
        \midrule
        Off-diagonal Block & 9 & 24 & 126 & 0.72\\
        \midrule
        Line/Star &  6 & 24 & 54 & 0.39 \\
        \midrule
        Bands & 3 & 23 & 74 & 0.28 \\
        \midrule
        \end{tabu}
       \begin{tabu} to \linewidth {X[1,l,m]|X[1,r,m]|X[1,r,m]|X[1,r,m]|X[1,r,m]}
        \toprule
        {\bfseries Anti-pattern} &
        {\bfseries Bandwidth} &
        {\bfseries Profile} &
        {\bfseries Linear arrangement} &
        {\bfseries Moran's $I$} \\ 
        \midrule
        Noise  & 8 & 28 & 76 & -0.19 \\
        \midrule
        Bandwidth & 4 & 24 & 60 & 0.12 \\
        \bottomrule
\end{tabu}
\end{table}

\subparagraph{Adjacency.}
Measures in this category are usually used to compute distances between two (adjacent) rows of a matrix, based on the directly adjacent cells in the respective rows. 
Two vertically or horizontally adjacent black squares correspond to two edges which share a vertex. Hence, in some sense these measures promote the clustering of the neighborhood of vertices in the graph into adjacent cells of the matrix.
Adjacency measures naturally generalize to the complete matrix and hence are used also as part of optimization functions.
So far, measures in this category do not appear to have been used as a quality measure for the final matrix visualization. 

Moran's $I$, which we will describe in greater detail in the next subsection, is an adjacency measure, and so are the \emph{measure of effectiveness} by McCormick~\etal~\cite{McCormick1969ME,McCormick1972ME} and the stress measure by Niermann~\cite{Niermann2005stress}. All standard distance measures for vectors fall into this category as well, such as the Euclidean distance $L_2$. Note that the squared Euclidean distance is identical to the Manhattan distance for $0-1$ vectors. 

Lenstra and Kan~\cite{LenstraKan1975me} observed that an optimal ordering for the measure of effectiveness is equivalent to a traveling salesperson path, where each matrix row corresponds to a city and the distance between two rows is measured by the number of pairwise vertical black-black adjacencies. In fact, every adjacency measure is optimized globally via a traveling salesperson path. The \emph{Bond Energy Algorithm} by McCormick~\etal~\cite{McCormick1969ME,McCormick1972ME} is in fact a heuristic for TSP using the measure of effectiveness. The popular leaf order heuristic is a heuristic for TSP as well, using the Euclidean distance. 

\subsection{Spatial auto-correlation: Moran's \emph{I}}

We observe that salient patterns are formed by clusters of black cells: moving rows (vertices) close together which have similar neighborhoods.
As such, we postulate that promoting patterns is a form of spatial auto-correlation. Spatial auto-correlation measures are global measures of structure in the data.
The matrix visualization in this context becomes the data for which we try to measure spatial auto-correlation.

We propose to use Moran's $I$ \cite{moran1950notes}, one of the prominent measures for spatial auto-correlation. This very general measure requires values associated with each cell and a weight matrix which captures how the cells are structurally (visually) related. In our case, we associate the values 1 and 0 with each cell, depending on whether the associated edge is in $E$ or not. We design the weight matrix such that each cell is considered adjacent (with weight 1) to the cells with which it shares a border and unadjacent (weight 0) otherwise. Higher scores in Moran's $I$ indicate a stronger correlation and are thus desirable for an ordering.

As described above, we use Rook's adjacency rather than Queen's adjacency\footnote{This refers to the movement capabilities of chess pieces.}. Our motivation is twofold. First, two cells $(i,j)$ and $(i',j')$ correspond to two edges and thus up to four vertices. If each of these has a different value, we get indeed four vertices and the two edges connect two arbitrary pairs. As such, it does not match to a local structure. However, if the two cells are in the same column or row, there are only three vertices, and thus it describes a small pattern of two vertices sharing a common neighbor. Second, Moran's $I$ with Queen's adjacency fails to capture negative spatial auto-correlation on the prototypical chess board of alternating black and white cells.

In our setting, Moran's $I$ effectively simplifies to the following expression (see below for a derivation).
\[ I(G,\rho) = c_B(G) \cdot B(G,\rho) + c_W(G) \cdot W(G,\rho) - 1 \]
Here, $B(G,\rho)$ and  $W(G,\rho)$ refer to the number of black-black and white-white  adjacencies, respectively. See Figure~\ref{fig:examples} for illustrations. The $c$-terms are constants, relying only on the number of vertices and edges in $G$, that weigh the relative impact of these types of adjacencies. Generally speaking, if the matrix has more white cells than black cells, black-black adjacencies have more impact and vice versa.

\begin{figure}[b]
    \centering
    \includegraphics[page=1]{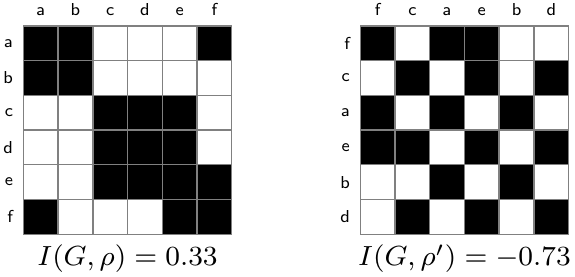}
    \includegraphics[page=2]{lemma-examples}
    \caption{Visualizations of the same graph based on two different orderings, $\rho$ and $\rho'$. The bottom row shows the same visualizations, with annotated horizontal and vertical adjacencies. More such lines mean more same-colored adjacencies and thus a higher quality in terms of Moran's $I$.}
    \label{fig:examples}
\end{figure}

\subparagraph{Derivation.}
Here we present a brief derivation of our simplified form of Moran's $I$, starting from its general form. For a detailed derivation, refer to Section~\ref{app:moransi} of the supplementary material.

A graph $G = (V,E)$ on $n$ vertices for some fixed ordering $\rho$ implies an $n \times n$ 0-1 matrix $M$ where $M_{ij} = 1$ if $\rho(i,j) \in E$ and 0 otherwise. The cells of $M$ are interpreted as the spatial units for the sake of Moran's $I$, and contains all necessary information to derive Moran's $I$. We hence omit dependencies on $G$ and $\rho$ in our derivation here, for sake of notational simplicity. 

We use slightly different notation than perhaps conventional for the general form of Moran's $I$, so as to distinguish between the general form and our simplified form. Specifically, we use $r$ instead of $N$ to denote the number of spatial units (regions) and use a weights matrix~$T$ (topology) with sum $t$ instead of $w$ with sum $W$. Moreover, we refer to the regions using indices $a$ and $b$ rather than $i$ and $j$.

Moran's $I$ is defined over $r$ spatial units, which have associated values $x_a$ for $a \in \{1, \ldots, r\}$. Furthermore, an $r\times r$ matrix $T$ encodes the weights for (typically neighboring) regions: entry $T_{ab}$ is the weight between region $a$ and $b$. We use $t$ to denote the sum over all weights in~$T$. Moreover, let $\overline{x}$ denote the average value $\frac{\sum_{a=1}^r x_a}{r}$. The general form of Moran's $I$ is as follows \cite{moran1950notes}:

\[ I = \frac{r}{t} \cdot \frac{\sum_{a=1}^{r}\sum_{b=1}^{r} T_{ab} (x_a - \overline{x})(x_b - \overline{x})}{\sum_{a=1}^{r} (x_a - \overline{x})^2} \]

As the cells of $M$ correspond to spatial units, we have $r = n^2$.
Let $m$ denote the total number of entries with value 1 in $M$ (the number of edges, with double counting). Hence, $\overline{x} = \frac{m}{n^2}$ and we can rewrite $x_a - \overline{x} = \frac{x_a n^2 - m}{n^2}$. This allows us to simplify the generic form to

\begin{align*} 
I &= \frac{n^2}{t} \cdot \frac{\sum_{a=1}^{r}\sum_{b=1}^{r} T_{ab} (x_a n^2 - m)(x_b n^2 - m)}{\sum_{i=1}^{r} (x_a n^2 - m)^2}.
\end{align*}

With Rook's adjacency, $T_{ab}$ is 1 if $a$ and $b$ are adjacent cells in $M$ and 0 otherwise. We thus need to consider only the terms for which $T_{ab}$ is 1, i.e., for adjacent cells in $M$. Furthermore, since $M$ is a 0-1 matrix, $x_a$ is either 0 or 1: the term $x_a n^2 - m$ is either $n^2 - m$ or $-m$. In the denominator, there are hence two cases to consider (white cells and black cells) and for the numerator there are three cases ($a$ and $b$ describe a white-white, black-black, or black-white adjacency). The contribution to Moran's $I$ is the same per case. 

We thus case simplify this summation by simply counting the number of occurrences of each case, and we identify $B$ (black-black), $W$ (white-white) and $D$ (black-white) with this count. There are $2n(n-1)$ adjacencies, but we count every adjacency twice since $T$ is symmetric: $t = \sum T_{ab} = 4n(n-1)$. 

\begin{align*}
    I &= \frac{n^2}{4n(n-1)} \cdot \frac{ 2 B \cdot (n^2 - m)^2 + 2 W\cdot (-m)^2 + 2 D\cdot (n^2 - m)(-m)}{m (n^2 - m)^2 + (n^2-m)(-m)^2} \\
    &= \frac{1}{2n(n-1) m (n^2 - m)} \cdot ( B \cdot (n^2 - m)^2 + W\cdot m^2 - D\cdot (n^2 - m)m)
\end{align*}

Since $B+W+D = 2n(n-1)$, we see that it suffices to count only $B$ and $W$. This allows us to simplify the expression to a sum of $B$ and $W$, introducing $c_B$ and $c_B$ as the coefficients of $B$ and $W$:

\begin{align*}
    I 
    &= B \cdot \frac{n}{2(n-1) m}  + W \cdot \frac{n}{2 (n-1) (n^2 - m)} - 1 \\
    &= c_B \cdot B + c_W \cdot W - 1
\end{align*}

\subparagraph{Using Moran's $I$.}
With Moran's $I$, we have a measure that is aimed at capturing how well-structured the matrix visualization is. That is, it aims to globally capture patterns, without specifically aiming to specify what a pattern actually constitutes. We indeed see in Table~\ref{tab:patternsandantipatterns} that the patterns score considerably higher than the anti-patterns. As such, this makes the measure more amenable for algorithmic use. 

Indeed, we see that we are effectively counting cells in determining similarity of neighborhoods. 
We denote with $B(G,u,v)$ the number of (vertical) black-black adjacencies we would obtain when $u$ and $v$ are made adjacent in the ordering -- the number of common neighboring vertices in $G$; similarly, $W(G,u,v)$ denotes the white-white adjacencies for $u$ and $v$ -- the number of vertices in $G$ that are neighboring to neither $u$ nor $v$.
Measuring the neighborhood similarity between $u$ and $v$ as $s(G,u,v) = c_B(G) \cdot B(G,u,v) + c_W(G) \cdot W(G,u,v)$, we get the following form of the metric:
\[ I(G,\rho) = -1 + 2 \sum_{i=1}^{n-1} s(G,\rho(i),\rho(i+1)) \] 
Maximizing Moran's $I$ is equivalent to maximizing the sum of~$s$~over adjacent rows.
However, many algorithms such as leaf order are based upon minimizing a sum of distances.
Hence, we define a Moran's $I$ metric $\delta_I(G,u,v) = 1 - s(G,u,v)$, which gives that $I(G,\rho) = n - 2 - 2 \sum_{i=1}^{n-1} \delta_I(G,\rho(i),\rho(i+1))$. Now, maximizing Moran's $I$ corresponds exactly to minimizing the sum of distances.
We do observe that $\delta_I$ is not a proper metric, since identical rows do not have a distance of zero. This identity of indiscernables is inherently incompatible with Moran's $I$ as two fully black rows and two fully white rows should score differently depending on the number of black cells over the entire matrix.
Nonetheless,  the triangle inequality holds: $\delta_I(G,u,w) \leq \delta_I(G,u,v) + \delta_I(G,v,w)$ rewrites to $1 + s(G,u,w) \geq s(G,u,v) + s(G,u,w)$. However, we know that $I(G,\rho)$ cannot exceed one and thus $s(G,u,v) + s(G,u,w) \leq 1$, which proves the claim. 

Measure $\delta_I$ can generally be used with methods that are based on distance measures between rows such as leaf order. It also works well with algorithms that are designed for metric TSP and rely on the triangle inequality in their approximation guarantees.
Furthermore, we believe that this relation between spatial auto-correlation and neighborhood similarities helps to explain why neighborhood measures for computing orderings have been successful and popular in practice.

\subparagraph{Relation to other adjacency measures.} 
The measure of effectiveness by McCormick~\etal~\cite{McCormick1969ME,McCormick1972ME}, when translated to $0-1$ valued (white-black) matrices, simply counts `1' for each vertical or horizontal black-black adjacency, and `0' otherwise. As such it is a less refined form of Moran's $I$.
The stress measure proposed by Niermann~\cite{Niermann2005stress}, which uses Queen's adjacency, is related to Moran's $I$ as well; the final score is the sum of all squared differences between a cell and its neighbors. As already mentioned above, diagonal adjacencies do not capture a graph property, they simply arise from two independent edges. Hence, this stress measure is less suitable to optimize adjacency matrices of graphs.
Finally, the popular leaf order heuristic\footnote{``[leaf order] consistently generates excellent results visually'' Fekete~\cite{fekete2015reorderjs}} uses the Euclidean distance between row vectors. For a 0-1 valued matrix, the Euclidean distance simply counts `1' for each black-white adjacency and then takes the square root of this sum. Orders produced with the leaf order algorithm hence tend to score very well on Moran's $I$.

\begin{figure}[t]
    \centering
    \includegraphics[width=.24\columnwidth]{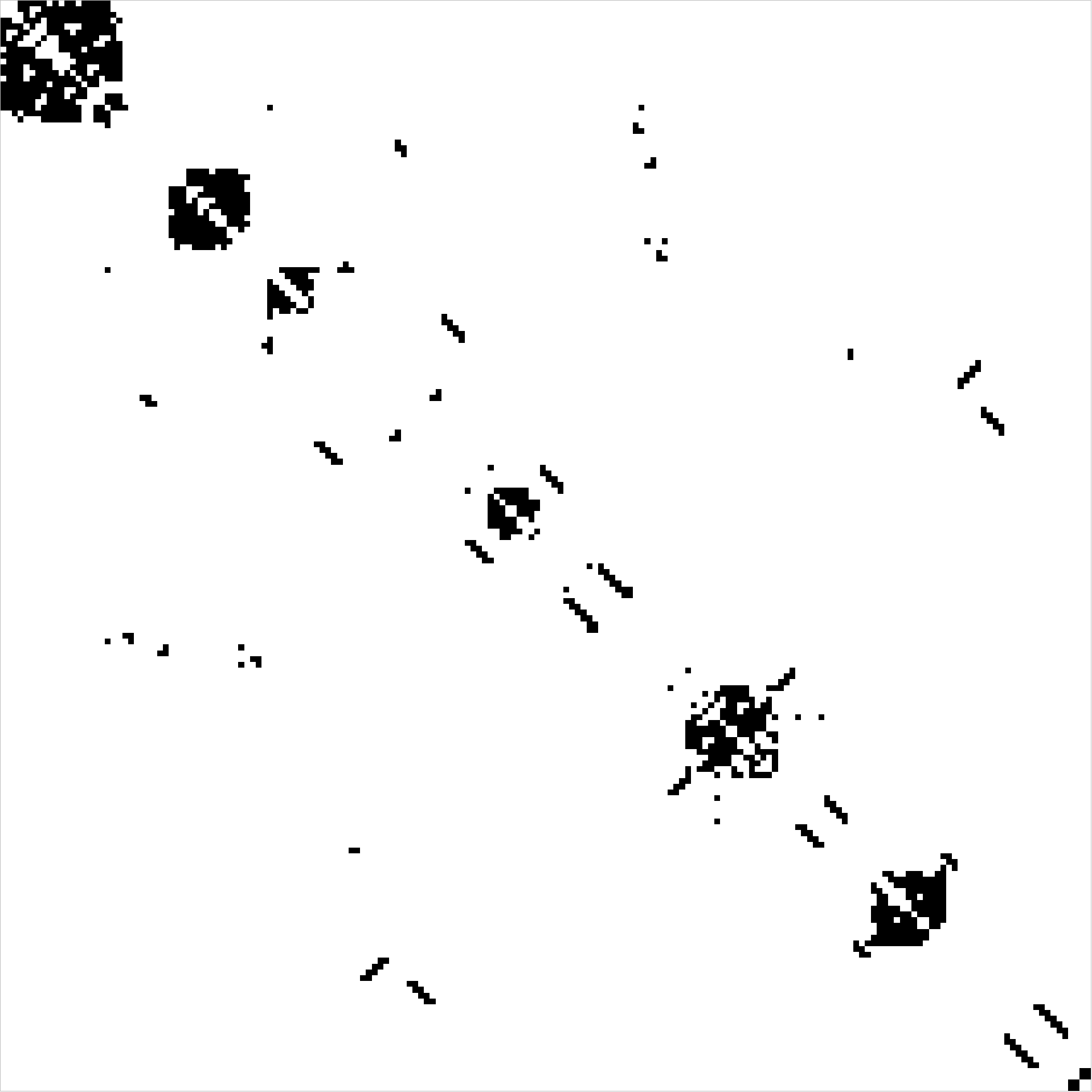}
    \includegraphics[width=.24\columnwidth]{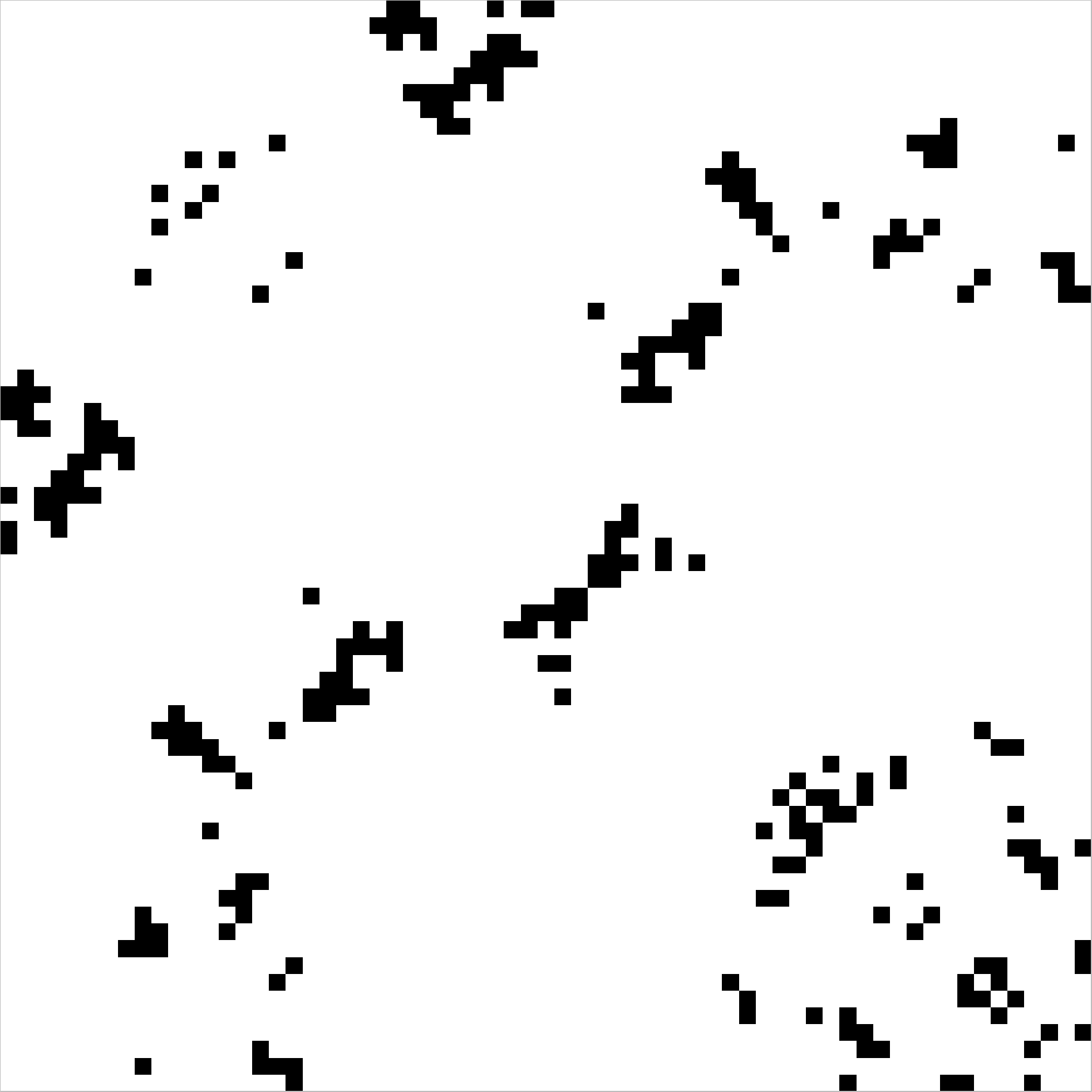}
    \includegraphics[width=.24\columnwidth]{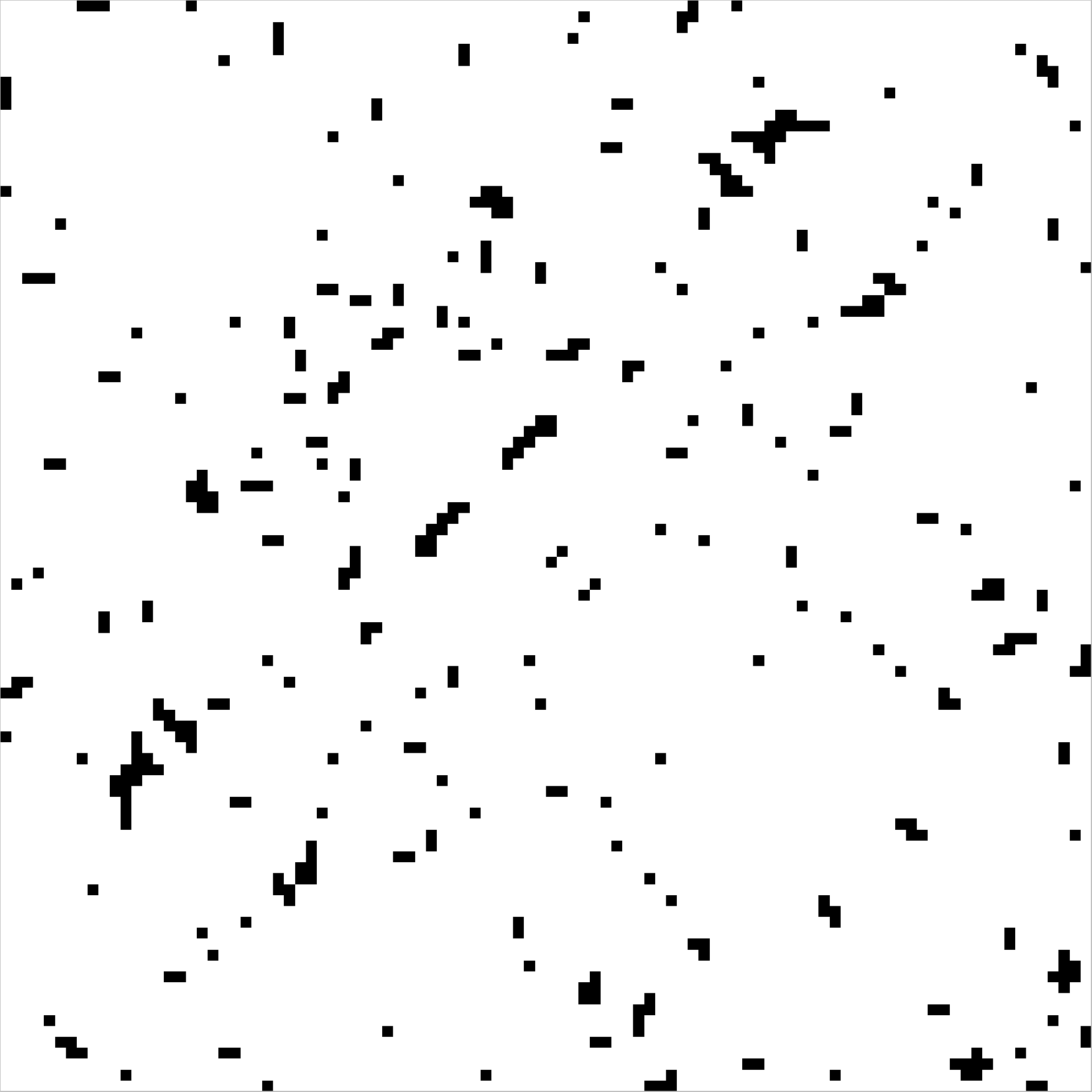}
    \includegraphics[width=.24\columnwidth]{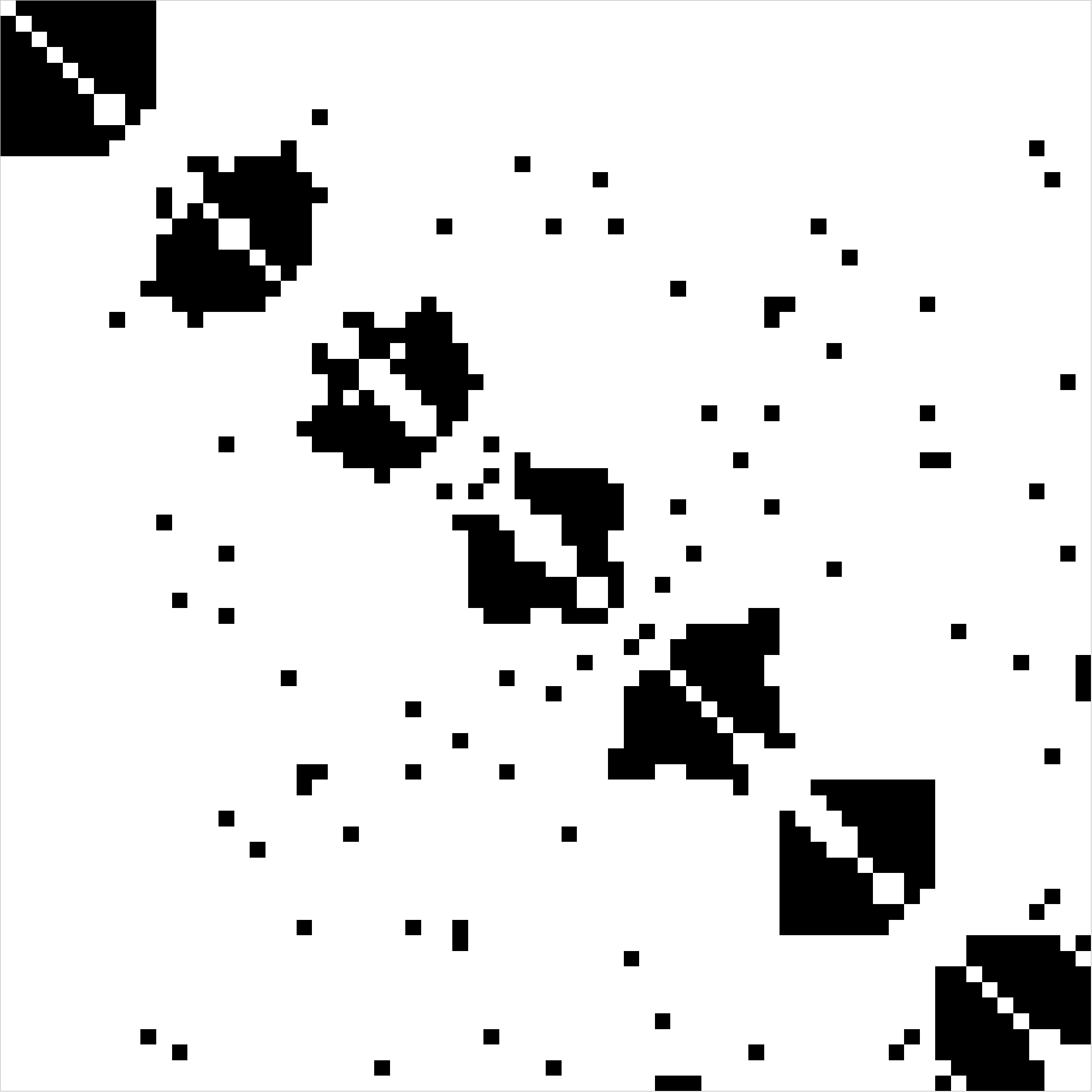}
    
    \noindent\hfill{\footnotesize
    $I(G,\rho)=0.66$
    \hfill $I(G,\rho)=0.33$
    \hfill $I(G,\rho)=0.29$
    \hfill $I(G,\rho)=0.63$ }\hfill\hfill
    
    \medskip
    
    \includegraphics[width=.24\columnwidth]{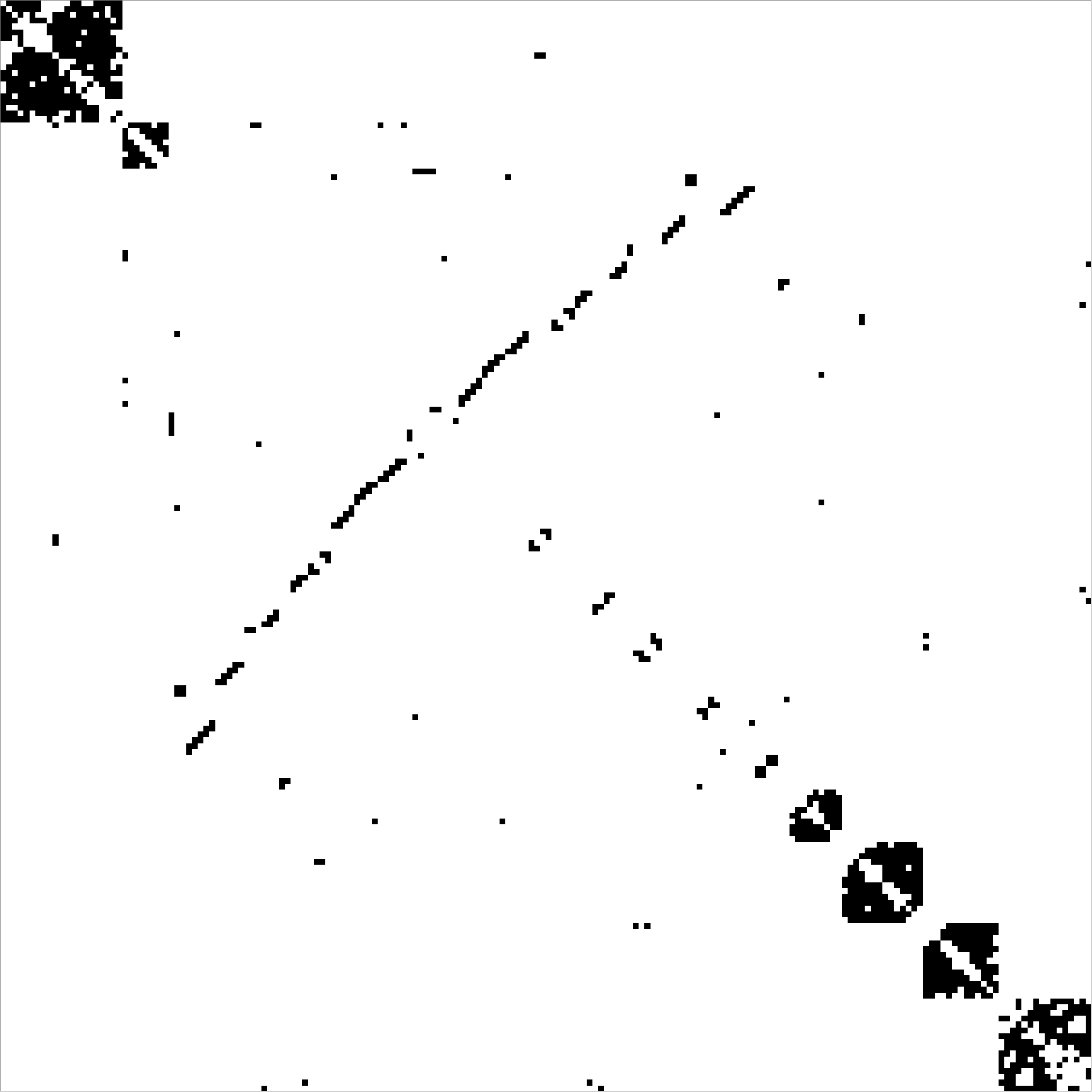}
    \includegraphics[width=.24\columnwidth]{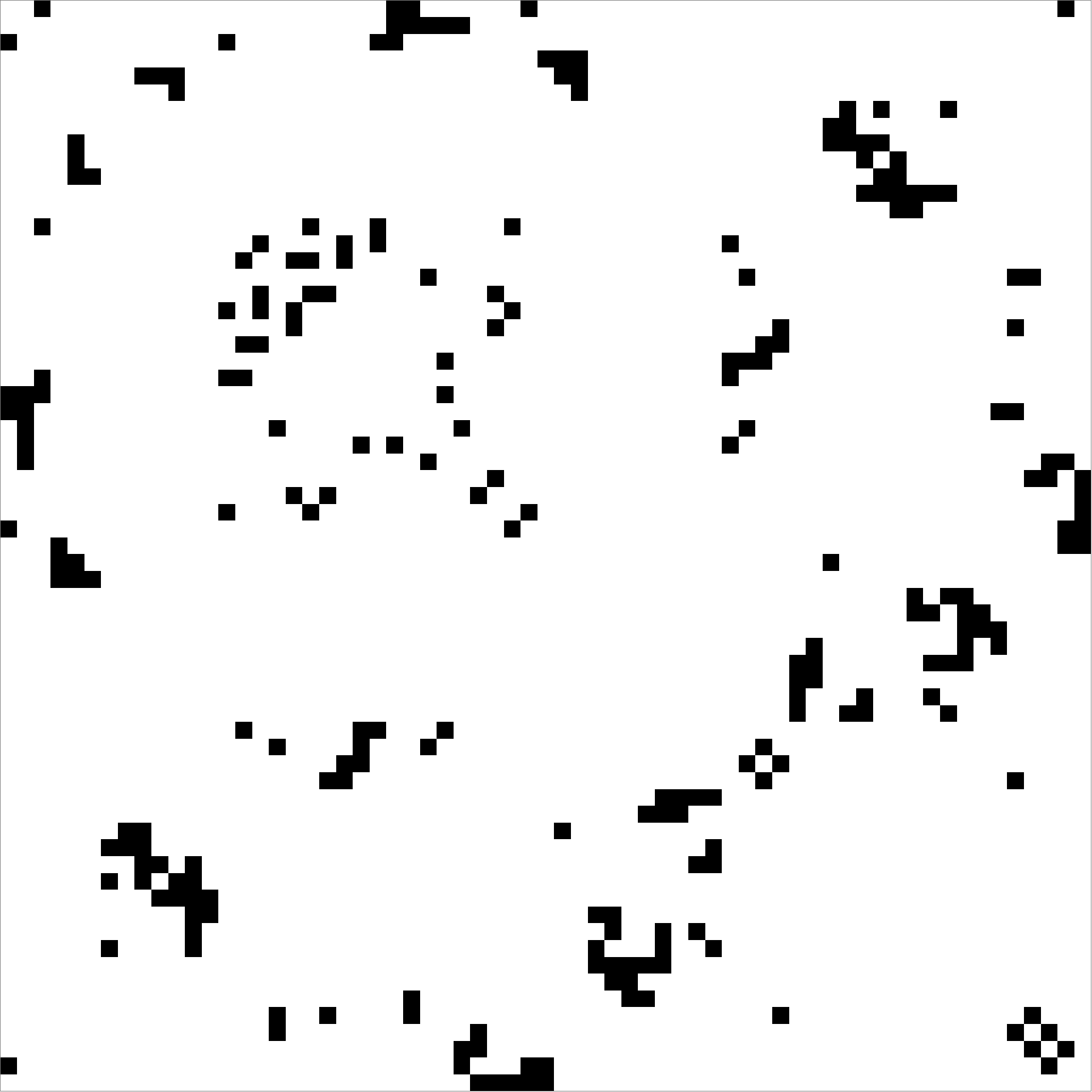}
    \includegraphics[width=.24\columnwidth]{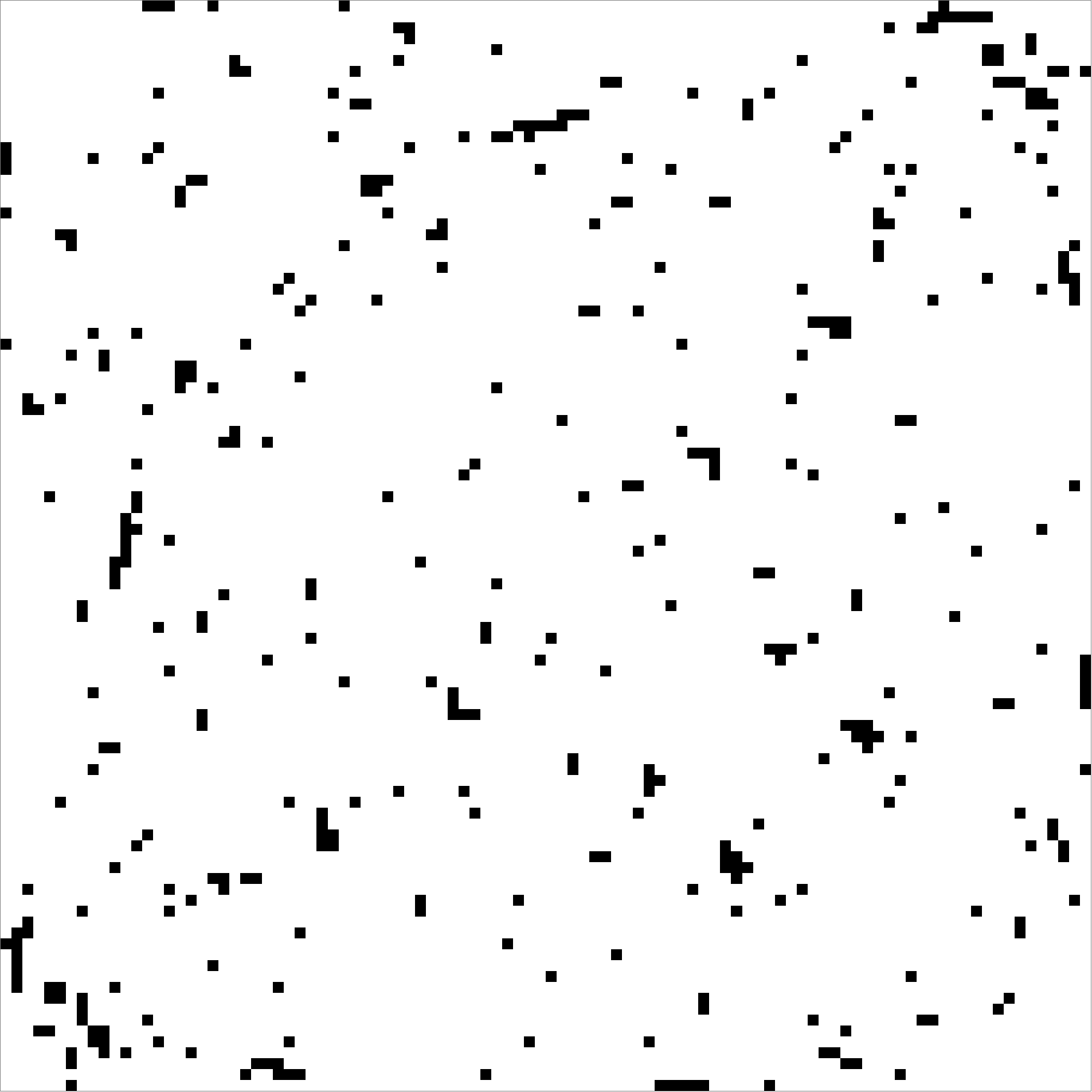}
    \includegraphics[width=.24\columnwidth]{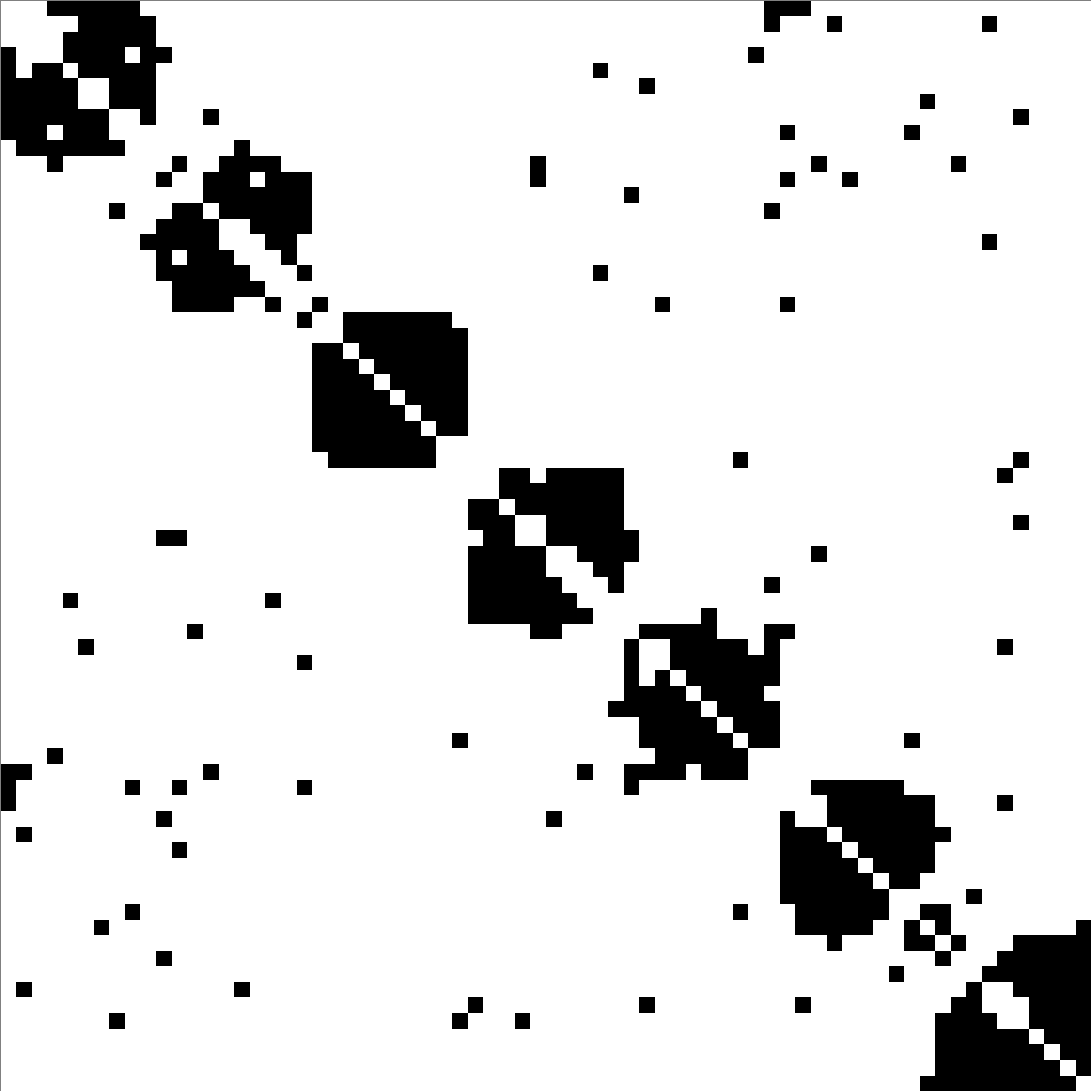}
    
    \noindent\hfill{\footnotesize
    $I(G,\rho)=0.65$
    \hfill $I(G,\rho)=0.30$
    \hfill $I(G,\rho)=0.20$
    \hfill $I(G,\rho)=0.63$ }\hfill\hfill
    
    \caption{Optimizing Moran's $I$ using NN-2OPT (top row) and leaf order (bottom row), matrices from~\cite{behrisch2016matrix}.}
    \label{fig:examples2}
\end{figure}

\subparagraph{Optimizing Moran's \emph{I}.}
As mentioned in Section~\ref{sec:metrics}, every adjacency measure is optimized via a traveling salesperson path. Since TSP is NP-hard to compute, we implemented the Nearest Neighbor (NN) heuristic using $\delta_I$ and further optimized the results using the 2-OPT algorithm~\cite{2opt1958} as long as Moran's $I$ improves by more than 0.0001. Figure~\ref{fig:examples2} shows the results for the four running-example matrices from the survey paper by Behrisch~\etal~\cite{behrisch2016matrix}; each was computed in less than half a second. For comparison, we also include the results using leaf order with the $L_2$ metric. We see that NN-2OPT manages to achieve higher quality in terms of Moran's $I$.


\section{Computing simultaneous orderings}
\label{sec:simulalgorithms}

There are various ways in which a simultaneous ordering can be computed. In the literature, we find two approaches that rely on computing an ordering for a single graph.
A simple method, used by both Cubix \cite{bach2014cubix} and MultiPiles \cite{bach2015multipiles}, to compute an ordering is to base it on a single graph, using any algorithm; this ordering is then simply applied to the other graphs as well. However, we do not consider such an approach to really address the simultaneous ordering problem, as it uses the structure of only a single graph.

Below we review what we refer to as the \emph{union} approach that is suggested by MultiPiles \cite{bach2015multipiles} for the leaf order method and apply it to the barycenter method as well. We then propose our new \emph{collection-aware} approach which overcomes the loss of information that arises in the union approach (see also Figure~\ref{fig:teaser}), and show how to apply it to the leaf order method 
and the barycenter method. As we extend these methods, below is a brief summary of their basic steps (see also the survey \cite{behrisch2016matrix} for details on these methods).

\subparagraph{Leaf order.}
The leaf order method \cite{bar2001fast} computes an ordering in three stages: first, the distance between each pair of vertices is determined, based on some distance measure $\delta(G,u,v)$; second, this information is used to construct a hierarchical clustering on the vertices -- our implementation is built on reorder.js\footnote{\url{https://github.com/jdfekete/reorder.js/}} which uses greedy complete-linkage clustering here; third, the ordering $\rho$ is computed that minimizes $\sum_{i=1}^{n-1} \delta(G,\rho(i),\rho(i+1))$ and ``adheres'' to the clustering tree.
``Adhering'' here means that the ordering matches the order in which an in-order traversal visits the leaves of the tree; effectively, we can choose for each internal node of the tree which of its two children will be the left and the right child. Note that this choice is for every node and thus allows completely reversing the leaves in any subtree. As not all permutations adhere to a given hierarchy, this method avoids the complexity that is inherent in the TSP formulation; indeed, this problem can be solved optimally in an efficient manner \cite{brandes2007optimal}.
Often, this approach is implemented using for $\delta$ some measure of (dis)similarity between $N(G,u)$ and $N(G,v)$.




\subparagraph{Barycenter method.}
The barycenter method \cite{makinen2005barycenter,gansner1993technique,eades1994edge} (see also the survey \cite{behrisch2016matrix}) focuses on optimizing ``crossings'', a standard measure of quality in traditional node-link drawings of graphs. The number of crossings incurred by an ordering $\rho$ of $G=(V,E)$ is defined as follows: we duplicate each vertex $v$ into $v^0$ and $v^1$, and draw $v^i$ at $(\rho(v),i)$; we then draw every edge $(u,v) \in E$ as two line segments, $u^0v^1$ and $v^1v^0$; the number of crossings caused by these line segments (excluding common endpoints) is the number of crossings that $\rho$ incurs.

The basic barycenter method works in two phases. First, starting from an arbitrary ordering, the ordering is repeatedly updated by sorting the vertices according to the median rank of its neighbors. This is done at least for a fixed number of iterations, possibly followed by additional iterations until the number of crossings no longer decreases (``convergence''). Second, a postprocessing step is applied which tries to swap adjacent vertices in the ordering, while such a swap reduces the number of crossings further.

\subsection{The union approach}

With MultiPiles, Bach~\etal \cite[Section 4.7]{bach2015multipiles} suggest an approach to computing a simultaneous ordering: ``MultiPiles can calculate a global ordering which tries to  find a topological clustering across all matrices''. The paper itself provides no further detail on this method, but as their code is open source\footnote{\url{https://github.com/benjbach/multipiles}, accessed Feb. 2021.}, we were able to extract the exact algorithm. What follows is first our general interpretation of their approach, followed by their exact implementation on the leaf order heuristic.


The \emph{union} approach takes the (weighted) union over all graphs in $\mathcal{G}$ to arrive at a graph $H$ on the same vertex set $V$, where each edge has a weight corresponding to the number of graphs of $\mathcal{G}$ it occurs in.
Effectively, it is the sum over all 0-1 adjacency matrices. 
We may apply any ordering algorithm to $H$ and apply the resulting ordering to all graphs in $\mathcal{G}$. 
For a vertex $v$, $N(H,v) = \sum_{G\in\mathcal{G}} N(G,v)$ is now a vector of length $n$ where each entry is an integer in $\{0,\ldots,k\}$, reflecting the weight of the associated edge.

The drawback of the union approach is a potential information loss, as $H$ does not store information on which of the underlying graphs the edges actually occur in. In the extreme case, $H$ could be a complete graph with unit weights for all edges, even though the underlying graphs have structures; imagine for example, two cliques in one graph and a biclique connecting the vertices of the two cliques in a second graph. In this case, all neighborhoods are identical and as such, there is no information left in $H$ to inform a suitable simultaneous order. This is in fact our example in Figure~\ref{fig:teaser}, though we introduced some ``noise'' in this figure. 
The advantage is its simplicity: techniques that inherently work on (or straightforwardly generalize to) weighted graphs can readily be applied.
It could in principle also be applied to techniques for unweighted graphs, but this only exacerbates the information loss.

\subparagraph{Union leaf order.}
MultiPiles \cite{bach2015multipiles} applies the union approach to the leaf order method. This algorithm readily works on weighted graphs, and all it needs is some choice of measure $\delta$. The MultiPiles implementation uses $\delta(u,v) = L_2(N(H,u),N(H,v)) = L_2(\sum_{G \in \mathcal{G}} N(G,u), \sum_{G \in \mathcal{G}} N(G,v))$, that is, the Euclidean distance between neighborhoods. In general, we may of course use other measures instead of $L_2$ in the same fashion. We observe that squaring has no effect on comparisons of distances and thus does not affect the result of the complete-linkage clustering; it may however alter the eventual ordering, since comparisons between sums of distances may change.

But we may also use, for example, a Moran's $I$-based metric. We shall refer to it in this context also as $\delta_I$, but observe that it is not simply counting the number of black-black and white-white adjacencies. Instead, we compute $\delta_I(u,v) = -\sum_{x \in V} (w(u,x) - \overline{w})(w(v,x) - \overline{w})$ where $w(a,b)$ is the weight of edge $(a,b)$ in $H$, i.e., the number of graphs in $\mathcal{G}$ it occurs in, and $\overline{w}$ is the average weight over all pairs. Effectively, this is the original Moran's $I$ formula where we forego the normalization terms as they are all the same for a given matrix; we multiply this by $-1$ as to obtain a distance function. As the implementation relies only on comparisons and sums, the negative values do not cause issues in this method -- though we could factor in the normalization and add 1 as we did for $\delta_I$ in Section~\ref{sec:metrics} to obtain a nonnegative measure that satisfies triangle inequality.

We observe the loss of information here as follows. Whereas indeed similar neighborhoods across the graphs lead to similar sums, the converse is not true: similar sums do not imply similar neighborhoods across the graphs. 
As such, the method may promote putting two vertices adjacent in the ordering that look similar in $H$ but are actually not similar in any of the individual graphs.


\subparagraph{Union barycenter.}
The barycenter method is typically applied in an unweighted setting. However, the crossing measure readily generalizes to a weighted setting. 
That is, rather than counting the number of crossings, we multiply the weights of two intersecting edges and sum the result of all intersecting pairs.
We refer to this as the union barycenter. 

We observe the information loss here as follows. Whereas crossings in one of the graphs contribute to the cost of an ordering, a counted crossing need not actually be present in any of the graphs in the collection. That is, two edges may cross in the union, but not occur together in any graph, and as such incorrectly contribute to the cost.

Though the second stage is readily affected by this change in measure, we note that it does not truly affect the first stage, beyond testing convergence: the ordering by median ranks of neighbors remains the same. We could apply a weighted median instead, or revert to using the actual (weighted) barycenters. As such, if this is pursued, it may be worthwhile to revisit this idea and in fact compute with the actual weighted barycenter. 
However, such extension is beyond the scope of our work here, as we focus on avoiding the loss of information that the union approach incurs.

\subsection{The collection-aware approach}

We propose a new approach to computing simultaneous orderings, one that is \emph{collection-aware}. That is, we push the graph collection actually into the algorithms that compute orderings as much as possible. 
Specifically, when an algorithm makes decisions about the quality of an ordering, or how to modify an ordering based on neighborhoods, we now base this information on all graphs, rather than some aggregated form. In this manner, we may prevent the information loss that the union approach incurs. How this is done specifically, naturally depends on the algorithm under consideration. 

\subparagraph{Collection-aware leaf order.}
As all information is captured in the pairwise distances between vertices, it suffices here to ensure that $\delta$ is collection-aware. That is, $\delta(u,v)$ should be low if $u$ and $v$ have similar neighborhoods in many graphs, \emph{and vice versa}.
Whereas the union approach gives the former, it does not succeed in the latter.

Rather than $\delta(u,v) = L_2(\sum_{G \in \mathcal{G}} N(G,u), \sum_{G \in \mathcal{G}} N(G,v))$, a collection-aware interpretation would apply these operators the other way around: $\delta(u,v) = \sum_{G \in \mathcal{G}} L_2(N(G,u), N(G,v))$. 

Again, we observe that we may replace $L_2$ with other distance measures as well. Specifically, observe that $L_2^2$, using the squared Euclidean distance, is different from using $L_2$ now also in the clustering step. Also, we can again use the Moran's $I$-based metric $\delta_I$ where we can now rely on the metric formulation as presented in Section~\ref{sec:metrics}.

Note that this provides no computational overhead in asymptotic terms. As adding vectors is typically slightly faster than computing distances, we may expect a slight overhead. Yet, this is likely overshadowed by the clustering and leaf order stages of the algorithm and as such barely noticeable.


\subparagraph{Collection-aware barycenter.}
In our collection-aware barycenter implementation, we go back to counting crossings (as we assume unweighted graphs in $\mathcal{G}$), but we now do so for each graph separately and then take the sum over all graphs. This makes the second stage collection-aware.

However, as before, changing the way we measure crossings does not readily influence the procedure in the first stage, beyond testing convergence, as this is based purely on the ranks of the neighbors. In a collection-aware approach, however, we should strive to base these decisions on ordering on all graphs. Rather than using all neighbors in all graphs to determine the target rank of a vertex, we determine the target rank per graph, and aggregate this information. We implement this as follows. 
For every vertex, we compute the median rank of its neighbors in each  graph in $\mathcal{G}$ separately to arrive at a set of median ranks. Subsequently, we compute the median of these median ranks. In the event that a vertex has no neighbors at all in a graph, it is omitted from the set of median ranks. Sorting then proceeds as before, but is now based on the median of medians instead.

\section{Experiments}
\label{sec:experiments}

Here we present the results of a brief experimental evaluation that aims to quantify how much improvement our collection-aware approach provides compared to the union approach, via various implementations. Moreover, we also investigate the use of Moran's $I$ as a distance measure for optimization.

\subsection{Setup}

\subparagraph{Algorithms.}
We have two approaches to solving the simultaneous ordering problem: \emph{(U)} union and \emph{(C)} collection-aware.
We have shown how to implement these approaches on two base algorithms: \emph{(LO)} leaf order 
 and \emph{(BC)} barycenter.
Finally, LO 
can work with various distance metrics. In particular, we use Euclidean distance $L_2$ and the Moran's $I$-based metric $\delta_I$, as well as their squared variants $L_2^2$ and $\delta_I^2$. We name each algorithm using a concatenation of these abbreviations.
For example, U-LO-$L_2$ refers to the union leaf order method using $L_2$, effectively the implementation suggested by MultiPiles; C-LO-$\delta_I^2$ refers to the collection-aware leaf order method using $\delta_I^2$ as a metric; C-BC refers to the collection-aware barycenter method. Our implementation of the ten resulting algorithms is openly available on GitHub\footnote{\url{https://github.com/nvbeusekom/reorder.js}}.

\subparagraph{Data.}
We test our algorithms with three datasets, chosen to obtain a variety of characteristics, ranging from many graphs with few vertices, to few graphs with many vertices;  Table~\ref{tab:datasets} provides some basic statistics.
\begin{description}
\item[FLT] The ``flashtap'' data that was used for MultiPiles \cite{bach2015multipiles}\footnote{\url{https://aviz.fr/~bbach/multipiles/}, accessed March 2021.}. It represents functional brain connectivity in a Parkinson's disease study.
\item[SCH] Social interaction between children and teachers at a primary school \cite{gemmetto2014mitigation,stehle2011high}\footnote{\url{http://www.sociopatterns.org/datasets/primary-school-temporal-network-data/}, accessed March 2021.}.
\item[VIS] Publications data in the InfoVis conference \cite{Isenberg:2017:VMC}\footnote{\url{https://sites.google.com/site/vispubdata/home}, accessed March 2021 at version 9.02.}. We construct a graph on the authors for each year from 2015 to 2020, where an edge between two authors is included if they had a joint paper in this period. We include only authors with publications in at least three years and their co-authors.
\end{description}
The density $m/n^2$ of these graphs varies (Table~\ref{tab:datasets}).
Arguably, graphs with low density such as the VIS co-authorship network are better visualized using other visual idioms, for instance, node-link diagrams or hybrid visualizations, but we include them here to investigate our algorithms under diverse circumstances.
For each dataset we also measured the change $\Delta$ between graphs, which is the number of cells in the matrix that change their value. Though expressed as a fraction of the entire matrix (where FLT changes most, and VIS the least), these values can also be interpreted with respect to the edge density (in which case VIS changes most, relatively speaking, and FLT the least).

\begin{table}[t]
    \centering
    \caption{Overview of datasets. $k$: number of graphs in collection; $n$: number of vertices in each graph; $m$: number of black cells in the matrix; $\Delta$: number of changing edges between graphs. $m$ and $\Delta$ are given as a percentage of the number of cells in the matrix ($n^2$), and we provide the mean $\mu$ and standard deviation $\sigma$.}
    \label{tab:datasets}
    \begin{tabu} to \linewidth {X[1.3,l,m]|X[1,r,m]|X[1,r,m]|X[1,r,m]X[1,r,m]|X[1,r,m]X[1,r,m]}
    \toprule
           {\bfseries Dataset} & \multicolumn{1}{c|}{$k$} & \multicolumn{1}{c|}{$n$} & \multicolumn{2}{c|}{$m/n^2$ (\%)} & \multicolumn{2}{c}{$\Delta/n^2$ (\%)} \\
            & & & \multicolumn{1}{c}{$\mu$} & \multicolumn{1}{c|}{$\sigma$} & \multicolumn{1}{c}{$\mu$} & \multicolumn{1}{c}{$\sigma$} \\
         \midrule
FLT & 96 & 29 & 44.06 & 11.31 & 19.67 & 20.61 \\
SCH & 17 & 242 & 5.15 & 1.48 & 3.31 & 3.68 \\
VIS & 6 & 536 & 0.28 & 0.10 & 0.21 & 0.26 \\
\bottomrule
    \end{tabu}
\end{table}

\subparagraph{Quality measures.}
We measure the quality of the resulting matrix orderings via the current standard of linear arrangement as well as with Moran's $I$. Whereas Moran's $I$ is normalized to $[-1,1]$ by definition, linear arrangement is not. We normalize a linear arrangement value $a$ to $1 - a/M$, where $M$ is the maximum value over all algorithms and graphs of the same dataset. Hence, the \emph{normalized linear arrangement} is a value in $[0,1]$, where 1 is the best performing matrix, such that in both measures, higher values correspond to higher quality.

Note that we measure our metrics for each graph in the collection separately. As such, we obtain a distribution of ordering quality, for each algorithm-dataset combination. We focus on the minimum, average and median quality. The latter two simply because we aim for high overall quality. The minimum, however, is also specifically useful since it gives us an idea of the worst matrix in the collection. Ideally, we would want to avoid slightly improving many graphs at the expense of greatly reducing the quality of one graph.

\begin{figure}[b]
    \centering
    \includegraphics{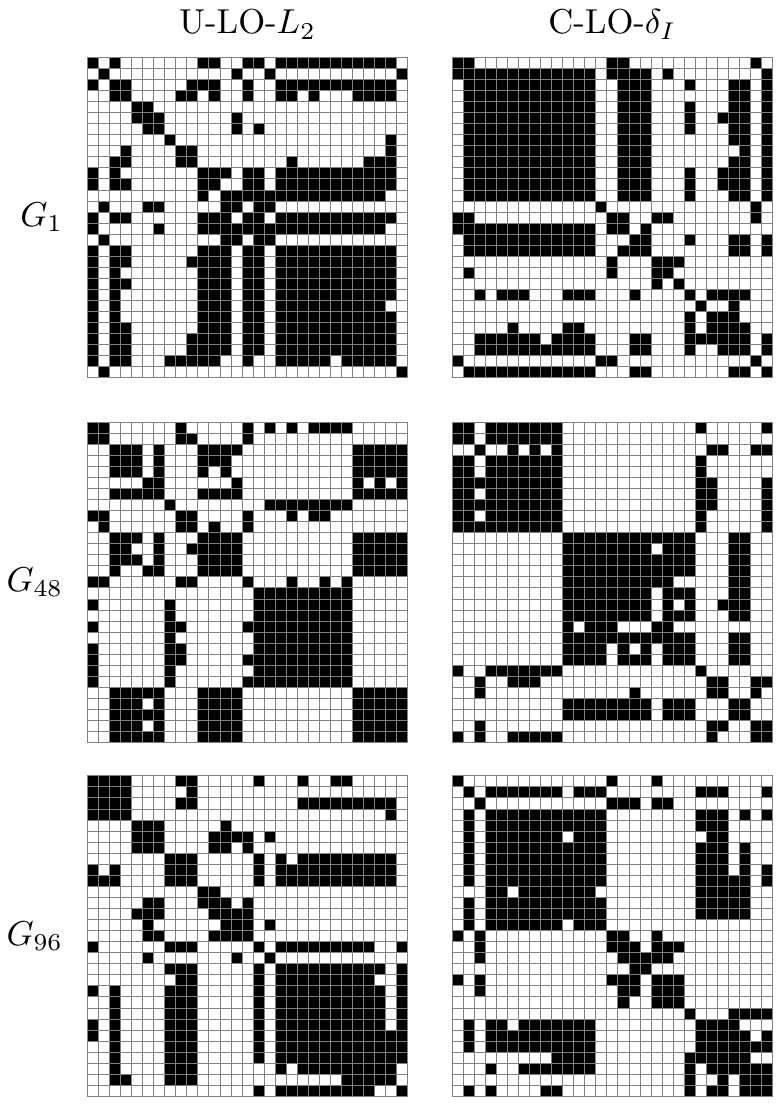}
    \caption{A comparison of the state-of-the-art (U-LO-$L_2$) against our main contribution (C-LO-$\delta_I$) on timesteps $G_1$, $G_{48}$ and $G_{96}$ of the FLT dataset.}
    \label{fig:flt-snippet}
\end{figure}

\begin{figure*}[t]
\centering
\includegraphics[page=2]{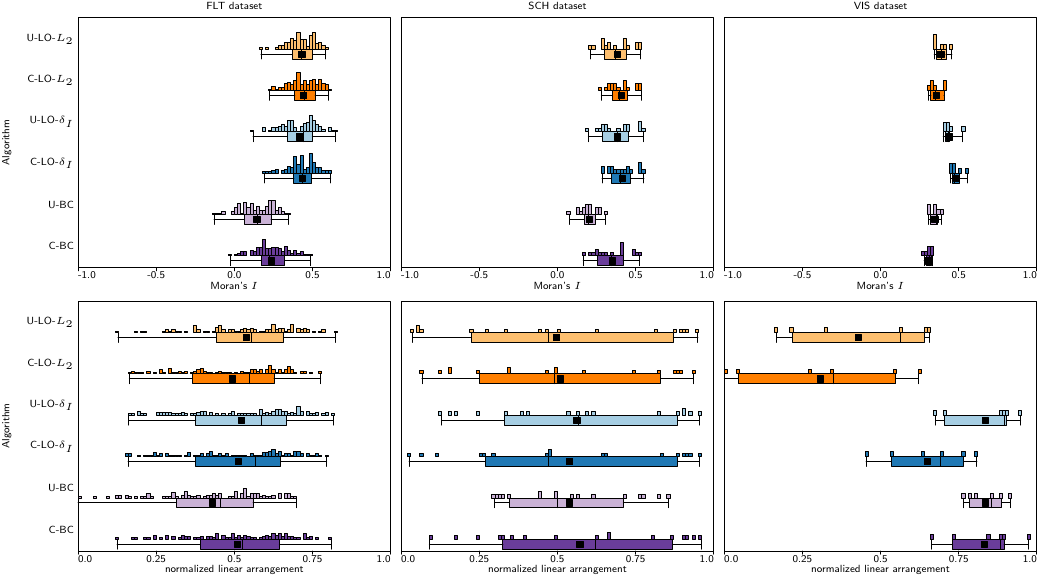}
\caption{Performance of the six main algorithms considered for each of the datasets (columns) in terms of Moran's $I$ (top row) and Linear Arrangement (bottom row). For both, higher values indicate higher quality.}
\label{fig:charts}
\end{figure*}

\subsection{Results}

Figure~\ref{fig:charts-full} in Section~\ref{app:figures} of the supplementary material shows the results for each algorithm on each dataset in terms on the two measures. Our primary observation here is that using squared versions of metrics changes the results, but does not seem to considerably change the overall performance (in terms of means, medians and minima). Refer to the tables in Section~\ref{app:figures} of the the supplementary material for precise differences; averaged over all datasets the differences do not exceed $0.02$ for Moran's $I$. A notable outlier is the differences in the VIS dataset between using $L_2^2$ and $L_2$ for collection-aware leaf order, though we do not investigate this further. 
 
Here, we focus on the six algorithms that do not use the squared distance functions. In Section~\ref{app:figures} of the supplementary material, Figures~\ref{fig:flt-results-1-7} to \ref{fig:flt-results-92-96}, we show all matrix visualizations obtained through the six ordering algorithms for the FLT dataset.
Figure~\ref{fig:flt-snippet} shows an excerpt of the FLT dataset for U-LO-$L_2$ (as proposed in \cite{bach2015multipiles}) and C-LO-$\delta_I$ (our proposed algorithm). Clearly, each of the matrices in isolation can be improved, but each column uses but one ordering by design, which must work well for all 96 graphs in the collection. We may observe that block patterns feature more strongly in the C-LO-$\delta_I$ result; the U-LO-$L_2$ result partitions these into smaller block patterns complemented by matching off-diagonal blocks. 

A summary of the results of the six algorithms are shown in Figure~\ref{fig:charts}. Below, we discuss various comparisons based on Moran's $I$, before briefly discussing linear arrangement.



\subparagraph{(U/C)-LO-$L_2$.} 
Let us first compare the performance of the union and collection-aware implementations of leaf order, using the standard Euclidean distance. Generally, we see very similar performance between these methods. For FLT and SCH, the median and average improve slightly in the collection-aware variant (difference of $0.02$ and $0.03$ for median, and $0.01$ and $0.02$ for average); the minimum improves more, $0.06$ (FLT) and $0.07$ (SCH).

For VIS we actually see a decrease, roughly $-0.03$ for each statistic. We attribute this to the overall sparsity of the graphs. The Euclidean distance may weigh white-white adjacencies too heavily, whereas they contribute little to improving Moran's $I$. Using a union gives for a slightly denser graph and thus may help in finding some structure in this specific case.

In this comparison, we conclude that there indeed is some information loss in the union approach. The medium-density graph improves most, hinting that very dense graphs and very sparse graphs suffer less from the information loss that complementary patterns may give -- though this may benefit more structural investigation using more (controlled) datasets.

\subparagraph{(U/C)-LO-$\delta_I$.} 
Let us now repeat the above comparison, but with algorithms based on the Moran's $I$ metric $\delta_I$. We now see an improvement for each dataset. On average, the minimum Moran's $I$ improves by $0.07$ and the median and average by $0.02$ and $0.03$ respectively. Most different from the previous comparison is the improvement for VIS. From this we conclude that the collection-aware approach actually does help to make find structures in very sparse graphs. That is, the decrease in performance can be attributed to the use of the Euclidean distance rather than the collection-aware approach itself.

\subparagraph{C-LO-($L_2$/$\delta_I$).} 
From the above, we may conclude that our collection-aware adaptations yield benefits in terms of Moran's $I$. So, we now briefly consider the choice of metrics between the collection-aware implementations. We see that on average (in terms of minimum, average and median) using $\delta_I$ improves Moran's $I$ by $0.04$ compared to using $L_2$. Some improvement was to be expected, since the method now (heuristically) optimizes the quality measure directly. Yet, the effect is somewhat small in that light. 

Primarily, we can conclude that the Euclidean distance serves well as a proxy. We see that this proxy suffers in case of very sparse graphs (VIS), where using $\delta_I$ improves considerably more (differences of $0.12$ to $0.14$) than in the other datasets. The equal treatment of types of adjacencies in the Euclidean distance does seem to hinder it, in obtaining high Moran's $I$.

\subparagraph{(U/C)-BC.} 
Let us now briefly turn to the barycenter method. We see here again a strong difference between datasets. Whereas C-BC clearly outperforms U-BC (differences of $0.08$ to $0.15$) for FLT and SCH, the opposite it true for VIS. Again, we can attribute this to the sparsity of the graph. As many vertices have edges in only one graph, their ranks in the collection-aware method are based primarily on that one graph -- which may interfere with other graphs. With the union approach, at least everything is treated centrally and thus the information in the slightly denser union graph is indeed a bit more rich.

\subparagraph{Linear arrangement.}
Though it focuses on block patterns only and does not generally measure the degree of structure of the visualization, linear arrangement is the de-facto standard for automatically measuring ordering quality \cite{behrisch2016matrix}. We briefly consider it for this reason. 

For FLT, we may observe that the U-LO variants outperform C-LO approaches, though the performance between C-LO variants is mostly similar. For the barycenter method, this is rather the other way around: C-BC outperforms U-BC.

For SCH, we see more spread of the linear arrangement within one graph collection, but this may be an effect of the normalization. Though U-LO-$\delta_I$ still slightly outperforms C-LO-$\delta_I$, we now see that C-LO-$L_2$ offers a slight improvement on U-LO-$L_2$. C-BC still outperforms U-BC in terms of median and average, though its minimum is now considerably lower.

For VIS, we again see that the U-LO variants outperform C-LO approaches, now with a considerable difference. Interestingly, C-BC is now slightly under U-BC as well. We may attribute this again to the sparsity of the dataset.

We conclude that, for linear arrangement, the union approach is typically better than the collection-aware approach to leaf order. For the barycenter method, however, this appears to be opposite, unless the dataset is excessively sparse. This is, perhaps, not surprising since the crossing measure underlying the barycenter methods discourages ``long edges'' and thus off-diagonal cells. With the collection-aware approach, we see that this indeed avoids information loss, unless there is very little information to begin with -- in which case the union method strengthens the signal of the little information that is available.

\subparagraph{Overall comparison.}
Considering the above, we may conclude that our collection-aware adaptations have been successful in improving Moran's $I$ and thereby the structure of the resulting matrix visualizations. We observe that C-LO-$\delta_I$ performs best over all methods, with higher minimum, median and average scores compared to all other methods. It specifically performs better on the minimum compared to other LO variants, and considerably improves upon BC variants.
Whereas the barycenter method focuses (indirectly) more on block structures and thus linear arrangement, we do not see it consistently outperforming the leaf order methods even in this criterion.
We would thus recommend the use of C-LO-$\delta_I$ as the most versatile algorithm for solving the simultaneous ordering problem.

\section{Conclusion}
\label{sec:conclusion}

In this paper we considered the problem of computing simultaneous orderings for graph collections. That is, given a set of graphs, compute a single ordering that works well for visualizing each graph as a matrix.
To automatically assess quality, we observed that patterns in the matrix can be seen as a form of spatial auto-correlation and thus proposed the use of Moran's $I$ as a global measure of quality. Moran's $I$ readily implies a distance metric between two vertices that can be used in algorithms such as leaf order, as an alternative to other common functions that measure adjacencies or neighborhood similarity.

Algorithmically, computing simultaneous orderings has received little attention. We generalized the \emph{union} approach that is found in MultiPiles, but observed that this may lead to hiding structural information from the ordering algorithms. Instead, we proposed a generic \emph{collection-aware} approach that avoids such loss of information and showed how to apply this approach to the common leaf order and barycenter methods.

Our experiments demonstrate that our collection-aware approach is effective, especially leaf order based on the Moran's $I$-inspired metric.
Our collection-aware leaf ordering method using $\delta_I$ is the most versatile and consistently performs equally or better than the other algorithms, though the magnitude of improvement varies between datasets.
For other algorithm comparisons, there is less consistency in relative performance, also interacting with dataset.
Our results confirm that the potential information loss as sketched for the union approach indeed occurs, also in real datasets. Though it does not occur to such a degree to really cause arbitrary orderings as in hypothetical constructed cases, it nonetheless leads to inferior orderings when not accounting for the collection of graphs in the ordering algorithms.

\subparagraph{Future work.}
There is not a clear metric, the optimization of which gives visually the best or most useful ordering, even for a single graph. Though we postulate that Moran's $I$ can be effective here, further research is needed on how this concept actually relates to perceiving structure in matrix visualizations. Our focus was with algorithm-compatible measures, but such future work can likely leverage the work on magnostics \cite{behrisch2016magnostics}. Furthermore, it may be worthwhile to investigate the perceptual effects that improving Moran's $I$ brings; for example, Beecham \etal~\cite{beecham2016map} study the ``just noticeable difference'' of Moran's $I$ and their experiment includes weight grids as well. It suggests a degree to which Moran's $I$ must improve for an observer to reliably identify the improved ordering, though in our case, we have black-and-white and symmetric matrices, which may influence the perception. 

With our work, we provide (to the best of our knowledge) the first explicit definition of the simultaneous ordering problem for visualizing graphs using matrices. Our general concept of making algorithms collection-aware is applicable to other alternatives. Especially distance-based approaches like leaf order are easily amended. But we could also consider other ways in making algorithms collection-aware, and indeed, some algorithms may be more suitable than others for such adaptation. Can we, for example, in fact directly modify the clustering algorithm in leaf ordering, or change the objective function in its final stage, to explicitly optimize e.g., the smallest sum?

Finally, our experiments show differences depending on the number of graphs and vertices in the dataset. Whereas for a few matrices (or a few MultiPiles) a single ordering can be effective, when using many matrices, one is bound to have to make concessions for many of them. As such, the question is perhaps, whether we could permit small changes to the ordering for each matrix, to better highlight the structure in each graph. In our next steps, we plan to address this challenge of ``stable'' orderings.

To conclude, perhaps the main take-away of our work here is that simultaneous matrix ordering is a complex algorithmic-visualization problem. We have now taken the first steps in bringing this to the fore and addressing the algorithmic challenges.





\begin{thebibliography}{10}

\bibitem{ABHIF2013}
B.~Alper, B.~Bach, N.~Henry~Riche, T.~Isenberg, and J.-D. Fekete.
\newblock Weighted graph comparison techniques for brain connectivity analysis.
\newblock In {\em Proc. SIGCHI Conference on Human Factors in Computing
  Systems}, pp. 483--492, 2013. doi: {{%
10\hspace{.1pt}\discretionary{.}{%
}{.}\hspace{.4pt}1145\discretionary{/}{%
}{/}2470654\hspace{.1pt}\discretionary{.}{%
}{.}\hspace{.4pt}2470724}}


\bibitem{bach2015multipiles}
B.~Bach, N.~Henry-Riche, T.~Dwyer, T.~Madhyastha, J.-D. Fekete, and
  T.~Grabowski.
\newblock Small {M}ulti{P}iles: Piling time to explore temporal patterns in
  dynamic networks.
\newblock {\em Computer Graphics Forum}, 34:31--40, 05 2015. doi: {{%
10\hspace{.1pt}\discretionary{.}{%
}{.}\hspace{.4pt}1111\discretionary{/}{%
}{/}cgf\hspace{.1pt}\discretionary{.}{%
}{.}\hspace{.4pt}12615}}


\bibitem{bach2014cubix}
B.~Bach, E.~Pietriga, and J.-D. Fekete.
\newblock Visualizing dynamic networks with matrix cubes.
\newblock In {\em Proc. SIGCHI Conference on Human Factors in Computing
  Systems}, pp. 877--–886, 2014. doi: {{%
10\hspace{.1pt}\discretionary{.}{%
}{.}\hspace{.4pt}1145\discretionary{/}{%
}{/}2556288\hspace{.1pt}\discretionary{.}{%
}{.}\hspace{.4pt}2557010}}


\bibitem{bar2001fast}
Z.~Bar-Joseph, D.~K. Gifford, and T.~S. Jaakkola.
\newblock Fast optimal leaf ordering for hierarchical clustering.
\newblock {\em Bioinformatics}, 17(suppl 1):S22--S29, 2001.

\bibitem{beck2017dynamic}
F.~Beck, M.~Burch, S.~Diehl, and D.~Weiskopf.
\newblock A taxonomy and survey of dynamic graph visualization.
\newblock {\em Computer Graphics Forum}, 36(1):133--159, 2017. doi: {{%
10\hspace{.1pt}\discretionary{.}{%
}{.}\hspace{.4pt}1111\discretionary{/}{%
}{/}cgf\hspace{.1pt}\discretionary{.}{%
}{.}\hspace{.4pt}12791}}


\bibitem{beecham2016map}
R.~Beecham, J.~Dykes, W.~Meulemans, A.~Slingsby, C.~Turkay, and J.~Wood.
\newblock Map {LineUps}: effects of spatial structure on graphical inference.
\newblock {\em IEEE Transactions on Visualization and Computer Graphics},
  23(1):391--400, 2016.

\bibitem{behrisch2016magnostics}
M.~Behrisch, B.~Bach, M.~Blumenschein, M.~Delz, L.~von Rueden, J.-D. Fekete,
  and T.~Schreck.
\newblock Magnostics: Image-based search of interesting matrix views for guided
  network exploration.
\newblock {\em IEEE Transactions on Visualization and Computer Graphics},
  23:31--40, 01 2016. doi: {{%
10\hspace{.1pt}\discretionary{.}{%
}{.}\hspace{.4pt}1109\discretionary{/}{%
}{/}TVCG\hspace{.1pt}\discretionary{.}{%
}{.}\hspace{.4pt}2016\hspace{.1pt}\discretionary{.}{%
}{.}\hspace{.4pt}2598467}}


\bibitem{behrisch2016matrix}
M.~Behrisch, B.~Bach, N.~Henry~Riche, T.~Schreck, and J.-D. Fekete.
\newblock Matrix reordering methods for table and network visualization.
\newblock In {\em Computer Graphics Forum}, vol.~35, pp. 693--716, 2016.

\bibitem{behrisch2018quality}
M.~Behrisch, M.~Blumenschein, N.~W. Kim, L.~Shao, M.~El-Assady, J.~Fuchs,
  D.~Seebacher, A.~Diehl, U.~Brandes, H.~Pfister, T.~Schreck, D.~Weiskopf, and
  D.~A. Keim.
\newblock Quality metrics for information visualization.
\newblock {\em Computer Graphics Forum}, 37(3):625--662, 2018. doi: {{%
10\hspace{.1pt}\discretionary{.}{%
}{.}\hspace{.4pt}1111\discretionary{/}{%
}{/}cgf\hspace{.1pt}\discretionary{.}{%
}{.}\hspace{.4pt}13446}}


\bibitem{blasius2013}
T.~Bl\"asius, S.~G. Kobourov, and I.~Rutter.
\newblock Simultaneous embedding of planar graphs.
\newblock In R.~Tamassia, ed., {\em Handbook of Graph Drawing and
  Visualization}, pp. 349--383. CRC Press, 2013.

\bibitem{brandes2007optimal}
U.~Brandes.
\newblock Optimal leaf ordering of complete binary trees.
\newblock {\em Journal of Discrete Algorithms}, 5(3):546--552, 2007.

\bibitem{cabello2011geometric}
S.~Cabello, M.~J. van Kreveld, G.~Liotta, H.~Meijer, B.~Speckmann, and
  K.~Verbeek.
\newblock Geometric simultaneous embeddings of a graph and a matching.
\newblock {\em Journal of Graph Algorithms and Applications}, 15(1):79--96,
  2011.

\bibitem{2opt1958}
G.~A. Croes.
\newblock A method for solving traveling-salesman problems.
\newblock {\em Operations Research}, 6(6):791--812, 1958.

\bibitem{eades1994edge}
P.~Eades and N.~C. Wormald.
\newblock Edge crossings in drawings of bipartite graphs.
\newblock {\em Algorithmica}, 11(4):379--403, 1994.

\bibitem{Elmqvist2008ZAME}
N.~Elmqvist, T.-N. Do, H.~Goodell, N.~Henry, and J.-D. Fekete.
\newblock {ZAME}: Interactive large-scale graph visualization.
\newblock In {\em Proc. 2008 {IEEE} Pacific Visualization Symposium}, pp.
  215--222, 2008. doi: {{%
10\hspace{.1pt}\discretionary{.}{%
}{.}\hspace{.4pt}1109\discretionary{/}{%
}{/}pacificvis\hspace{.1pt}\discretionary{.}{%
}{.}\hspace{.4pt}2008\hspace{.1pt}\discretionary{.}{%
}{.}\hspace{.4pt}4475479}}


\bibitem{estrella2008}
A.~Estrella-Balderrama, E.~Gassner, M.~J{\"u}nger, M.~Percan, M.~Schaefer, and
  M.~Schulz.
\newblock Simultaneous geometric graph embeddings.
\newblock In {\em Proc. 15th International Symposium on Graph Drawing}, LNCS
  4875, pp. 280--290, 2008.

\bibitem{fekete2015reorderjs}
J.-D. Fekete.
\newblock {Reorder.js: A {J}ava{S}cript Library to Reorder Tables and
  Networks}.
\newblock In {\em Abstr. 2015 {IEEE VIS} posters}, 2015.
\newblock Available at
  \url{https://hal.inria.fr/hal-01214274/file/reorder.pdf}.

\bibitem{fowler2011characterizations}
J.~J. Fowler, M.~J{\"u}nger, S.~G. Kobourov, and M.~Schulz.
\newblock Characterizations of restricted pairs of planar graphs allowing
  simultaneous embedding with fixed edges.
\newblock {\em Computational Geometry}, 44(8):385--398, 2011.

\bibitem{gansner1993technique}
E.~R. Gansner, E.~Koutsofios, S.~C. North, and K.-P. Vo.
\newblock A technique for drawing directed graphs.
\newblock {\em IEEE Transactions on Software Engineering}, 19(3):214--230,
  1993.

\bibitem{gemmetto2014mitigation}
V.~Gemmetto, A.~Barrat, and C.~Cattuto.
\newblock Mitigation of infectious disease at school: targeted class closure vs
  school closure.
\newblock {\em BMC infectious diseases}, 14:Article no. 695, 2014. doi: {{%
10\hspace{.1pt}\discretionary{.}{%
}{.}\hspace{.4pt}1186\discretionary{/}{%
}{/}s12879\discretionary{%
}{-}{-}014\discretionary{%
}{-}{-}0695\discretionary{%
}{-}{-}9}}


\bibitem{GFC2004}
M.~{Ghoniem}, J.-D. {Fekete}, and P.~{Castagliola}.
\newblock A comparison of the readability of graphs using node-link and
  matrix-based representations.
\newblock In {\em Proc. IEEE Symposium on Information Visualization}, pp.
  17--24, 2004. doi: {{%
10\hspace{.1pt}\discretionary{.}{%
}{.}\hspace{.4pt}1109\discretionary{/}{%
}{/}INFVIS\hspace{.1pt}\discretionary{.}{%
}{.}\hspace{.4pt}2004\hspace{.1pt}\discretionary{.}{%
}{.}\hspace{.4pt}1}}


\bibitem{haeupler2013}
B.~{Haeupler}, K.~{Jampani}, and A.~{Lubiw}.
\newblock Testing simultaneous planarity when the common graph is 2-connected.
\newblock {\em Journal of Graph Algorithms and Applications}, 17(3):147--171,
  2013. doi: {{%
10\hspace{.1pt}\discretionary{.}{%
}{.}\hspace{.4pt}7155\discretionary{/}{%
}{/}jgaa\hspace{.1pt}\discretionary{.}{%
}{.}\hspace{.4pt}00289}}


\bibitem{henryriche2007matlink}
N.~Henry~Riche and J.-D. Fekete.
\newblock Mat{L}ink: Enhanced matrix visualization for analyzing social
  networks.
\newblock In {\em Proc. 13th IFIP TC13 International Conference on
  Human-Computer Interaction}, LNCS 4663, pp. 288--302, 2007.

\bibitem{HHKN2021}
M.~I. {Hossain}, V.~{Huroyan}, S.~{Kobourov}, and R.~{Navarrete}.
\newblock Multi-perspective, simultaneous embedding.
\newblock {\em IEEE Transactions on Visualization and Computer Graphics},
  27(2):1569--1579, 2021. doi: {{%
10\hspace{.1pt}\discretionary{.}{%
}{.}\hspace{.4pt}1109\discretionary{/}{%
}{/}TVCG\hspace{.1pt}\discretionary{.}{%
}{.}\hspace{.4pt}2020\hspace{.1pt}\discretionary{.}{%
}{.}\hspace{.4pt}3030373}}


\bibitem{Isenberg:2017:VMC}
P.~Isenberg, F.~Heimerl, S.~Koch, T.~Isenberg, P.~Xu, C.~Stolper, M.~Sedlmair,
  J.~Chen, T.~M{\"o}ller, and J.~Stasko.
\newblock vispubdata.org: A metadata collection about {IEEE} {V}isualization
  ({VIS}) publications.
\newblock {\em IEEE Transactions on Visualization and Computer Graphics},
  23(9):2199--2206, 2017. doi: {{%
10\hspace{.1pt}\discretionary{.}{%
}{.}\hspace{.4pt}1109\discretionary{/}{%
}{/}TVCG\hspace{.1pt}\discretionary{.}{%
}{.}\hspace{.4pt}2016\hspace{.1pt}\discretionary{.}{%
}{.}\hspace{.4pt}2615308}}


\bibitem{LenstraKan1975me}
J.~K. Lenstra and A.~H. G.~R. Kan.
\newblock Some simple applications of the travelling salesman problem.
\newblock {\em Operational Research Quarterly (1970--1977)}, 26(4):717--733,
  1975.

\bibitem{makinen2005barycenter}
E.~M{\"a}kinen and H.~Siirtola.
\newblock The barycenter heuristic and the reorderable matrix.
\newblock {\em Informatica (Slovenia)}, 29(3):357--364, 2005.

\bibitem{McCormick1969ME}
W.~McCormick, S.~B. Deutsch, J.~Martin, and P.~Schweitzer.
\newblock Identification of data structures and relationships by matrix
  reordering techniques.
\newblock 1969.

\bibitem{McCormick1972ME}
W.~T. McCormick, P.~J. Schweitzer, and T.~W. White.
\newblock Problem decomposition and data reorganization by a clustering
  technique.
\newblock {\em Operations Research}, 20(5):993--1009, 1972.

\bibitem{moran1950notes}
P.~A.~P. Moran.
\newblock Notes on continuous stochastic phenomena.
\newblock {\em Biometrika}, 37(1/2):17--23, 1950.

\bibitem{Niermann2005stress}
S.~Niermann.
\newblock Optimizing the ordering of tables with evolutionary computation.
\newblock {\em The American Statistician}, 59(1):41--46, 2005. doi: {{%
10\hspace{.1pt}\discretionary{.}{%
}{.}\hspace{.4pt}1198\discretionary{/}{%
}{/}000313005X22770}}


\bibitem{shen2007paths}
Z.~Shen and K.-L. Ma.
\newblock Path visualization for adjacency matrices.
\newblock In {\em Proc. 9th Joint Eurographics / IEEE VGTC Conference on
  Visualization}, 2007. doi: {{%
10\hspace{.1pt}\discretionary{.}{%
}{.}\hspace{.4pt}2312\discretionary{/}{%
}{/}VisSym\discretionary{/}{%
}{/}EuroVis07\discretionary{/}{%
}{/}083\discretionary{%
}{-}{-}090}}


\bibitem{stehle2011high}
J.~Stehl{\'e}, N.~Voirin, A.~Barrat, C.~Cattuto, L.~Isella, J.-F. Pinton,
  M.~Quaggiotto, W.~Van~den Broeck, C.~R{\'e}gis, B.~Lina, et~al.
\newblock High-resolution measurements of face-to-face contact patterns in a
  primary school.
\newblock {\em PloS one}, 6(8):e23176, 2011. doi: {{%
10\hspace{.1pt}\discretionary{.}{%
}{.}\hspace{.4pt}1371\discretionary{/}{%
}{/}journal\hspace{.1pt}\discretionary{.}{%
}{.}\hspace{.4pt}pone\hspace{.1pt}\discretionary{.}{%
}{.}\hspace{.4pt}0023176}}


\bibitem{yi2010timematrix}
J.~S. Yi, N.~Elmqvist, and S.~Lee.
\newblock Time{M}atrix: Analyzing temporal social networks using interactive
  matrix-based visualizations.
\newblock {\em International Journal of Human–Computer Interaction},
  26(11--12):1031--1051, 2010. doi: {{%
10\hspace{.1pt}\discretionary{.}{%
}{.}\hspace{.4pt}1080\discretionary{/}{%
}{/}10447318\hspace{.1pt}\discretionary{.}{%
}{.}\hspace{.4pt}2010\hspace{.1pt}\discretionary{.}{%
}{.}\hspace{.4pt}516722}}


\end{thebibliography}


\begin{thebibliography}{1}

\bibitem{LenstraKan1975me}
J.~K. Lenstra and A.~H. G.~R. Kan.
\newblock Some simple applications of the travelling salesman problem.
\newblock {\em Operational Research Quarterly (1970--1977)}, 26(4):717--733,
  1975.

\bibitem{moran1950notes}
P.~A.~P. Moran.
\newblock Notes on continuous stochastic phenomena.
\newblock {\em Biometrika}, 37(1/2):17--23, 1950.

\end{thebibliography}

\putbib
\end{bibunit}
\flushcolsend

\clearpage
\begin{bibunit}
\appendix

\begin{figure*}[t]
\begin{minipage}{\textwidth}
\centering
\makeatletter
\medskip
{\sffamily\huge\vgtc@sectionfont%
Simultaneous Matrix Orderings for Graph Collections\\
\medskip
\bfseries\LARGE Supplementary Material\\}
\vspace{14pt}
{\large\sffamily\vgtc@sectionfont%
Nathan van Beusekom, Wouter Meulemans, and Bettina Speckmann\\}
\vspace{10pt}
\makeatother
\includegraphics{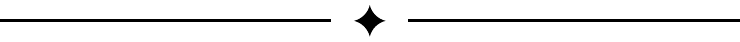}
\end{minipage}
\end{figure*}

\setcounter{figure}{0}
\setcounter{table}{0}

\section{Moran's \emph{I} for matrix visualization}
\label{app:moransi}

We here present further details on Moran's $I$ and its use for assessing matrix visualizations of graphs. In Section~\ref{sec:derivation} we present the full derivation of our simplified form of Moran's $I$, starting from its general form. For the most part, our derivation here is structurally the same as the one presented in the main paper -- but we present it in full nonetheless, to make this supplementary material more accessible.

In Section~\ref{sec:distinct} we consider distinct orderings for rows and columns in more detail. This allows us to further reason about the behavior of Moran's $I$ in context of matrix visualizations.

\subsection{Derivation}\label{sec:derivation}

Consider a graph $G = (V,E)$ on $n$ vertices with some ordering $\rho$ for its rows and a possibly different ordering $\sigma$ for its columns. This combination of graph and two orderings implies a $n \times n$ 0-1 matrix $M$ where $M_{ij} = 1$ if $(\rho(i),\sigma(j)) \in E$ and 0 otherwise. The cells of $M$ are interpreted as the spatial units for the sake of Moran's $I$, and contains all necessary information to derive Moran's $I$. We hence omit dependencies on $G$, $\rho$ and $\sigma$ in our derivation here, for sake of notational simplicity. 

We use slightly different notation than perhaps conventional for the general form of Moran's $I$, so as to distinguish between the general form and our simplified form. Specifically, we use $r$ instead of $N$ to denote the number of spatial units (regions) and use a weights matrix~$T$ (topology) with sum $t$ instead of $w$ with sum $W$. Moreover, we refer to the regions using indices $a$ and $b$ rather than $i$ and $j$.

Moran's $I$ is defined over $r$ spatial units, which have associated values $x_a$ for $a \in \{1, \ldots, r\}$. Furthermore, an $r\times r$ matrix $T$ encodes the weights for (typically neighboring) regions: entry $T_{ab}$ is the weight between region $a$ and $b$. We use $t$ to denote the sum over all weights in~$T$. Moreover, let $\overline{x}$ denote the average value $\frac{\sum_{a=1}^r x_a}{r}$. The general form of Moran's $I$ is as follows \cite{moran1950notes}:

\[ I = \frac{r}{t} \cdot \frac{\sum_{a=1}^{r}\sum_{b=1}^{r} T_{ab} (x_a - \overline{x})(x_b - \overline{x})}{\sum_{a=1}^{r} (x_a - \overline{x})^2} \]

Cells of $M$ represent the spatial units and thus we have $r = n^2$.
Note that we can directly map between region indices $a$ to matrix indices $(i,j)$ using $a = i+n \cdot j$.
Furthermore, we observe that $\sum_{a=1}^r x_a$ thus becomes $\sum_{i=1}^n\sum_{j=1}^n M_{ij}$ which simplifies to $m$, the total number of entries with value 1 in $M$. Note that $m$ matches the number of edges in $G$ in a directed model: that is, an undirected edge $\{u,v\}$ with $u \neq v$ is interpreted as occurring twice in $E$, once as $(u,v)$ and once as $(v,u)$.
Hence, $\overline{x} = \frac{m}{n^2}$ and we can rewrite $x_a - \overline{x} = \frac{x_a n^2 - m}{n^2}$. We now obtain the following expression:

\begin{align*} 
I &= \frac{n^2}{t} \cdot \frac{\sum_{a=1}^{r}\sum_{b=1}^{r} T_{ab}\frac{x_a n^2 - m}{n^2} \cdot \frac{x_b n^2 - m}{n^2}}{\sum_{a=1}^{r} (\frac{x_a n^2 - m}{n^2})^2} \\
 &= \frac{n^2}{t} \cdot \frac{\sum_{a=1}^{r}\sum_{b=1}^{r} T_{ab} (x_a n^2 - m)(x_b n^2 - m)}{\sum_{i=1}^{r} (x_a n^2 - m)^2} 
\end{align*}
We model Rook's adjacency in $T$. Thus, specifically, weight $T_{ab}$ is 1 if region $a$ and $b$ are adjacent cells of matrix $M$. That is, $\sum_{a=1}^{r}\sum_{b=1}^{r} T_{ab} (x_a n^2 - m)(x_b n^2 - m)$ effectively adds up $(x_a n^2 - m)(x_b n^2 - m)$ for horizontally or vertically adjacent regions $a$ and $b$. Since region values are either 0 or 1, we thus distinguish three cases: 
\begin{enumerate}
    \item[$B$)] If $x_a = x_b = 1$, then $(x_a n^2 - m)(x_b n^2 - m)$ simplifies to $(n^2 - m)^2$. We denote this case with $B$, since it refers to two black cells being adjacent in the visualization.
    \item[$W$)] If $x_a = x_b = 0$, then$(x_a n^2 - m)(x_b n^2 - m)$ simplifies to $(-m)^2 = m^2$. We denote this case with $W$, since it refers to two white cells being adjacent in the visualization.
    \item[$D$)] If $x_a \neq x_b$, one of the values is 0 and the other 1; observing the symmetry in the formula, $(x_a n^2 - m)(x_b n^2 - m)$ simplifies to $-(n^2 - m)m$. We denote this case with $D$, since it refers to two differently colored cells being adjacent in the visualization.
\end{enumerate}
We thus can simplify this summation by simply counting the number of occurrences of each case, and we identify $B$, $W$ and $D$ with this count. We observe that in total, there are $2n(n-1)$ adjacencies. Since Rook's adjacency gives a symmetric matrix $T$, we count every adjacency twice, but also we find that $t = 4n(n-1)$. Hence, $\sum_{a=1}^{r}\sum_{b=1}^{r} T_{ab} (x_a n^2 - m)(x_b n^2 - m)$ simplifies to $2 B \cdot (n^2 - m)^2 + 2 W\cdot m^2 - 2 D\cdot (n^2 - m)m$. As such, we get the following form:
\[ I = \frac{n^2}{4n(n-1)} \cdot \frac{ 2 B \cdot (n^2 - m)^2 + 2 W\cdot m^2 - 2 D\cdot (n^2 - m)m}{\sum_{i=1}^{r} (x_a n^2 - m)^2} \]
We can apply the same principle to the denominator $\sum_{i=1}^{r} (x_a n^2 - m)^2$, by counting the number of cells for which $x_a = 1$ and for which it is 0. The former case is simply $m$ and the latter the remaining $n^2 - m$ cells. Thus, we get $m \cdot (n^2 - m)^2 + (n^2-m) \cdot (-m)^2 = m \cdot ((n^2 - m)(n^2 - m) + (n^2 - m)m) = m \cdot ((n^2 - m)(n^2 - m + m) = m(n^2 -m)n^2 = n^2 m (n^2 - m)$. Filling this into the overall formula, we get:
\begin{align*}
    I &= \frac{n^2}{4n(n-1)} \cdot \frac{2 B \cdot (n^2 - m)^2 + 2 W\cdot m^2 - 2 D\cdot (n^2 - m)m}{n^2 m (n^2 - m)} \\ 
    &= \frac{2 n^2}{4n(n-1)n^2 m (n^2 - m)} \cdot ( B \cdot (n^2 - m)^2 + W\cdot m^2 - D\cdot (n^2 - m)m)\\
    &= \frac{1}{2n(n-1) m (n^2 - m)} \cdot ( B \cdot (n^2 - m)^2 + W\cdot m^2 - D\cdot (n^2 - m)m)
\end{align*}
Using a shorthand $N = 2n(n-1) m (n^2 - m)$ and our earlier observation, that $B+W+D = 2n(n-1)$ and thus $D = 2n(n-1) - B- W$, we derive:
\begin{align*}
I &= \frac{1}{N} \cdot ( B \cdot (n^2 - m)^2 + W\cdot m^2 - D\cdot (n^2 - m)m) \\
 &= \frac{1}{N} \cdot ( B \cdot (n^2 - m)^2 + W\cdot m^2 - (2n(n-1) - B- W) \cdot (n^2 - m)m) \\
 &= \frac{1}{N} \cdot ( B \cdot ((n^2 - m)^2 + (n^2 - m)m)  + W \cdot (m^2 + (n^2 - m)m) \\
 & \hspace{4em} - 2n(n-1) (n^2 - m)m) \\
 &= \frac{1}{N} \cdot ( B \cdot ((n^2 - m)^2 + (n^2 - m)m)  + W \cdot (m^2 + (n^2 - m)m) - N) \\
 &= \frac{1}{N} \cdot ( B \cdot ((n^2 - m)^2 + (n^2 - m)m)  + W \cdot (m^2 + (n^2 - m)m)) - 1 
\end{align*}
We observe that $(n^2 - m)^2 + (n^2 - m)m = (n^2 - m)(n^2 - m+m) = (n^2 - m)n^2$ and  $(m^2 + (n^2 - m)m = m(m + n^2 - m) = m n^2$. We can thus simplify to the following:
\[ I = B \cdot \frac{(n^2 - m)n^2}{N}  + W \cdot \frac{m n^2}{N} - 1 \]
Defining $c_B$ as $\frac{(n^2 - m)n^2}{N} = \frac{(n^2 - m)n^2}{2n(n-1) m (n^2 - m)} = \frac{n}{2(n-1) m}$ and $c_W$ as $\frac{m n^2}{N} = \frac{m n^2}{2n(n-1) m (n^2 - m)} = \frac{n}{2 (n-1) (n^2 - m)}$, we obtain the following simplified formula.
\[ I = c_B \cdot B + c_W \cdot W - 1 \]
\smallskip

\subsection{The effect of distinct orderings}\label{sec:distinct}

Let us now turn towards understanding the relation between Moran's~$I$ and the use of distinct row and column orderings.
We have so far dropped the dependency on the graph $G$ and orderings $\rho$ (rows) and $\sigma$ (columns): after all, such a triple determines $M$ completely, which in turn determines the value of $I$ as expressed in the previous section.
However, to compare different orderings, we need to extend our notation to be able to distinguish between the value of Moran's $I$ for different orderings.
So, we make this explicit by denoting the value of Moran's $I$ for the matrix visualization of a graph $G$ using row ordering $\rho$ and column ordering $\sigma$ by $I(G,\rho,\sigma)$. When $\rho = \sigma$, we use $I(G,\rho)$.

We observe that $c_B$ and $c_W$ depend only on $m$ and $n$ which are fully defined by $G$ and thus independent of the orderings. In counting the number of adjacencies $B$ and $W$, the orderings do matter. However, we observe that, in Rook's adjacency, adjacencies are either within the same row, or within the same column. The row ordering $\rho$ has no influence on adjacencies within a row and, analogously, column ordering $\sigma$ has no influence on adjacencies within a column. We can thus count these independently, and we use $B_\textrm{v}(\rho)$ and $B_\textrm{h}(\sigma)$ to denote the vertical and horizontal black-black adjacencies respectively. Analogously, $W_\textrm{v}(\rho)$ and $W_\textrm{h}(\sigma)$ to count vertical and horizontal white-white adjacencies. Applying these observations, we thus obtain the form as used in the main paper:
\begin{align*}
I(G,\rho,\sigma) 
&= c_B(G) \cdot (B_\textrm{h}(G,\rho) + B_\textrm{v}(G,\sigma)) \\
&+ c_W(G) \cdot (W_\textrm{h}(G,\rho) + W_\textrm{v}(G,\sigma)) \\ 
&- 1 
\end{align*}
The above indicates that we may count horizontal and vertical adjacencies independently and as such the row and column ordering are independently influencing Moran's $I$. As a result, for undirected graphs, the score of any matrix using two distinct orderings is the average of using either ordering for both rows and columns (see the proof below). This implies that for undirected graphs it is in fact optimal according to Moran's $I$ to use the same ordering.

\begin{lemma}\label{lem:twoOrders}
For an undirected graph $G = (V,E)$ with two orderings $\rho$ and $\sigma$, $I(G,\rho,\sigma) = \frac{1}{2} (I(G,\rho) + I(G,\sigma) )$.
\end{lemma}
\begin{proof}
The main observation is that our graph is undirected. That is, using the same ordering gives a symmetric matrix: the number of adjacencies of any type is then the same between the rows and the columns.
That is, $B_\textrm{h}(G,\pi) = B_\textrm{v}(G,\pi)$ and $W_\textrm{h}(G,\pi) = W_\textrm{v}(G,\pi)$ hold for any ordering $\pi$. As such, we omit the subscript from now on.

From the definition of Moran's $I$ in our context, it is clear that the horizontal and vertical order independently influence the score. Indeed, using the observation above, we can rewrite  Moran's $I$ to a sum of two terms, one of which is defined wholly by $\rho$ and the other by $\sigma$. That is, $I(G,\rho,\sigma) = f(G,\rho) + f(G,\sigma) - 1$ where $f(G,\pi) =  c_B(G) B(G,\pi) + c_W(G) W(G,\pi)$.

With this notation, we find that $I(G,\rho) = 2 f(G,\rho)$ and $I(G,\sigma) = 2 f(G,\sigma)$. 
Rewriting now gives us $I(G,\rho,\sigma) = I(G,\rho)/2 + I(G,\sigma)/2 = \frac{1}{2} (I(G,\rho) + I(G,\sigma))$. In other words, using the distinct orderings gives the average quality of using the two same-ordering options.
\end{proof}
Note that this lemma does rely on our use of Rook's adjacency and $G$ being an undirected graph. Figure~\ref{fig:counter} illustrates a counterexample for each condition: if the graph is directed or Moran's $I$ is used with Queen's adjacency, using distinct orderings may lead to a lower Moran's $I$ than using either ordering for both rows and columns.

We further note that the argument in Lemma~\ref{lem:twoOrders} does not rely on the specific constants used: indeed any measure that is a linear combination of items that are independent between row and column orderings would satisfy this lemma. Examples include simply counting the number of black-black adjacencies. Indeed, similar observations have been made by Lenstra and Kan~\cite{LenstraKan1975me} for the measure of effectiveness.

\begin{figure}
    \centering
    \includegraphics{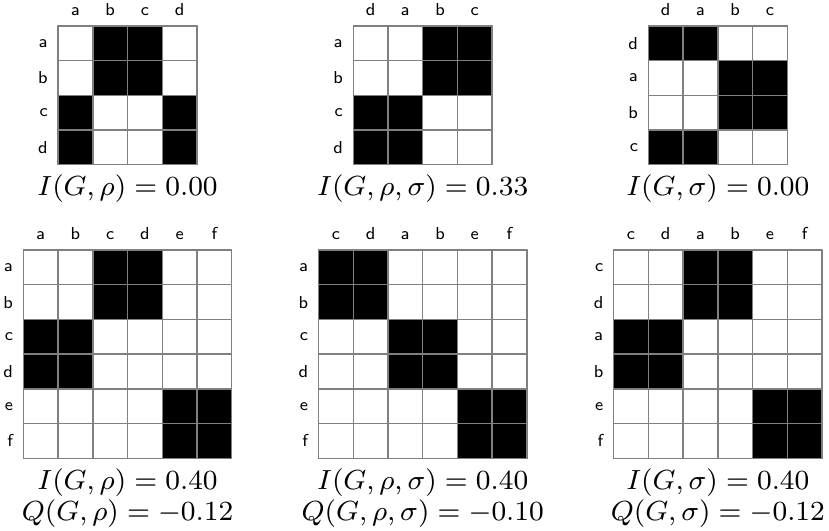}
    \caption{Top row: for this directed graph $G$, $I(G,\rho,\sigma)$ (middle) is larger than $I(G,\rho)$ and $I(G,\sigma)$; note that, as $G$ is directed, using the same ordering results in an asymmetric matrix. Bottom row: using Moran's $I$ with Queen's adjacency (denoted here as function $Q$) may result in distinct orderings (middle) being better than using the same ordering, even for an undirected graph.}
    \label{fig:counter}
\end{figure}

\begin{figure*}[t]
\centering
\includegraphics[page=4]{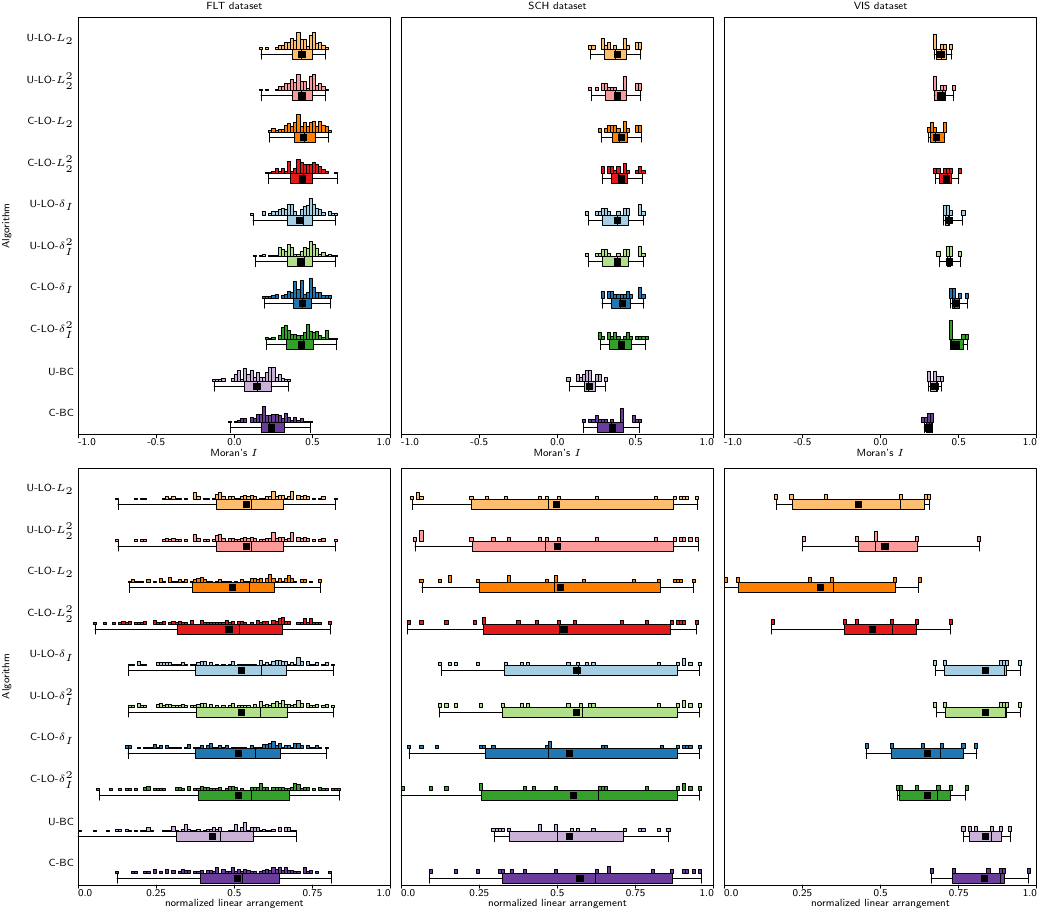}
\caption{Performance of all ten algorithms for each of the datasets (columns) in terms of Moran's $I$ (top row) and Linear Arrangement (bottom row). For both, higher values indicate higher quality. Main observation is that squaring the distance function has little influence.}
\label{fig:charts-full}
\end{figure*}

\section{Additional figures and tables}
\label{app:figures}

Figure~\ref{fig:charts-full} shows the performance of all ten algorithms, including the leaf order variants that are using a squared distance measure.
In Tables~\ref{tab:diff_ALL_MI} through~\ref{tab:diff_VIS_MI}, we provide the differences in minimum, median and average Moran's $I$ between each pair of algorithms for all datasets as well as per dataset. Figures~\ref{fig:flt-results-1-7} to \ref{fig:flt-results-92-96} show the matrix visualizations based on the ordering computed by each of the six main algorithms on the FLT dataset: each algorithm is a column, spread over multiple pages; each row corresponds to one of the 96 graphs in the collection.

\putbib
\end{bibunit}
\flushcolsend

\begin{table*}[t]
\centering
\caption{Differences in minimum, median and average Moran's $I$, averaged over all datasets. A positive value implies that the algorithm for column outperforms the algorithm for the row.}
\label{tab:diff_ALL_MI}
\begin{tabu} to \linewidth {X[1,l,m]|X[1,r,m]X[1,r,m]|X[1,r,m]X[1,r,m]|X[1,r,m]X[1,r,m]|X[1,r,m]X[1,r,m]|X[1,r,m]|X[1,r,m]}
\toprule
{\bfseries minimum} & {U-LO-$L_2$} & {U-LO-$L_2^2$} & {C-LO-$L_2$} & {C-LO-$L_2^2$} & {U-LO-$\delta_I$} & {U-LO-$\delta_I^2$} & {C-LO-$\delta_I$} & {C-LO-$\delta_I^2$} & {U-BC} & {C-BC} \\
 \midrule
 {U-LO-$L_2$} &  & 0.00 & 0.03 & 0.04 & -0.00 & -0.00 & 0.07 & 0.07 & -0.16 & -0.10 \\
 {U-LO-$L_2^2$} & -0.00 &  & 0.03 & 0.04 & -0.00 & -0.01 & 0.07 & 0.07 & -0.16 & -0.10 \\
 \midrule
 {C-LO-$L_2$} & -0.03 & -0.03 &  & 0.01 & -0.03 & -0.03 & 0.04 & 0.04 & -0.19 & -0.13 \\
 {C-LO-$L_2^2$} & -0.04 & -0.04 & -0.01 &  & -0.05 & -0.05 & 0.02 & 0.02 & -0.20 & -0.15 \\
 \midrule
 {U-LO-$\delta_I$} & 0.00 & 0.00 & 0.03 & 0.05 &  & -0.00 & 0.07 & 0.07 & -0.15 & -0.10 \\
 {U-LO-$\delta_I^2$} & 0.00 & 0.01 & 0.03 & 0.05 & 0.00 &  & 0.07 & 0.07 & -0.15 & -0.10 \\
 \midrule
 {C-LO-$\delta_I$} & -0.07 & -0.07 & -0.04 & -0.02 & -0.07 & -0.07 &  & -0.00 & -0.22 & -0.17 \\
 {C-LO-$\delta_I^2$} & -0.07 & -0.07 & -0.04 & -0.02 & -0.07 & -0.07 & 0.00 &  & -0.22 & -0.17 \\
 \midrule
 {U-BC} & 0.16 & 0.16 & 0.19 & 0.20 & 0.15 & 0.15 & 0.22 & 0.22 &  & 0.05 \\
 \midrule
 {C-BC} & 0.10 & 0.10 & 0.13 & 0.15 & 0.10 & 0.10 & 0.17 & 0.17 & -0.05 &  \\
\bottomrule
\end{tabu}

\medskip
\begin{tabu} to \linewidth {X[1,l,m]|X[1,r,m]X[1,r,m]|X[1,r,m]X[1,r,m]|X[1,r,m]X[1,r,m]|X[1,r,m]X[1,r,m]|X[1,r,m]|X[1,r,m]}
\toprule
{\bfseries median} & {U-LO-$L_2$} & {U-LO-$L_2^2$} & {C-LO-$L_2$} & {C-LO-$L_2^2$} & {U-LO-$\delta_I$} & {U-LO-$\delta_I^2$} & {C-LO-$\delta_I$} & {C-LO-$\delta_I^2$} & {U-BC} & {C-BC} \\
 \midrule
 {U-LO-$L_2$} &  & 0.00 & 0.00 & 0.02 & 0.03 & 0.03 & 0.04 & 0.03 & -0.16 & -0.10 \\
 {U-LO-$L_2^2$} & -0.00 &  & 0.00 & 0.02 & 0.02 & 0.03 & 0.04 & 0.03 & -0.16 & -0.10 \\
 \midrule
 {C-LO-$L_2$} & -0.00 & -0.00 &  & 0.02 & 0.02 & 0.03 & 0.04 & 0.03 & -0.16 & -0.10 \\
 {C-LO-$L_2^2$} & -0.02 & -0.02 & -0.02 &  & 0.00 & 0.01 & 0.02 & 0.01 & -0.18 & -0.12 \\
 \midrule
 {U-LO-$\delta_I$} & -0.03 & -0.02 & -0.02 & -0.00 &  & 0.00 & 0.02 & 0.01 & -0.19 & -0.13 \\
 {U-LO-$\delta_I^2$} & -0.03 & -0.03 & -0.03 & -0.01 & -0.00 &  & 0.01 & 0.00 & -0.19 & -0.13 \\
 \midrule
 {C-LO-$\delta_I$} & -0.04 & -0.04 & -0.04 & -0.02 & -0.02 & -0.01 &  & -0.01 & -0.20 & -0.14 \\
 {C-LO-$\delta_I^2$} & -0.03 & -0.03 & -0.03 & -0.01 & -0.01 & -0.00 & 0.01 &  & -0.19 & -0.13 \\
 \midrule
 {U-BC} & 0.16 & 0.16 & 0.16 & 0.18 & 0.19 & 0.19 & 0.20 & 0.19 &  & 0.06 \\
 \midrule
 {C-BC} & 0.10 & 0.10 & 0.10 & 0.12 & 0.13 & 0.13 & 0.14 & 0.13 & -0.06 &  \\
\bottomrule
\end{tabu}

\medskip
\begin{tabu} to \linewidth {X[1,l,m]|X[1,r,m]X[1,r,m]|X[1,r,m]X[1,r,m]|X[1,r,m]X[1,r,m]|X[1,r,m]X[1,r,m]|X[1,r,m]|X[1,r,m]}
\toprule
{\bfseries average} & {U-LO-$L_2$} & {U-LO-$L_2^2$} & {C-LO-$L_2$} & {C-LO-$L_2^2$} & {U-LO-$\delta_I$} & {U-LO-$\delta_I^2$} & {C-LO-$\delta_I$} & {C-LO-$\delta_I^2$} & {U-BC} & {C-BC} \\
 \midrule
 {U-LO-$L_2$} &  & 0.00 & 0.00 & 0.02 & 0.02 & 0.02 & 0.04 & 0.04 & -0.17 & -0.10 \\
 {U-LO-$L_2^2$} & -0.00 &  & 0.00 & 0.02 & 0.01 & 0.01 & 0.04 & 0.04 & -0.17 & -0.10 \\
 \midrule
 {C-LO-$L_2$} & -0.00 & -0.00 &  & 0.02 & 0.01 & 0.01 & 0.04 & 0.04 & -0.17 & -0.10 \\
 {C-LO-$L_2^2$} & -0.02 & -0.02 & -0.02 &  & -0.01 & -0.01 & 0.02 & 0.02 & -0.19 & -0.12 \\
 \midrule
 {U-LO-$\delta_I$} & -0.02 & -0.01 & -0.01 & 0.01 &  & -0.00 & 0.03 & 0.02 & -0.18 & -0.12 \\
 {U-LO-$\delta_I^2$} & -0.02 & -0.01 & -0.01 & 0.01 & 0.00 &  & 0.03 & 0.02 & -0.18 & -0.12 \\
 \midrule
 {C-LO-$\delta_I$} & -0.04 & -0.04 & -0.04 & -0.02 & -0.03 & -0.03 &  & -0.01 & -0.21 & -0.15 \\
 {C-LO-$\delta_I^2$} & -0.04 & -0.04 & -0.04 & -0.02 & -0.02 & -0.02 & 0.01 &  & -0.21 & -0.14 \\
 \midrule
 {U-BC} & 0.17 & 0.17 & 0.17 & 0.19 & 0.18 & 0.18 & 0.21 & 0.21 &  & 0.07 \\
 \midrule
 {C-BC} & 0.10 & 0.10 & 0.10 & 0.12 & 0.12 & 0.12 & 0.15 & 0.14 & -0.07 &  \\
\bottomrule
\end{tabu}
\end{table*}
\begin{table*}[t]
\centering
\caption{Differences in minimum, median and average Moran's $I$ for the FLT dataset. A positive value implies that the algorithm for column outperforms the algorithm for the row.}
\label{tab:diff_FLT_MI}
\begin{tabu} to \linewidth {X[1,l,m]|X[1,r,m]X[1,r,m]|X[1,r,m]X[1,r,m]|X[1,r,m]X[1,r,m]|X[1,r,m]X[1,r,m]|X[1,r,m]|X[1,r,m]}
\toprule
{\bfseries minimum} & {U-LO-$L_2$} & {U-LO-$L_2^2$} & {C-LO-$L_2$} & {C-LO-$L_2^2$} & {U-LO-$\delta_I$} & {U-LO-$\delta_I^2$} & {C-LO-$\delta_I$} & {C-LO-$\delta_I^2$} & {U-BC} & {C-BC} \\
 \midrule
 {U-LO-$L_2$} &  & 0.00 & 0.06 & 0.05 & -0.05 & -0.03 & 0.02 & 0.04 & -0.30 & -0.19 \\
 {U-LO-$L_2^2$} & 0.00 &  & 0.06 & 0.05 & -0.05 & -0.03 & 0.02 & 0.04 & -0.30 & -0.19 \\
 \midrule
 {C-LO-$L_2$} & -0.06 & -0.06 &  & -0.01 & -0.11 & -0.09 & -0.03 & -0.02 & -0.35 & -0.25 \\
 {C-LO-$L_2^2$} & -0.05 & -0.05 & 0.01 &  & -0.10 & -0.08 & -0.03 & -0.01 & -0.35 & -0.24 \\
 \midrule
 {U-LO-$\delta_I$} & 0.05 & 0.05 & 0.11 & 0.10 &  & 0.02 & 0.07 & 0.09 & -0.25 & -0.14 \\
 {U-LO-$\delta_I^2$} & 0.03 & 0.03 & 0.09 & 0.08 & -0.02 &  & 0.06 & 0.07 & -0.27 & -0.16 \\
 \midrule
 {C-LO-$\delta_I$} & -0.02 & -0.02 & 0.03 & 0.03 & -0.07 & -0.06 &  & 0.01 & -0.32 & -0.22 \\
 {C-LO-$\delta_I^2$} & -0.04 & -0.04 & 0.02 & 0.01 & -0.09 & -0.07 & -0.01 &  & -0.33 & -0.23 \\
 \midrule
 {U-BC} & 0.30 & 0.30 & 0.35 & 0.35 & 0.25 & 0.27 & 0.32 & 0.33 &  & 0.10 \\
 \midrule
 {C-BC} & 0.19 & 0.19 & 0.25 & 0.24 & 0.14 & 0.16 & 0.22 & 0.23 & -0.10 &  \\
\bottomrule
\end{tabu}

\medskip
\begin{tabu} to \linewidth {X[1,l,m]|X[1,r,m]X[1,r,m]|X[1,r,m]X[1,r,m]|X[1,r,m]X[1,r,m]|X[1,r,m]X[1,r,m]|X[1,r,m]|X[1,r,m]}
\toprule
{\bfseries median} & {U-LO-$L_2$} & {U-LO-$L_2^2$} & {C-LO-$L_2$} & {C-LO-$L_2^2$} & {U-LO-$\delta_I$} & {U-LO-$\delta_I^2$} & {C-LO-$\delta_I$} & {C-LO-$\delta_I^2$} & {U-BC} & {C-BC} \\
 \midrule
 {U-LO-$L_2$} &  & 0.00 & 0.02 & 0.01 & 0.01 & 0.02 & 0.01 & 0.01 & -0.28 & -0.20 \\
 {U-LO-$L_2^2$} & 0.00 &  & 0.02 & 0.01 & 0.01 & 0.02 & 0.01 & 0.01 & -0.28 & -0.20 \\
 \midrule
 {C-LO-$L_2$} & -0.02 & -0.02 &  & -0.01 & -0.00 & 0.00 & -0.01 & -0.01 & -0.30 & -0.22 \\
 {C-LO-$L_2^2$} & -0.01 & -0.01 & 0.01 &  & 0.00 & 0.01 & -0.00 & 0.00 & -0.29 & -0.21 \\
 \midrule
 {U-LO-$\delta_I$} & -0.01 & -0.01 & 0.00 & -0.00 &  & 0.00 & -0.01 & -0.00 & -0.29 & -0.22 \\
 {U-LO-$\delta_I^2$} & -0.02 & -0.02 & -0.00 & -0.01 & -0.00 &  & -0.01 & -0.01 & -0.30 & -0.22 \\
 \midrule
 {C-LO-$\delta_I$} & -0.01 & -0.01 & 0.01 & 0.00 & 0.01 & 0.01 &  & 0.00 & -0.29 & -0.21 \\
 {C-LO-$\delta_I^2$} & -0.01 & -0.01 & 0.01 & -0.00 & 0.00 & 0.01 & -0.00 &  & -0.29 & -0.22 \\
 \midrule
 {U-BC} & 0.28 & 0.28 & 0.30 & 0.29 & 0.29 & 0.30 & 0.29 & 0.29 &  & 0.08 \\
 \midrule
 {C-BC} & 0.20 & 0.20 & 0.22 & 0.21 & 0.22 & 0.22 & 0.21 & 0.22 & -0.08 &  \\
\bottomrule
\end{tabu}

\medskip
\begin{tabu} to \linewidth {X[1,l,m]|X[1,r,m]X[1,r,m]|X[1,r,m]X[1,r,m]|X[1,r,m]X[1,r,m]|X[1,r,m]X[1,r,m]|X[1,r,m]|X[1,r,m]}
\toprule
{\bfseries average} & {U-LO-$L_2$} & {U-LO-$L_2^2$} & {C-LO-$L_2$} & {C-LO-$L_2^2$} & {U-LO-$\delta_I$} & {U-LO-$\delta_I^2$} & {C-LO-$\delta_I$} & {C-LO-$\delta_I^2$} & {U-BC} & {C-BC} \\
 \midrule
 {U-LO-$L_2$} &  & 0.00 & 0.01 & 0.00 & -0.01 & -0.01 & 0.00 & -0.00 & -0.29 & -0.20 \\
 {U-LO-$L_2^2$} & 0.00 &  & 0.01 & 0.00 & -0.01 & -0.01 & 0.00 & -0.00 & -0.29 & -0.20 \\
 \midrule
 {C-LO-$L_2$} & -0.01 & -0.01 &  & -0.01 & -0.02 & -0.02 & -0.01 & -0.02 & -0.30 & -0.21 \\
 {C-LO-$L_2^2$} & -0.00 & -0.00 & 0.01 &  & -0.01 & -0.01 & 0.00 & -0.01 & -0.29 & -0.20 \\
 \midrule
 {U-LO-$\delta_I$} & 0.01 & 0.01 & 0.02 & 0.01 &  & 0.00 & 0.01 & 0.01 & -0.28 & -0.18 \\
 {U-LO-$\delta_I^2$} & 0.01 & 0.01 & 0.02 & 0.01 & -0.00 &  & 0.01 & 0.01 & -0.28 & -0.18 \\
 \midrule
 {C-LO-$\delta_I$} & -0.00 & -0.00 & 0.01 & -0.00 & -0.01 & -0.01 &  & -0.01 & -0.29 & -0.20 \\
 {C-LO-$\delta_I^2$} & 0.00 & 0.00 & 0.02 & 0.01 & -0.01 & -0.01 & 0.01 &  & -0.28 & -0.19 \\
 \midrule
 {U-BC} & 0.29 & 0.29 & 0.30 & 0.29 & 0.28 & 0.28 & 0.29 & 0.28 &  & 0.09 \\
 \midrule
 {C-BC} & 0.20 & 0.20 & 0.21 & 0.20 & 0.18 & 0.18 & 0.20 & 0.19 & -0.09 &  \\
\bottomrule
\end{tabu}
\end{table*}
\begin{table*}[t]
\centering
\caption{Differences in minimum, median and average Moran's $I$ for the SCH dataset. A positive value implies that the algorithm for column outperforms the algorithm for the row.}
\label{tab:diff_SCH_MI}
\begin{tabu} to \linewidth {X[1,l,m]|X[1,r,m]X[1,r,m]|X[1,r,m]X[1,r,m]|X[1,r,m]X[1,r,m]|X[1,r,m]X[1,r,m]|X[1,r,m]|X[1,r,m]}
\toprule
{\bfseries minimum} & {U-LO-$L_2$} & {U-LO-$L_2^2$} & {C-LO-$L_2$} & {C-LO-$L_2^2$} & {U-LO-$\delta_I$} & {U-LO-$\delta_I^2$} & {C-LO-$\delta_I$} & {C-LO-$\delta_I^2$} & {U-BC} & {C-BC} \\
 \midrule
 {U-LO-$L_2$} &  & 0.00 & 0.07 & 0.07 & -0.02 & -0.02 & 0.07 & 0.06 & -0.13 & -0.05 \\
 {U-LO-$L_2^2$} & -0.00 &  & 0.06 & 0.07 & -0.02 & -0.02 & 0.07 & 0.06 & -0.14 & -0.05 \\
 \midrule
 {C-LO-$L_2$} & -0.07 & -0.06 &  & 0.01 & -0.08 & -0.08 & 0.01 & -0.01 & -0.20 & -0.12 \\
 {C-LO-$L_2^2$} & -0.07 & -0.07 & -0.01 &  & -0.09 & -0.09 & -0.00 & -0.01 & -0.21 & -0.12 \\
 \midrule
 {U-LO-$\delta_I$} & 0.02 & 0.02 & 0.08 & 0.09 &  & 0.00 & 0.09 & 0.08 & -0.12 & -0.03 \\
 {U-LO-$\delta_I^2$} & 0.02 & 0.02 & 0.08 & 0.09 & 0.00 &  & 0.09 & 0.08 & -0.12 & -0.03 \\
 \midrule
 {C-LO-$\delta_I$} & -0.07 & -0.07 & -0.01 & 0.00 & -0.09 & -0.09 &  & -0.01 & -0.21 & -0.12 \\
 {C-LO-$\delta_I^2$} & -0.06 & -0.06 & 0.01 & 0.01 & -0.08 & -0.08 & 0.01 &  & -0.20 & -0.11 \\
 \midrule
 {U-BC} & 0.13 & 0.14 & 0.20 & 0.21 & 0.12 & 0.12 & 0.21 & 0.20 &  & 0.09 \\
 \midrule
 {C-BC} & 0.05 & 0.05 & 0.12 & 0.12 & 0.03 & 0.03 & 0.12 & 0.11 & -0.09 &  \\
\bottomrule
\end{tabu}

\medskip
\begin{tabu} to \linewidth {X[1,l,m]|X[1,r,m]X[1,r,m]|X[1,r,m]X[1,r,m]|X[1,r,m]X[1,r,m]|X[1,r,m]X[1,r,m]|X[1,r,m]|X[1,r,m]}
\toprule
{\bfseries median} & {U-LO-$L_2$} & {U-LO-$L_2^2$} & {C-LO-$L_2$} & {C-LO-$L_2^2$} & {U-LO-$\delta_I$} & {U-LO-$\delta_I^2$} & {C-LO-$\delta_I$} & {C-LO-$\delta_I^2$} & {U-BC} & {C-BC} \\
 \midrule
 {U-LO-$L_2$} &  & 0.00 & 0.03 & 0.02 & 0.02 & 0.01 & 0.04 & 0.02 & -0.17 & -0.03 \\
 {U-LO-$L_2^2$} & -0.00 &  & 0.03 & 0.02 & 0.01 & 0.01 & 0.03 & 0.02 & -0.17 & -0.03 \\
 \midrule
 {C-LO-$L_2$} & -0.03 & -0.03 &  & -0.01 & -0.01 & -0.02 & 0.01 & -0.01 & -0.20 & -0.06 \\
 {C-LO-$L_2^2$} & -0.02 & -0.02 & 0.01 &  & -0.00 & -0.00 & 0.02 & 0.01 & -0.19 & -0.04 \\
 \midrule
 {U-LO-$\delta_I$} & -0.02 & -0.01 & 0.01 & 0.00 &  & -0.00 & 0.02 & 0.01 & -0.19 & -0.04 \\
 {U-LO-$\delta_I^2$} & -0.01 & -0.01 & 0.02 & 0.00 & 0.00 &  & 0.02 & 0.01 & -0.18 & -0.04 \\
 \midrule
 {C-LO-$\delta_I$} & -0.04 & -0.03 & -0.01 & -0.02 & -0.02 & -0.02 &  & -0.01 & -0.21 & -0.06 \\
 {C-LO-$\delta_I^2$} & -0.02 & -0.02 & 0.01 & -0.01 & -0.01 & -0.01 & 0.01 &  & -0.19 & -0.05 \\
 \midrule
 {U-BC} & 0.17 & 0.17 & 0.20 & 0.19 & 0.19 & 0.18 & 0.21 & 0.19 &  & 0.14 \\
 \midrule
 {C-BC} & 0.03 & 0.03 & 0.06 & 0.04 & 0.04 & 0.04 & 0.06 & 0.05 & -0.14 &  \\
\bottomrule
\end{tabu}

\medskip
\begin{tabu} to \linewidth {X[1,l,m]|X[1,r,m]X[1,r,m]|X[1,r,m]X[1,r,m]|X[1,r,m]X[1,r,m]|X[1,r,m]X[1,r,m]|X[1,r,m]|X[1,r,m]}
\toprule
{\bfseries average} & {U-LO-$L_2$} & {U-LO-$L_2^2$} & {C-LO-$L_2$} & {C-LO-$L_2^2$} & {U-LO-$\delta_I$} & {U-LO-$\delta_I^2$} & {C-LO-$\delta_I$} & {C-LO-$\delta_I^2$} & {U-BC} & {C-BC} \\
 \midrule
 {U-LO-$L_2$} &  & 0.00 & 0.02 & 0.03 & 0.00 & 0.00 & 0.03 & 0.03 & -0.18 & -0.03 \\
 {U-LO-$L_2^2$} & -0.00 &  & 0.02 & 0.02 & 0.00 & 0.00 & 0.03 & 0.02 & -0.18 & -0.03 \\
 \midrule
 {C-LO-$L_2$} & -0.02 & -0.02 &  & 0.00 & -0.02 & -0.02 & 0.01 & 0.00 & -0.20 & -0.06 \\
 {C-LO-$L_2^2$} & -0.03 & -0.02 & -0.00 &  & -0.02 & -0.02 & 0.01 & -0.00 & -0.21 & -0.06 \\
 \midrule
 {U-LO-$\delta_I$} & -0.00 & -0.00 & 0.02 & 0.02 &  & 0.00 & 0.03 & 0.02 & -0.18 & -0.03 \\
 {U-LO-$\delta_I^2$} & -0.00 & -0.00 & 0.02 & 0.02 & -0.00 &  & 0.03 & 0.02 & -0.18 & -0.03 \\
 \midrule
 {C-LO-$\delta_I$} & -0.03 & -0.03 & -0.01 & -0.01 & -0.03 & -0.03 &  & -0.01 & -0.21 & -0.06 \\
 {C-LO-$\delta_I^2$} & -0.03 & -0.02 & -0.00 & 0.00 & -0.02 & -0.02 & 0.01 &  & -0.21 & -0.06 \\
 \midrule
 {U-BC} & 0.18 & 0.18 & 0.20 & 0.21 & 0.18 & 0.18 & 0.21 & 0.21 &  & 0.15 \\
 \midrule
 {C-BC} & 0.03 & 0.03 & 0.06 & 0.06 & 0.03 & 0.03 & 0.06 & 0.06 & -0.15 &  \\
\bottomrule
\end{tabu}
\end{table*}
\begin{table*}[t]
\centering
\caption{Differences in minimum, median and average Moran's $I$ for the VIS dataset. A positive value implies that the algorithm for column outperforms the algorithm for the row.}
\label{tab:diff_VIS_MI}
\begin{tabu} to \linewidth {X[1,l,m]|X[1,r,m]X[1,r,m]|X[1,r,m]X[1,r,m]|X[1,r,m]X[1,r,m]|X[1,r,m]X[1,r,m]|X[1,r,m]|X[1,r,m]}
\toprule
{\bfseries minimum} & {U-LO-$L_2$} & {U-LO-$L_2^2$} & {C-LO-$L_2$} & {C-LO-$L_2^2$} & {U-LO-$\delta_I$} & {U-LO-$\delta_I^2$} & {C-LO-$\delta_I$} & {C-LO-$\delta_I^2$} & {U-BC} & {C-BC} \\
 \midrule
 {U-LO-$L_2$} &  & 0.00 & -0.03 & 0.01 & 0.06 & 0.03 & 0.11 & 0.10 & -0.04 & -0.07 \\
 {U-LO-$L_2^2$} & -0.00 &  & -0.03 & 0.01 & 0.06 & 0.03 & 0.11 & 0.10 & -0.04 & -0.07 \\
 \midrule
 {C-LO-$L_2$} & 0.03 & 0.03 &  & 0.04 & 0.09 & 0.07 & 0.14 & 0.14 & -0.00 & -0.03 \\
 {C-LO-$L_2^2$} & -0.01 & -0.01 & -0.04 &  & 0.05 & 0.03 & 0.10 & 0.10 & -0.04 & -0.07 \\
 \midrule
 {U-LO-$\delta_I$} & -0.06 & -0.06 & -0.09 & -0.05 &  & -0.02 & 0.05 & 0.05 & -0.09 & -0.12 \\
 {U-LO-$\delta_I^2$} & -0.03 & -0.03 & -0.07 & -0.03 & 0.02 &  & 0.07 & 0.07 & -0.07 & -0.10 \\
 \midrule
 {C-LO-$\delta_I$} & -0.11 & -0.11 & -0.14 & -0.10 & -0.05 & -0.07 &  & -0.00 & -0.14 & -0.17 \\
 {C-LO-$\delta_I^2$} & -0.10 & -0.10 & -0.14 & -0.10 & -0.05 & -0.07 & 0.00 &  & -0.14 & -0.17 \\
 \midrule
 {U-BC} & 0.04 & 0.04 & 0.00 & 0.04 & 0.09 & 0.07 & 0.14 & 0.14 &  & -0.03 \\
 \midrule
 {C-BC} & 0.07 & 0.07 & 0.03 & 0.07 & 0.12 & 0.10 & 0.17 & 0.17 & 0.03 &  \\
\bottomrule
\end{tabu}

\medskip
\begin{tabu} to \linewidth {X[1,l,m]|X[1,r,m]X[1,r,m]|X[1,r,m]X[1,r,m]|X[1,r,m]X[1,r,m]|X[1,r,m]X[1,r,m]|X[1,r,m]|X[1,r,m]}
\toprule
{\bfseries median} & {U-LO-$L_2$} & {U-LO-$L_2^2$} & {C-LO-$L_2$} & {C-LO-$L_2^2$} & {U-LO-$\delta_I$} & {U-LO-$\delta_I^2$} & {C-LO-$\delta_I$} & {C-LO-$\delta_I^2$} & {U-BC} & {C-BC} \\
 \midrule
 {U-LO-$L_2$} &  & 0.01 & -0.04 & 0.04 & 0.05 & 0.06 & 0.09 & 0.06 & -0.03 & -0.07 \\
 {U-LO-$L_2^2$} & -0.01 &  & -0.04 & 0.03 & 0.04 & 0.05 & 0.08 & 0.06 & -0.04 & -0.08 \\
 \midrule
 {C-LO-$L_2$} & 0.04 & 0.04 &  & 0.07 & 0.09 & 0.10 & 0.12 & 0.10 & 0.00 & -0.04 \\
 {C-LO-$L_2^2$} & -0.04 & -0.03 & -0.07 &  & 0.01 & 0.02 & 0.05 & 0.03 & -0.07 & -0.11 \\
 \midrule
 {U-LO-$\delta_I$} & -0.05 & -0.04 & -0.09 & -0.01 &  & 0.01 & 0.04 & 0.02 & -0.08 & -0.12 \\
 {U-LO-$\delta_I^2$} & -0.06 & -0.05 & -0.10 & -0.02 & -0.01 &  & 0.03 & 0.01 & -0.09 & -0.13 \\
 \midrule
 {C-LO-$\delta_I$} & -0.09 & -0.08 & -0.12 & -0.05 & -0.04 & -0.03 &  & -0.02 & -0.12 & -0.16 \\
 {C-LO-$\delta_I^2$} & -0.06 & -0.06 & -0.10 & -0.03 & -0.02 & -0.01 & 0.02 &  & -0.10 & -0.14 \\
 \midrule
 {U-BC} & 0.03 & 0.04 & -0.00 & 0.07 & 0.08 & 0.09 & 0.12 & 0.10 &  & -0.04 \\
 \midrule
 {C-BC} & 0.07 & 0.08 & 0.04 & 0.11 & 0.12 & 0.13 & 0.16 & 0.14 & 0.04 &  \\
\bottomrule
\end{tabu}

\medskip
\begin{tabu} to \linewidth {X[1,l,m]|X[1,r,m]X[1,r,m]|X[1,r,m]X[1,r,m]|X[1,r,m]X[1,r,m]|X[1,r,m]X[1,r,m]|X[1,r,m]|X[1,r,m]}
\toprule
{\bfseries average} & {U-LO-$L_2$} & {U-LO-$L_2^2$} & {C-LO-$L_2$} & {C-LO-$L_2^2$} & {U-LO-$\delta_I$} & {U-LO-$\delta_I^2$} & {C-LO-$\delta_I$} & {C-LO-$\delta_I^2$} & {U-BC} & {C-BC} \\
 \midrule
 {U-LO-$L_2$} &  & -0.00 & -0.03 & 0.03 & 0.06 & 0.05 & 0.10 & 0.09 & -0.04 & -0.08 \\
 {U-LO-$L_2^2$} & 0.00 &  & -0.03 & 0.03 & 0.06 & 0.05 & 0.10 & 0.10 & -0.04 & -0.08 \\
 \midrule
 {C-LO-$L_2$} & 0.03 & 0.03 &  & 0.06 & 0.09 & 0.09 & 0.13 & 0.13 & -0.01 & -0.05 \\
 {C-LO-$L_2^2$} & -0.03 & -0.03 & -0.06 &  & 0.02 & 0.02 & 0.06 & 0.06 & -0.07 & -0.11 \\
 \midrule
 {U-LO-$\delta_I$} & -0.06 & -0.06 & -0.09 & -0.02 &  & -0.00 & 0.04 & 0.04 & -0.09 & -0.13 \\
 {U-LO-$\delta_I^2$} & -0.05 & -0.05 & -0.09 & -0.02 & 0.00 &  & 0.04 & 0.04 & -0.09 & -0.13 \\
 \midrule
 {C-LO-$\delta_I$} & -0.10 & -0.10 & -0.13 & -0.06 & -0.04 & -0.04 &  & -0.00 & -0.14 & -0.18 \\
 {C-LO-$\delta_I^2$} & -0.09 & -0.10 & -0.13 & -0.06 & -0.04 & -0.04 & 0.00 &  & -0.13 & -0.17 \\
 \midrule
 {U-BC} & 0.04 & 0.04 & 0.01 & 0.07 & 0.09 & 0.09 & 0.14 & 0.13 &  & -0.04 \\
 \midrule
 {C-BC} & 0.08 & 0.08 & 0.05 & 0.11 & 0.13 & 0.13 & 0.18 & 0.17 & 0.04 &  \\
\bottomrule
\end{tabu}
\end{table*}

\clearpage

\begin{figure*}[t]
    \centering
    \includegraphics[page=1,width=\linewidth]{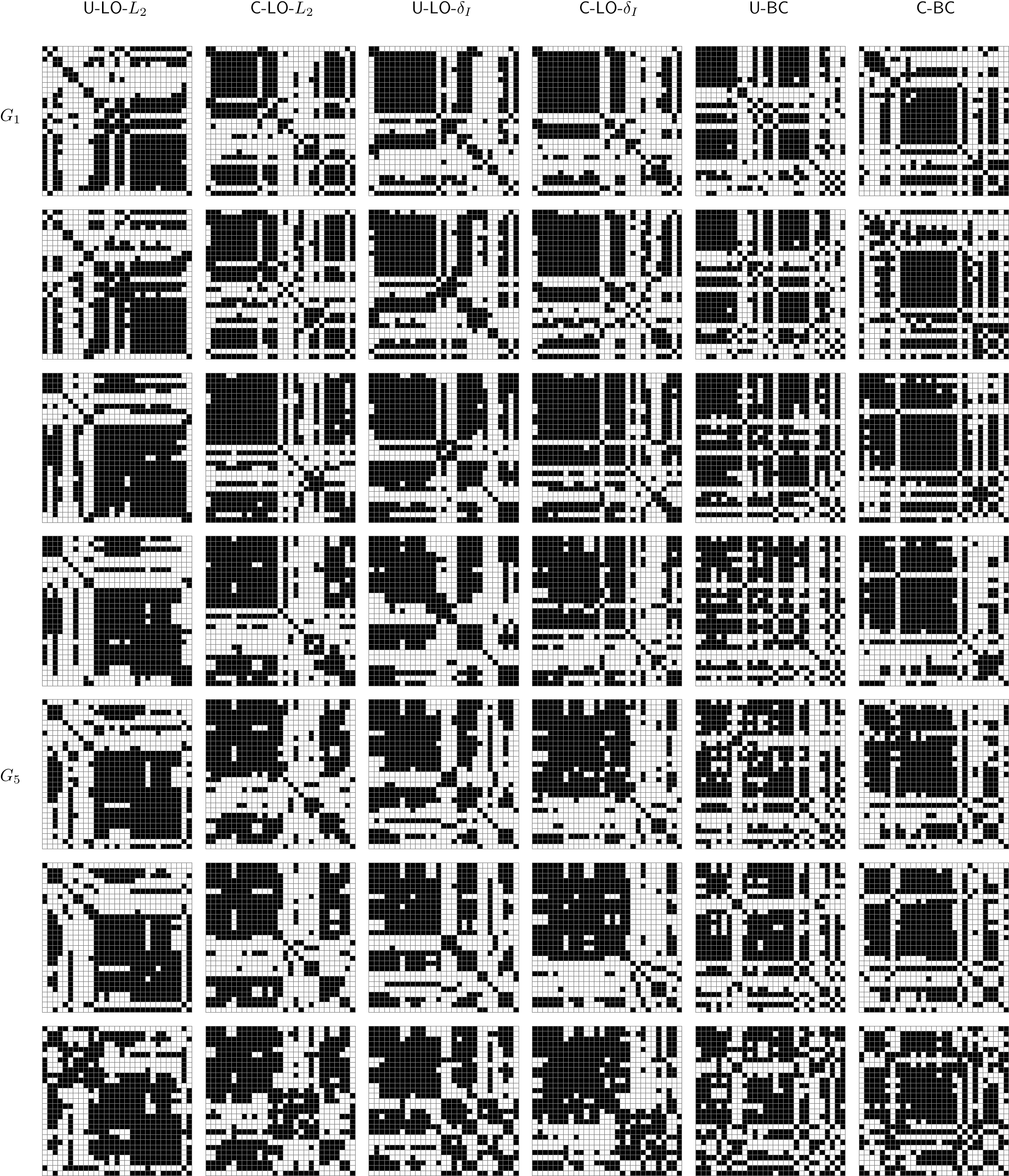}
    \caption{Matrix visualizations of graphs $G_{1}$ to $G_{7}$ of the FLT dataset, produced by the six algorithms considered in the main paper.}
    \label{fig:flt-results-1-7}
\end{figure*}

\clearpage

\begin{figure*}[t]
    \centering
    \includegraphics[page=2,width=\linewidth]{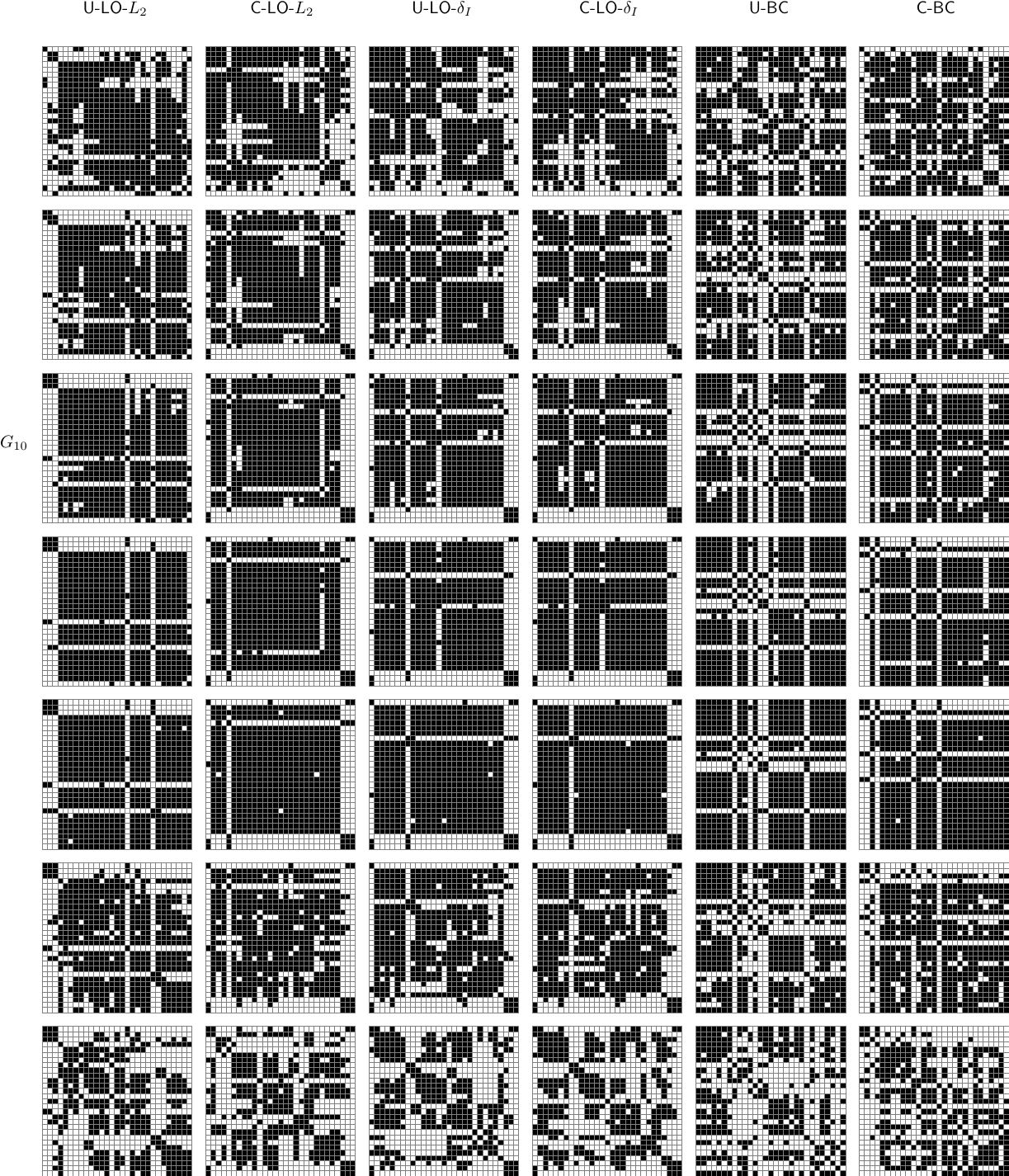}
    \caption{Matrix visualizations of graphs $G_{8}$ to $G_{14}$ of the FLT dataset, produced by the six algorithms considered in the main paper.}
    \label{fig:flt-results-8-14}
\end{figure*}

\clearpage

\begin{figure*}[t]
    \centering
    \includegraphics[page=3,width=\linewidth]{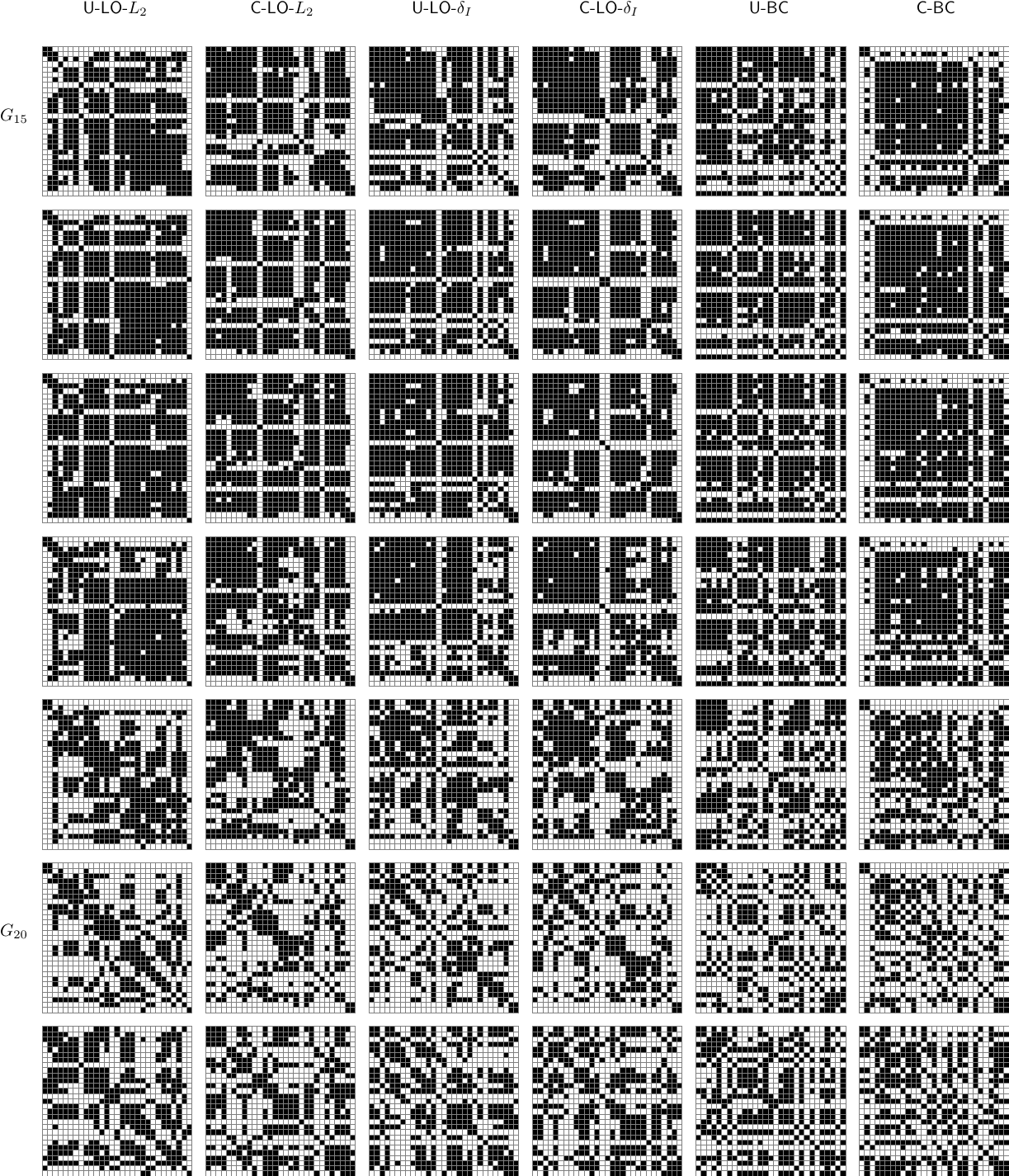}
    \caption{Matrix visualizations of graphs $G_{15}$ to $G_{21}$ of the FLT dataset, produced by the six algorithms considered in the main paper.}
    \label{fig:flt-results-15-21}
\end{figure*}

\clearpage

\begin{figure*}[t]
    \centering
    \includegraphics[page=4,width=\linewidth]{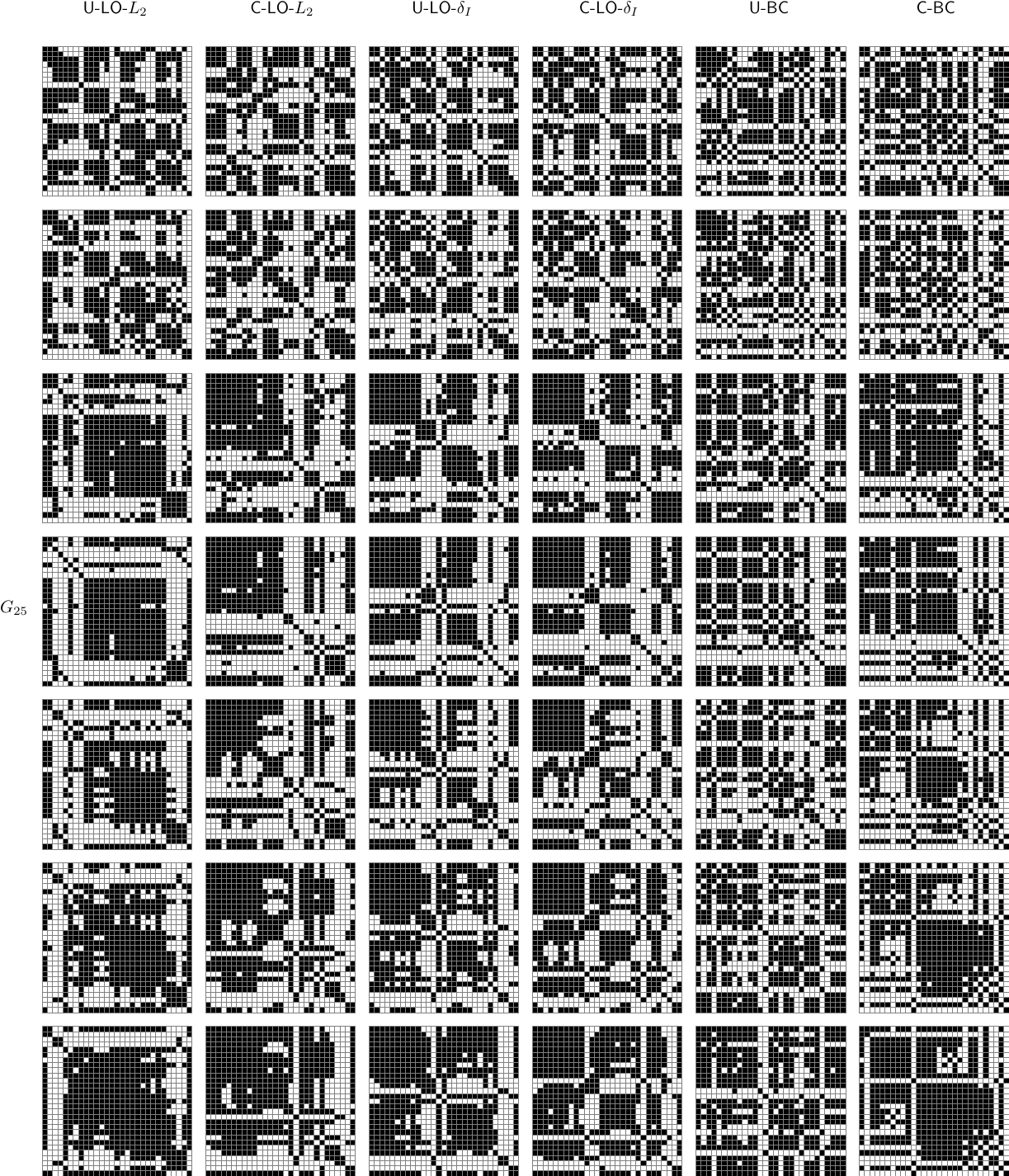}
    \caption{Matrix visualizations of graphs $G_{22}$ to $G_{28}$ of the FLT dataset, produced by the six algorithms considered in the main paper.}
    \label{fig:flt-results-22-28}
\end{figure*}

\clearpage

\begin{figure*}[t]
    \centering
    \includegraphics[page=5,width=\linewidth]{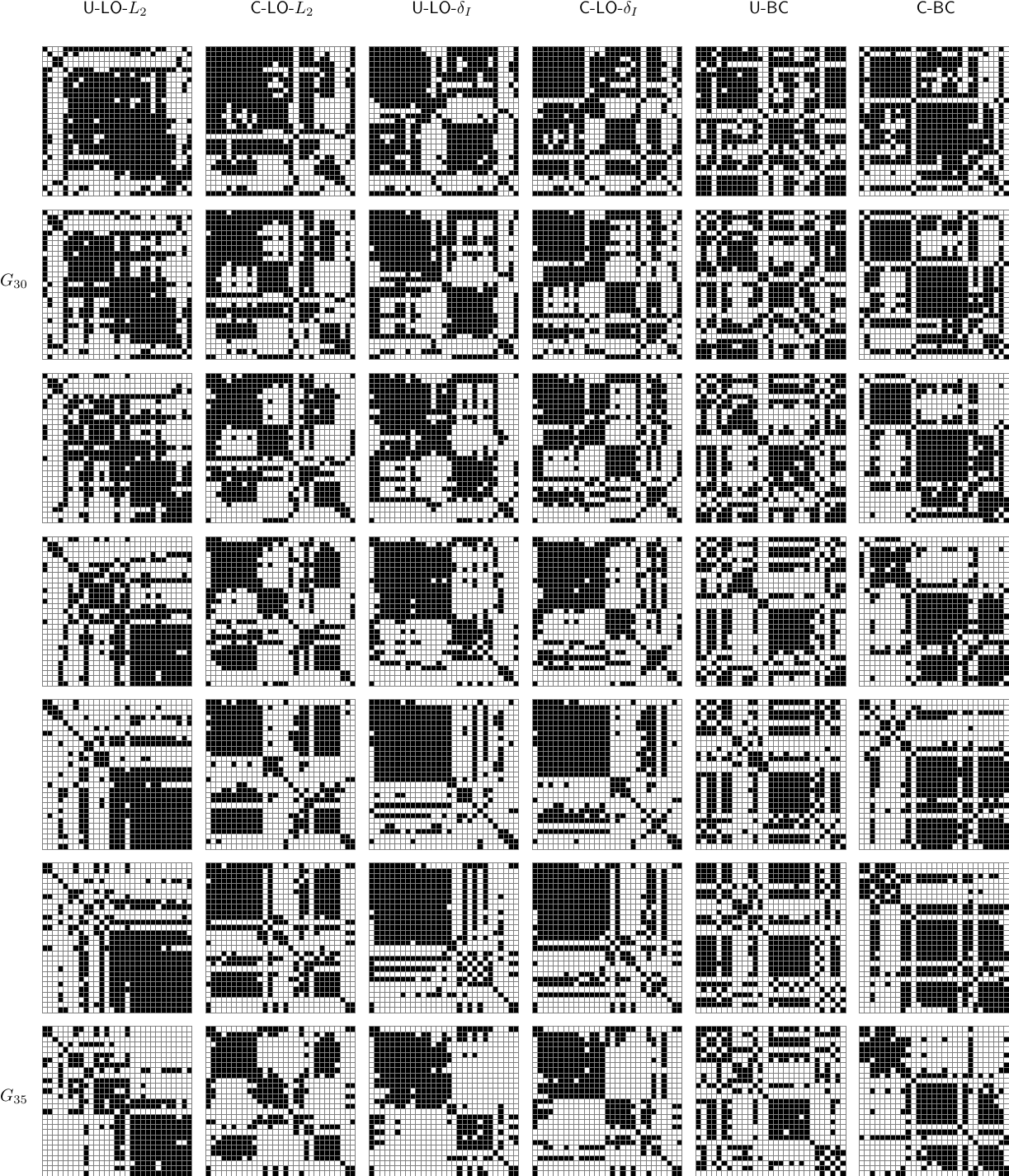}
    \caption{Matrix visualizations of graphs $G_{29}$ to $G_{35}$ of the FLT dataset, produced by the six algorithms considered in the main paper.}
    \label{fig:flt-results-29-35}
\end{figure*}

\clearpage

\begin{figure*}[t]
    \centering
    \includegraphics[page=6,width=\linewidth]{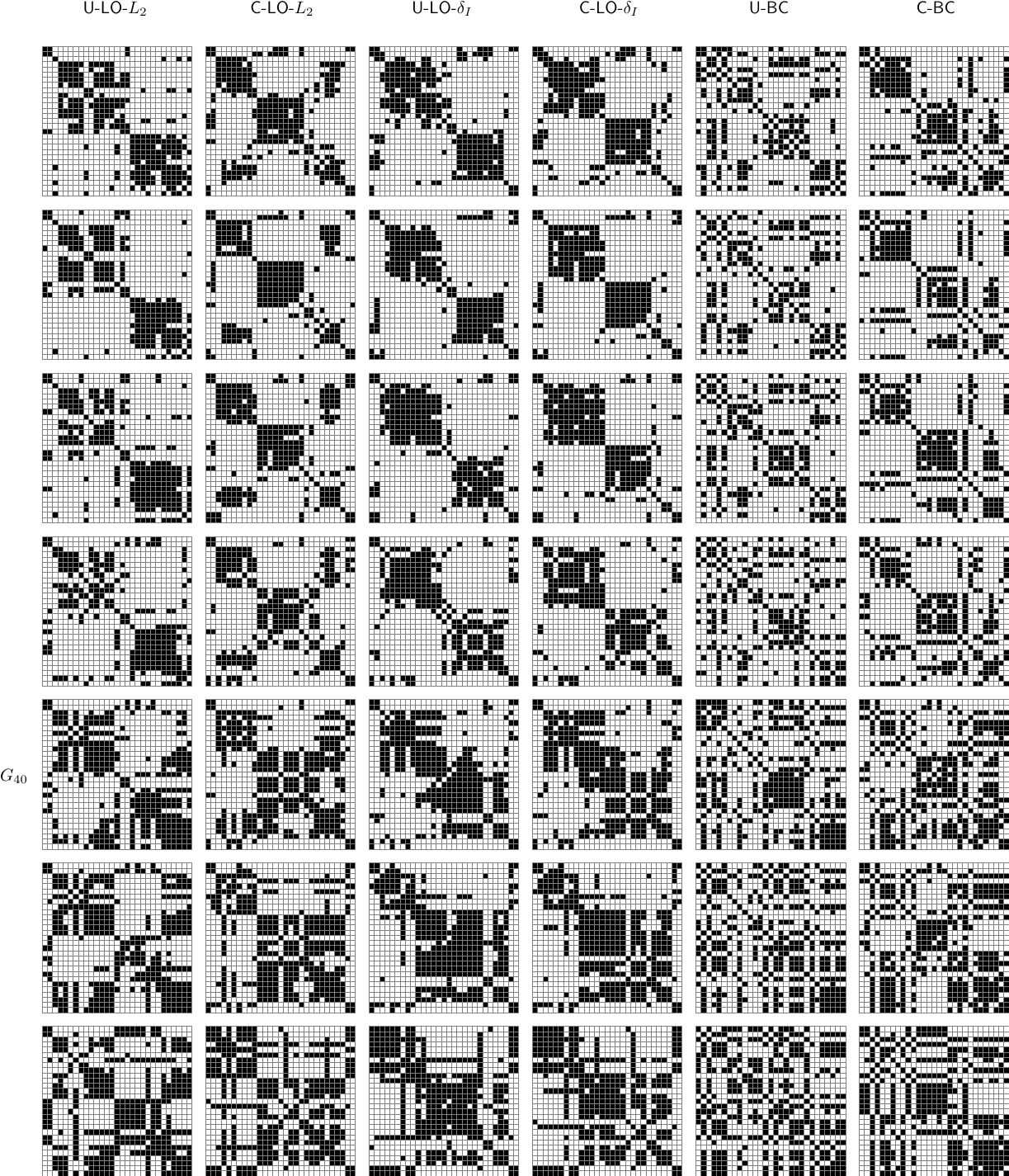}
    \caption{Matrix visualizations of graphs $G_{36}$ to $G_{42}$ of the FLT dataset, produced by the six algorithms considered in the main paper.}
    \label{fig:flt-results-36-42}
\end{figure*}

\clearpage

\begin{figure*}[t]
    \centering
    \includegraphics[page=7,width=\linewidth]{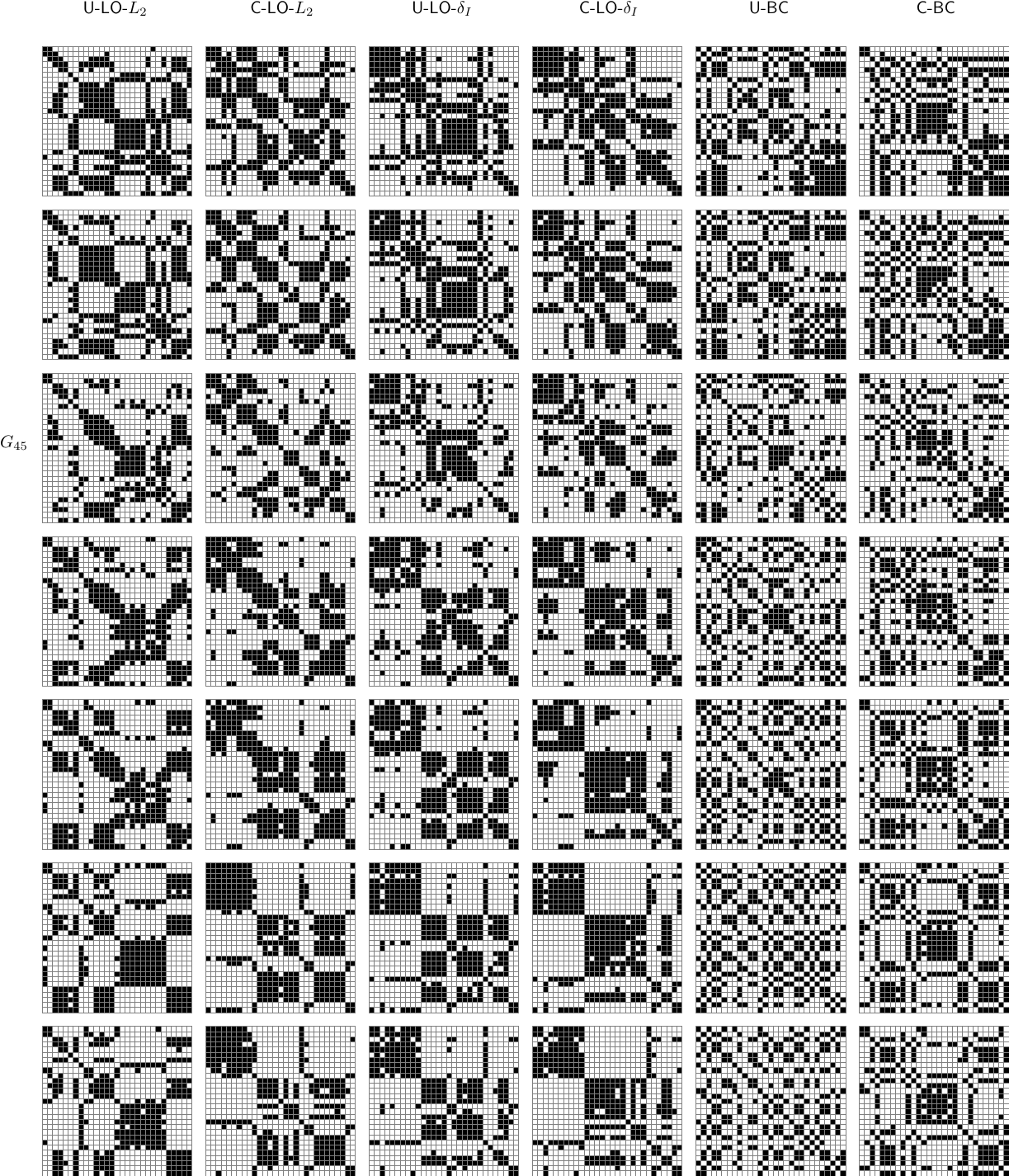}
    \caption{Matrix visualizations of graphs $G_{43}$ to $G_{49}$ of the FLT dataset, produced by the six algorithms considered in the main paper.}
    \label{fig:flt-results-43-49}
\end{figure*}

\clearpage

\begin{figure*}[t]
    \centering
    \includegraphics[page=8,width=\linewidth]{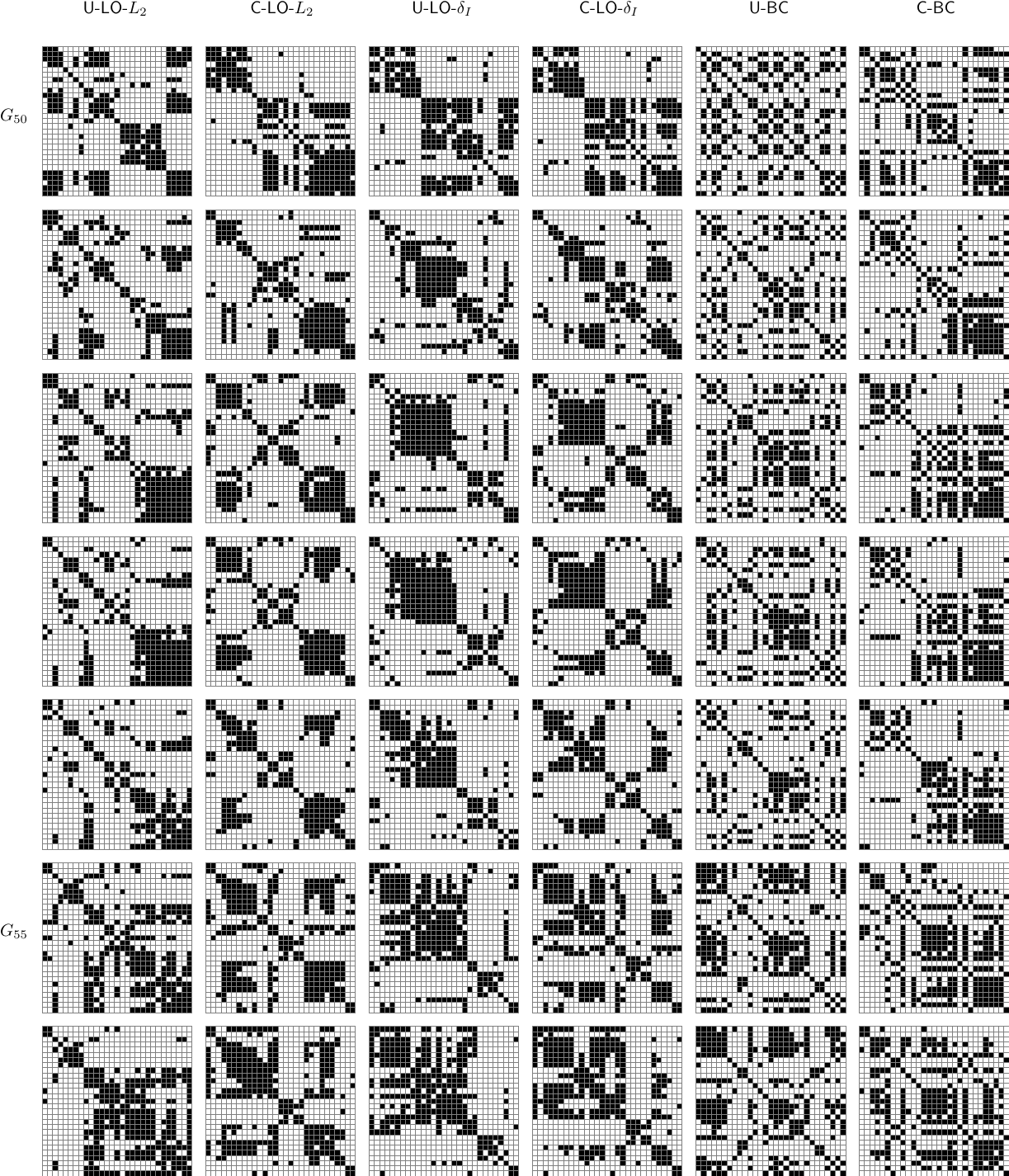}
    \caption{Matrix visualizations of graphs $G_{50}$ to $G_{56}$ of the FLT dataset, produced by the six algorithms considered in the main paper.}
    \label{fig:flt-results-50-56}
\end{figure*}

\clearpage

\begin{figure*}[t]
    \centering
    \includegraphics[page=9,width=\linewidth]{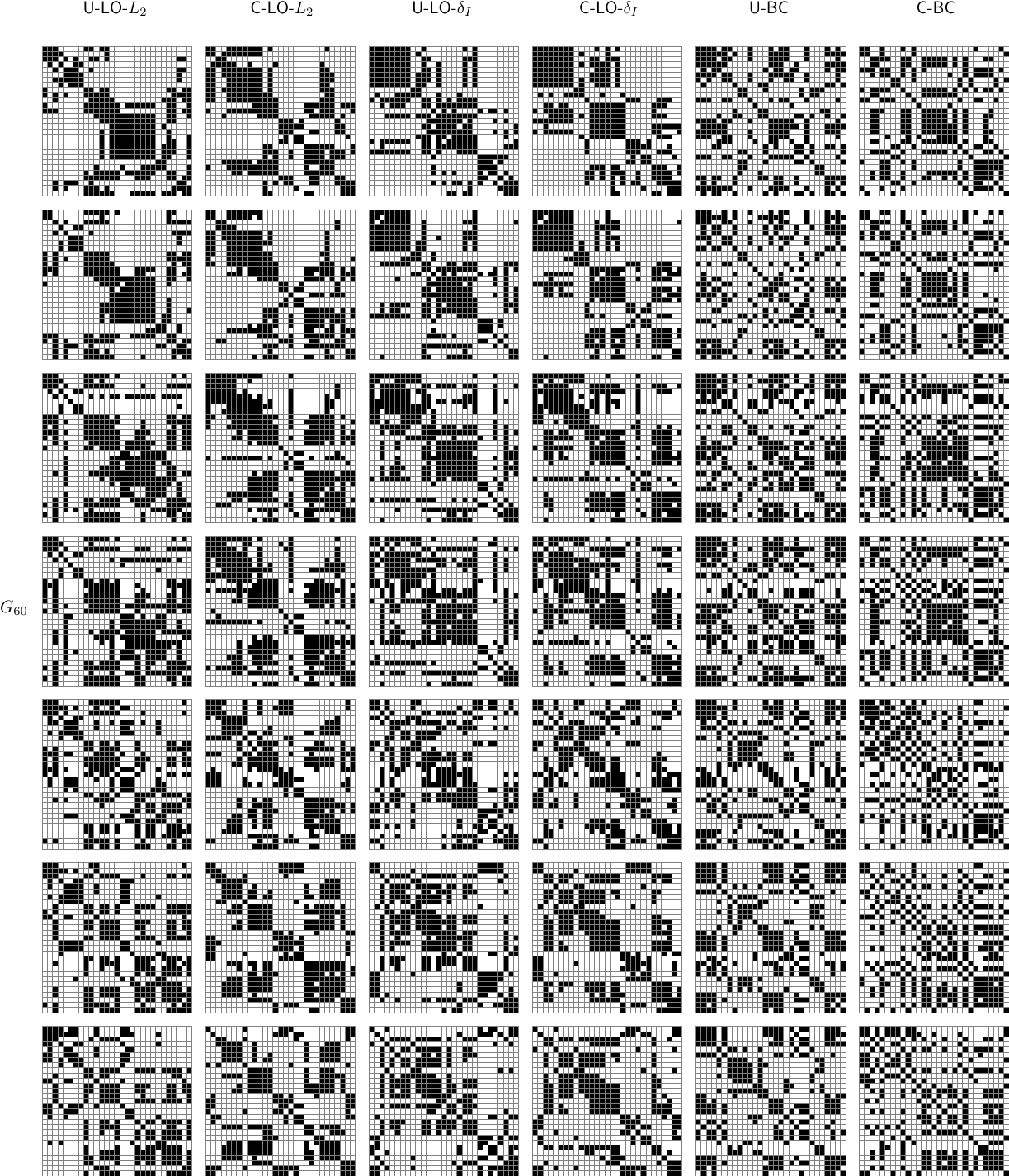}
    \caption{Matrix visualizations of graphs $G_{57}$ to $G_{63}$ of the FLT dataset, produced by the six algorithms considered in the main paper.}
    \label{fig:flt-results-57-63}
\end{figure*}

\clearpage

\begin{figure*}[t]
    \centering
    \includegraphics[page=10,width=\linewidth]{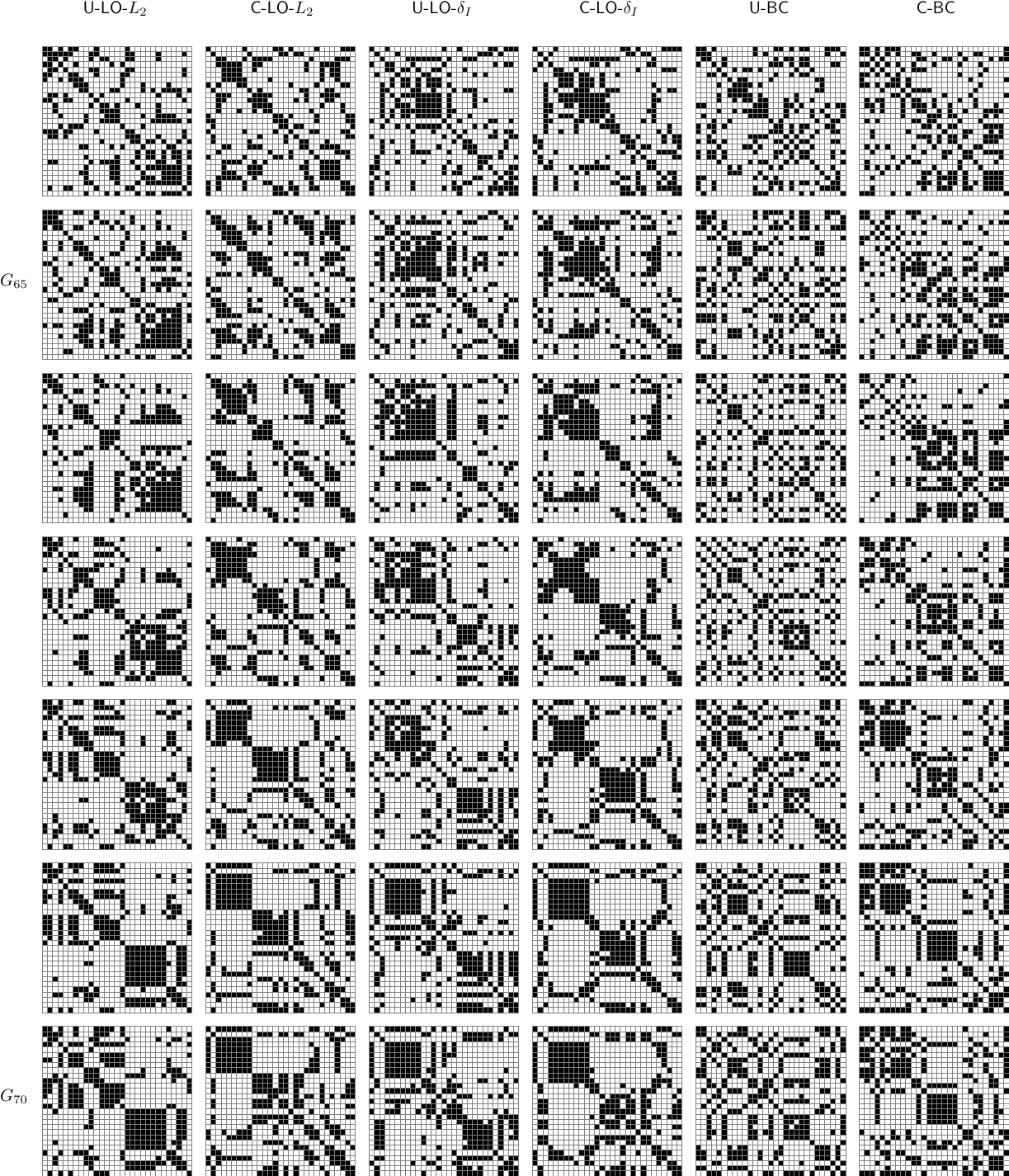}
    \caption{Matrix visualizations of graphs $G_{64}$ to $G_{70}$ of the FLT dataset, produced by the six algorithms considered in the main paper.}
    \label{fig:flt-results-64-70}
\end{figure*}

\clearpage

\begin{figure*}[t]
    \centering
    \includegraphics[page=11,width=\linewidth]{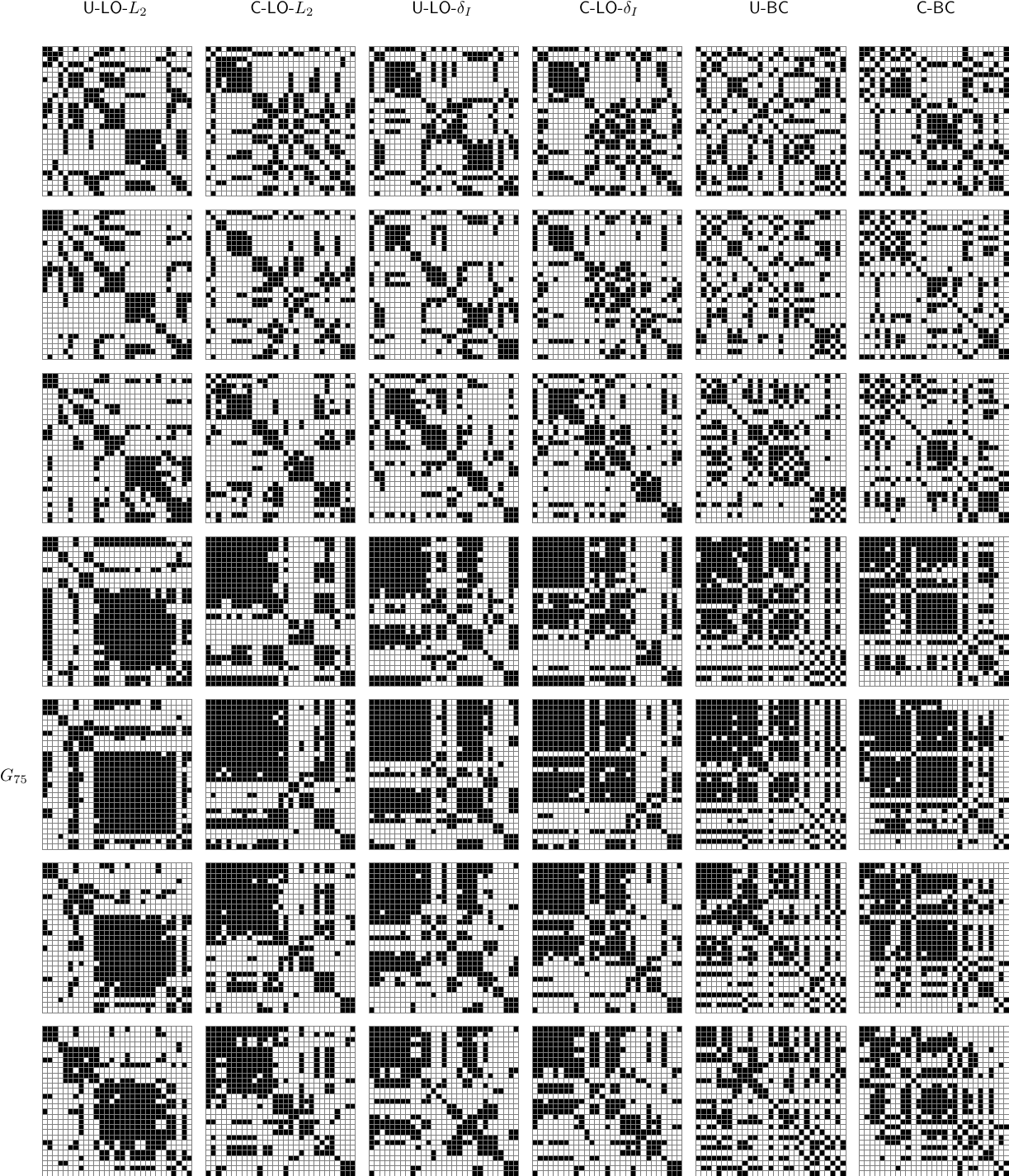}
    \caption{Matrix visualizations of graphs $G_{71}$ to $G_{77}$ of the FLT dataset, produced by the six algorithms considered in the main paper.}
    \label{fig:flt-results-71-77}
\end{figure*}

\clearpage

\begin{figure*}[t]
    \centering
    \includegraphics[page=12,width=\linewidth]{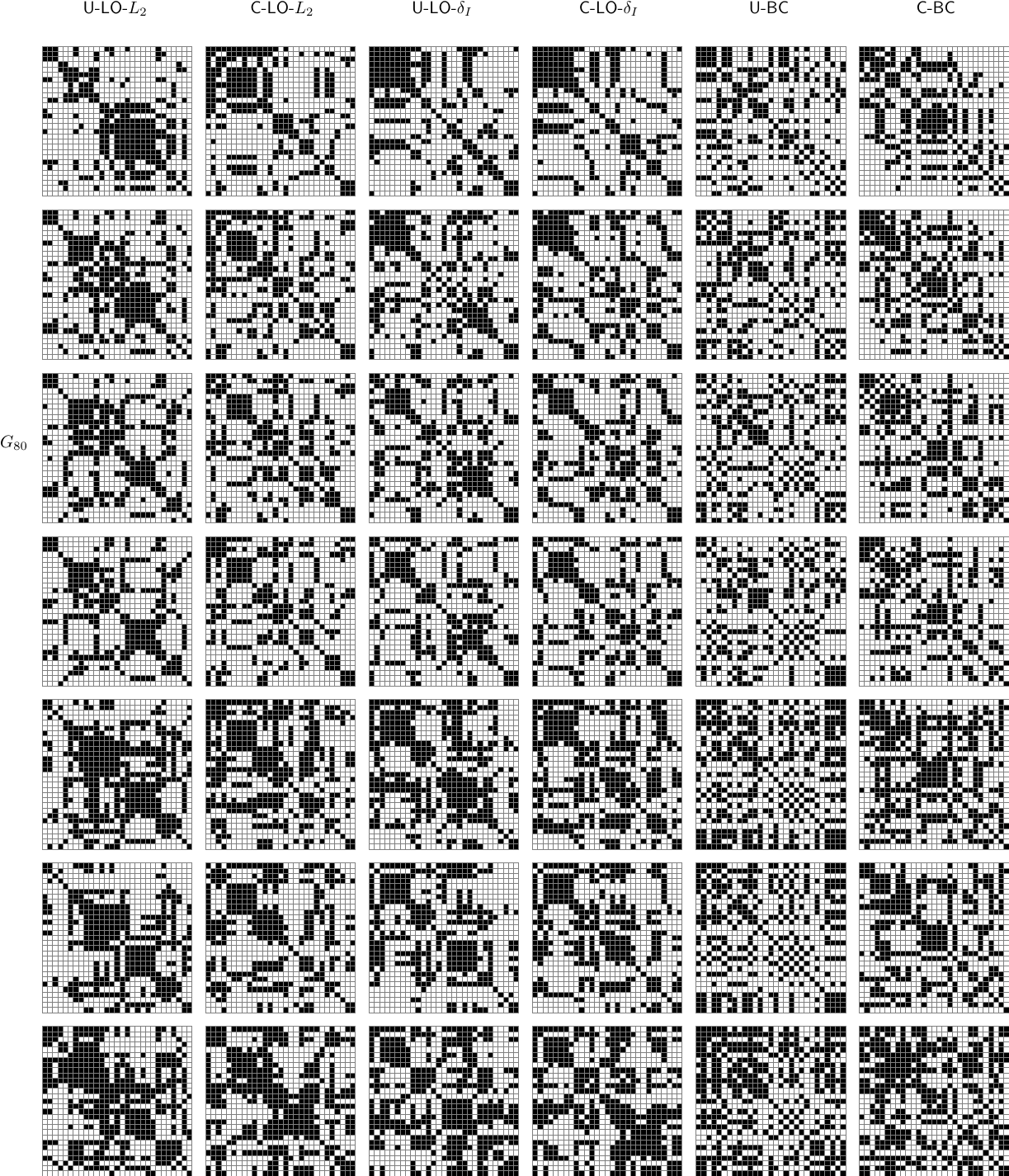}
    \caption{Matrix visualizations of graphs $G_{78}$ to $G_{84}$ of the FLT dataset, produced by the six algorithms considered in the main paper.}
    \label{fig:flt-results-78-84}
\end{figure*}

\clearpage

\begin{figure*}[t]
    \centering
    \includegraphics[page=13,width=\linewidth]{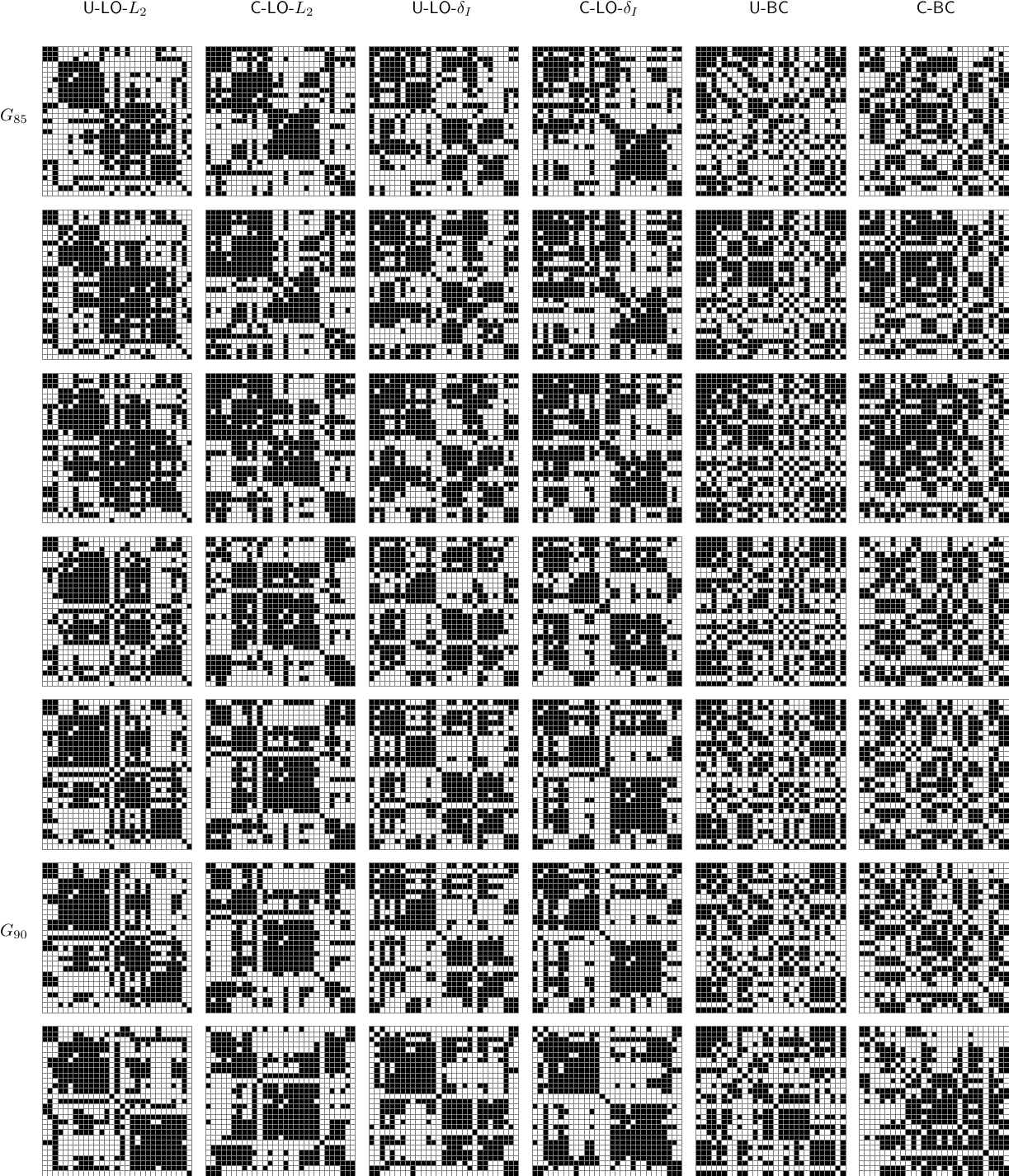}
    \caption{Matrix visualizations of graphs $G_{85}$ to $G_{91}$ of the FLT dataset, produced by the six algorithms considered in the main paper.}
    \label{fig:flt-results-85-91}
\end{figure*}

\clearpage

\begin{figure*}[t]
    \centering
    \includegraphics[page=14,width=\linewidth]{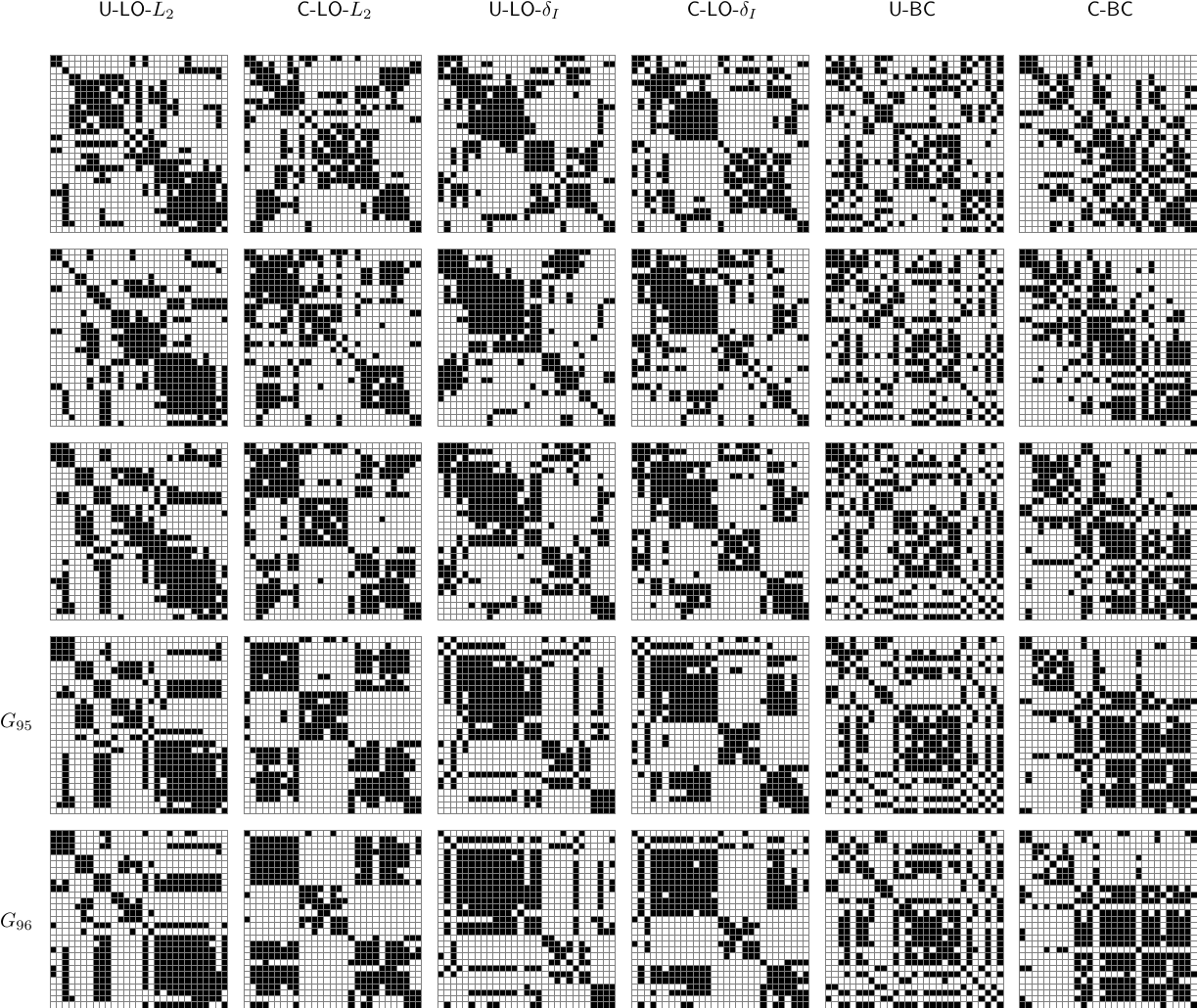}
    \caption{Matrix visualizations of graphs $G_{92}$ to $G_{95}$ of the FLT dataset, produced by the six algorithms considered in the main paper.}
    \label{fig:flt-results-92-96}
\end{figure*}

\end{document}